\documentclass[english, thm-restate]{lipics-v2021}
\usepackage[utf8]{inputenc}
\usepackage{amsmath}
\usepackage{amsfonts}
\usepackage{amsthm}
\usepackage{wrapfig}
\usepackage{makecell}
\usepackage{xspace}
\usepackage{caption}
\usepackage{imakeidx}
\usepackage{thmtools}
\usepackage{makecell}
\usepackage{todonotes}
\usepackage{algorithm}
\usepackage{algpseudocode}
\usepackage{booktabs}
\usepackage{ragged2e}
\usepackage{xspace}
\usepackage{adjustbox}
\usepackage{changepage}
\usepackage{tcolorbox}
\usepackage{graphicx,caption}
\usepackage{eucal}

\hideLIPIcs

\usepackage{tikz}
\usepackage{pgfplots}
\pgfplotsset{compat=newest}

\usepgfplotslibrary{groupplots}
\pgfplotsset{every axis/.style={scale only axis}}

\pgfplotsset{
  major grid style={thin,dotted},
  minor grid style={thin,dotted},
  ymajorgrids,
  yminorgrids,
  every axis/.append style={
    line width=0.7pt,
    tick style={
      line cap=round,
      thin,
      major tick length=4pt,
      minor tick length=2pt,
    },
  },
  legend cell align=left,
  legend style={
    line width=0.7pt,
    /tikz/every even column/.append style={column sep=3mm,black},
    /tikz/every odd column/.append style={black},
  },
  legend style={font=\small},
  title style={yshift=-2pt},
  enlarge x limits=0.15,
  every axis y label/.append style={yshift=-1ex},
  /pgf/number format/1000 sep={},
  axis lines*=left,
  xlabel near ticks,
  ylabel near ticks,
  label style={font=\footnotesize},       
  tick label style={font=\footnotesize},
}

\author{Ivor van der Hoog}{Technical University of Denmark, Denmark}{idjva@dtu.dk}{
https://orcid.org/0009-0006-2624-0231}{This project has received funding from the European Union's Horizon 2020 research and innovation programme under the Marie Sk\l{}odowska-Curie grant agreement No 899987}

\author{Lara Ost}{University of Vienna, Austria}{lara.ost@univie.ac.at}{
https://orcid.org/0000-0003-4311-9928}{Funded by the Vienna Graduate School on Computational Optimization (VGSCO), FWF-Project No.~W1260-N35}

\author{Eva Rotenberg}{Technical University of Denmark, Denmark}{erot@dtu.dk}{
https://orcid.org/0000-0001-5853-7909}{}

\author{Daniel Rutschmann}{Technical University of Denmark, Denmark}{daru@dtu.dk}{
https://orcid.org/0009-0005-6838-2628}{}

\funding{{\it Eva Rotenberg} and {\it Daniel Rutschmann} are grateful to the Carlsberg Foundation for supporting this research via Eva Rotenberg's Young Researcher Fellowship CF21-0302 ``Graph Algorithms with Geometric Applications''.}

\supplementdetails[]{Software}{www.github.com/laraost/DiscreteSubtrajectoryClustering}
\supplementdetails[]{Test suite}{https://github.com/dacin21/discrete-subtrajectory-clustering-test-suite}

\begin{document}

\newcommand{\fancy}[1]{\EuScript{#1}}

\title{\Large Efficient Greedy Discrete Subtrajectory Clustering}
\titlerunning{Efficient Greedy Discrete Subtrajectory Clustering}

\newcommand{\frechet}{Fr\'{e}chet\xspace}
\newcommand{\Cov}{\texttt{\small cov}}
\newcommand{\fd}{\mathcal{D}_F}
\newcommand{\means}{\textnormal{\small mean}}
\newcommand{\centercost}{{\textnormal{\small cen}}}
\newcommand{\minrow}{\texttt{\small row}}
\newcommand{\maxcol}{\texttt{\small col}}
\newcommand{\pathlet}{cluster\xspace}
\newcommand{\pathlets}{clusters\xspace}
\newcommand{\cP}{{\fancy{P}}}
\newcommand{\cT}{\mathcal{T}}
\newcommand{\ra}[1]{\renewcommand{\arraystretch}{#1}}

\newtheorem{algodef}{Clustering}

\nolinenumbers

\bibliographystyle{plainurl}

\authorrunning{I.  van der Hoog, L. Ost, E. Rotenberg, and D. Rutschmann}


\date{}

\maketitle

\begin{abstract}
We cluster a set of trajectories $\mathcal{T}$ using subtrajectories of $\mathcal{T}$. We require for a clustering $C$ that any two subtrajectories $(\mathcal{T}[a, b], \mathcal{T}[c, d])$ in a cluster have disjoint intervals $[a,b]$ and $[c, d]$.
Clustering quality may be measured by the number of clusters, the number of vertices of $\mathcal{T}$ that are absent from the clustering, and by the Fr\'{e}chet distance between subtrajectories in a cluster. 

 A $\Delta$-cluster of $\mathcal{T}$ is a cluster ${\mathcal{P}}$ of subtrajectories  of $\mathcal{T}$ with a centre $P \in {\mathcal{P}}$, where all subtrajectories in ${\mathcal{P}}$ have Fr\'{e}chet distance at most $\Delta$ to $P$. Buchin, Buchin, Gudmundsson, L\"{o}ffler and Luo  present two $O(n^2 + n m \ell)$-time algorithms: \textbf{SC($\max$, $\ell$, $\Delta$, $\mathcal{T}$)} computes a \emph{single} $\Delta$-cluster where
 $P$ has at least $\ell$ vertices and maximises the cardinality $m$ of ${\mathcal{P}}$.
 \textbf{SC($m$, $\max$, $\Delta$, $\mathcal{T}$)} computes a \emph{single} $\Delta$-cluster where  ${\mathcal{P}}$ has cardinality $m$ and maximises the complexity $\ell$ of $P$. 
 
 In this paper, which is a mixture of algorithms engineering and theoretical insights, we use such maximum-cardinality clusters in a greedy clustering algorithm. 
We first provide an efficient implementation of \textbf{SC($\max$, $\ell$, $\Delta$, $\mathcal{T}$)} and \textbf{SC($m$, $\max$, $\Delta$, $\mathcal{T}$)} that significantly outperforms previous implementations. 
Next, we use these functions as a subroutine in a greedy clustering algorithm, which performs well when compared to existing subtrajectory clustering algorithms on real-world data.  
Finally, we observe that, for fixed $\Delta$ and $\mathcal{T}$, these two functions always output a point on the Pareto front of some bivariate function $\theta(\ell, m)$.
We design a new algorithm \textbf{PSC($\Delta$, $\mathcal{T}$)} that in $O( n^2 \log^4 n)$ time computes a $2$-approximation of this Pareto front. This yields a broader set of candidate clusters, with comparable quality to the output of the previous functions.
We show that using \textbf{PSC($\Delta$, $\mathcal{T}$)} as a subroutine improves the clustering quality and performance even further.
\end{abstract}

\clearpage 

\newpage

\section{Introduction}

Trajectories describe the movement of objects.
The moving objects are modeled as a sequence of locations.
Trajectory data often comes from GPS samples, which are being generated on an incredible scale through cars, cellphones, and other trackers.  
We focus on \emph{subtrajectory clustering} where clusters are sets of subtrajectories. The goal is to create few clusters and to ensure that subtrajectories within a cluster have low pairwise discrete \frechet distance. This algorithmic problem has been studied extensively~\cite{agarwal2018subtrajectory, bruning2022faster, buchin2020improved, buchin2011detecting,  buchin2013median, chen2013pathlet, conradi2023finding, gaffney1999trajectory, gudmundsson2012gpu, gudmundsson2022cubic, gudmundsson2012motifs, hung2015clustering, lee2007trajectory, sung2012trajectory} in both theoretical and applied settings. 
For surveys, see~\cite{bian2018survey} and \cite{machado2024survey}.

\subparagraph{Subtrajectory clustering. } 
Comparing whole trajectories gives little information about shared structures. Consider two trajectories $T_1$ and $T_2$ that represent individuals who drive the same route to work, but live in different locations. 
For almost all distance metrics, the distance between $T_1$ and $T_2$ is large because their endpoints are (relatively) far apart. However, $T_1$ and $T_2$ share similar subtrajectories. 
Subtrajectory clustering allows us to model a wide range of behaviours: from detecting commuting patterns, to flocking behaviour, to congestion areas~\cite{buchin2011detecting}. 
In the special case where trajectories correspond to traffic on a known road network,
trajectories can be mapped to the road network. In many cases, the underlying road network might not be known (or it may be too dense to represent).
In such cases, it is interesting to construct a sparse representation of the underlying road network based on the trajectories.
The latter algorithmic problem is called \emph{map construction}, and it has also been studied extensively~\cite{
ahmed2015comparison, biagioni2012map, ahmed2012constructing,cao2009gps, davies2006scalable,karagiorgou2012vehicle, karagiorgou2013segmentation}.
We note that subtrajectory clustering algorithms may be used for map construction: if each cluster is assigned one centre, then the collection of centres can be used to construct an underlying road network~\cite{buchin2017clustering, buchin2020improved}.

\subparagraph{Related work. }
BBGLL~\cite{buchin2011detecting} introduce the following problem. 
Rather than clustering $\mathcal{T}$, their goal is to compute a \emph{single} cluster. Given $\mathcal{T}$,  a $\Delta$-\emph{\pathlet} is any pair $(P, {\fancy{P}})$ where:

\begin{itemize}
    \item $P$ is a subtrajectory of a trajectory in $\mathcal{T}$,
    \item ${\fancy{P}}$ is a set of subtrajectories, where for all $\mathcal{T}[a, b], \mathcal{T}[c, d] \in {\fancy{P}}$, $[a, b] \cap [c, d] = \emptyset$,
    \item and, for all $P' \in {\fancy{P}}$, the \frechet distance between $P$ and $P'$ is at most $\Delta$. 
\end{itemize}

See Figure~\ref{fig:pathlet}.
They define two functions that each output a single $\Delta$-\pathlet $(P, {\fancy{P}})$:

\begin{itemize}
    \item \textbf{SC($\max$, $\ell$, $\Delta$, $\mathcal{T}$)} requires that $|P| = \ell$ and maximises $|{\fancy{P}}|$.
    \item \textbf{SC($m$, $\max$, $\Delta$, $\mathcal{T}$)}  requires  that $|{\fancy{P}}| = m$ and maximises $|P|$. 
\end{itemize}

  If $|P| = \ell$ and $|{\fancy{P}}| = m$, their algorithms use $O(n^2 + n m \ell)$ time and $O(n \ell)$ space.

AFMNPT~\cite{agarwal2018subtrajectory} propose a clustering framework.
They define a clustering $C$ as any set of \pathlets.
They propose a scoring function with constants $(c_1, c_2, c_3)$ that respectively weigh: the number of \pathlets, the radius of each \pathlet, and the number of uncovered vertices. 
Their implement a heuristic algorithm that uses this scoring function.

BBDFJSSSW~\cite{buchin2017clustering} do map construction.
They compute a set of $\Delta$-\pathlets $C$ for various $\Delta$. 
For each $(P, {\fancy{P}}) \in C$, they use $P$ as a road on the map they are constructing. 
Thus, it may be argued that BBDFJSSSW~\cite{buchin2017clustering} do subtrajectory clustering. 
As a preprocessing step, they implement \textbf{SC($\max$, $\ell$, $\Delta$, $\mathcal{T}$)}. By trying various choices for $\ell$ and $\Delta$, they compute a collection $S$ of `candidate' \pathlets. 
On a high level, they greedily select a \pathlet $(P, {\fancy{P}}) \in S$, add it to the clustering, and remove all subtrajectories in ${\fancy{P}}$ from all remaining \pathlets in $S$. For a \pathlet $(Q, {\fancy{Q}}) \in S$, this may split a $Q' \in {\fancy{Q}}$ into two subtrajectories. BBGHSSSSSW~\cite{buchin2020improved} continue this line of work with improved implementations. 

ABCD~\cite{bruning_subtrajectory_2023} do not require that for all $\mathcal{T}[a, b], \mathcal{T}[c, d] \in {\fancy{P}}$, $[a, b] \cap [c, d] = \emptyset$. They do require that the final clustering $C$ has no uncovered vertices, and that for each \pathlet $(P, {\fancy{P}})$, $|P| \leq \ell$ for some given $\ell$. For any $\Delta$, denote by $k_\Delta$ the minimum-sized clustering that meets their requirements. Their randomised polynomial-time algorithm computes  a clustering of $O(\Delta)$-\pathlets of size $\tilde{O}(k_\Delta \ell)$.
BCD~\cite{bruning2022faster} give polynomial-time clustering algorithm of $O(\Delta)$-\pathlets of size $\tilde{O}(k_\Delta)$.  
Conradi and Driemel~\cite{conradi2023finding} implement an $\tilde{O}(\ell^2 n^4 + k_\Delta \ell n^4)$-time algorithm that computes a set of $O(\Delta)$-\pathlets of size $O(k_\Delta \log n)$. 

\subparagraph{Contribution.} Our contribution is threefold.
First, in Section~\ref{sec:clustering}, we describe an efficient implementation of \textbf{SC($\max$, $\ell$, $\Delta$, $\mathcal{T}$)} and \textbf{SC($m$, $\max$, $\Delta$, $\mathcal{T}$)}.
We compare our implementation to all existing single-core implementations. 

In Section~\ref{sec:clustering} we propose that \textbf{SC($\max$, $\ell$, $\Delta$, $\mathcal{T}$)} produces high-quality \pathlets across all clustering metrics. 
We create a greedy clustering algorithm that fixes $(\ell, \Delta)$ and iteratively adds the solution $(P, {\fancy{P}}) =$ \textbf{SC($\max$, $\ell$, $\Delta$, $\mathcal{T}$)} to the clustering (removing all vertices in ${\fancy{P}}$ from $\mathcal{T}$). An analogous approach for \textbf{SC($m$, $\max$, $\Delta$, $\mathcal{T}$)} gives two clustering algorithms. 

In Section~\ref{sec:pareto} we note that the functions \textbf{SC($\max$, $\ell$, $\Delta$, $\mathcal{T}$)} and \textbf{SC($m$, $\max$, $\Delta$, $\mathcal{T}$)} search a highly restricted solution space, having fixed either the minimum length of the centre of a \pathlet or the minimum cardinality. 
We observe that these functions always output a \pathlet on the Pareto front of some bivariate function $\theta(m, \ell)$. 
Since points on this Pareto front give high-quality \pathlets, we argue that it may be beneficial to approximate all vertices of this Pareto front directly -- sacrificing some \pathlet quality for \pathlet flexibility. 
We create a new algorithm \textbf{PSC($\Delta$, $\mathcal{T}$)} that computes in $O( n^2 \log^4 n)$ time and $O(n)$ space a set of \pathlets that is a $2$-approximation of the Pareto font of $\theta(m, \ell)$. 
We show that using \textbf{PSC($\Delta$, $\mathcal{T}$)}  as a  subroutine for greedy clustering gives even better results. 

In Section~\ref{sec:experiments} we do experimental analysis. 
We compare our clustering algorithms to AFMNPT~\cite{agarwal2018subtrajectory} and BBGHSSSSSW~\cite{buchin2020improved}.
We adapt the algorithm in~\cite{buchin2020improved} so that it terminates immediately after computing the subtrajectory clustering. 
We do not compare to Conradi and Driemel~\cite{conradi2023finding} as their setting is so different that it produces incomparable \pathlets. 
We compare the algorithms based on their overall run time and space usage. 
We evaluate the clustering quality along various metrics such as the number of \pathlets, the average cardinality of each \pathlet, and the \frechet distance between subtrajectories in a \pathlet. 
Finally, we consider the scoring function by AFMNPT~\cite{agarwal2018subtrajectory}. Our repository can be found here~\cite{respository}.

 \begin{figure}[t]
        \centering
        \includegraphics[]{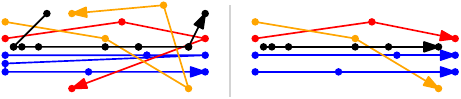}
        \caption{ (left) Four trajectories. (right) Five subtrajectories ${\fancy{P}}$. Depending on the choice of $P \in \cP$, we may get five different \pathlets $(P, {\fancy{P}})$. \label{fig:pathlet}} 
    \end{figure}

  \begin{figure}[h]
        \centering
        \includegraphics[]{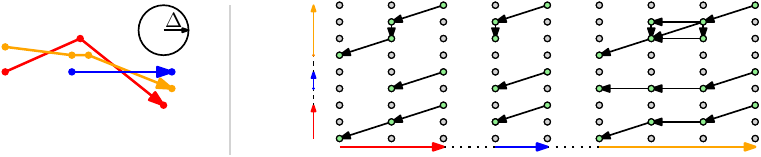}
        \caption{ A set of trajectories $\mathcal{T}$ and the matrix $M_\Delta(\mathcal{T}, \mathcal{T})$ as a graph. \label{fig:matrix}}
    \end{figure}

\newpage
\section{Preliminaries}

A \emph{trajectory} $T$ in $\mathbb{R}^2$ is an ordered sequence of vertices.  
The algorithmic input is a set $\mathcal{T}$ of trajectories where $n = \sum_{T \in \mathcal{T}} |T|$. 
We denote by $\mathcal{X}$ the set of all vertices in $T \in \mathcal{T}$. 
We view any trajectory $T$ with $n_T$ vertices as a map from $[n_T]$ to $\mathbb{R}^2$ where $T(i)$ is the $i$'th vertex on $T$. 
For brevity, we refer to $\mathcal{T}$ as a single trajectory obtained by concatenating all trajectories in $\mathcal{T}$ arbitrarily.
 A \emph{boundary edge} in $\mathcal{T}$ is any edge in $\mathcal{T}$ between the endpoints of two consecutive trajectories.
A \emph{subtrajectory} $\mathcal{T}[a, b]$ of $\mathcal{T}$ is defined by $a, b \in [n]$ with $a \leq b$, where we require that the subcurve from $\mathcal{T}(a)$ to $\mathcal{T}(b)$ contains no boundary edge.

\subparagraph{\frechet distance.}
We first define discrete walks for pair of trajectories $P = (p_1, \ldots, p_x)$ and $T = (t_1, \ldots, t_y)$ in $\mathbb{R}^2$.
We say that an ordered sequence $F$ of points in $[x] \times [y]$ is a \emph{discrete walk} if consecutive pair $(i, j), (s, l) \in F$ have $s\in \{i,i-1\}$ and $l\in\{j,j-1\}$.
The \emph{cost} of a discrete walk $F$ is the maximum distance $d(P(i), Q(j))$ over $(i, j) \in F$.
The discrete \frechet distance is the minimum-cost discrete walk from $(x, y)$ to $(1, 1)$:
\[
\fd(P, Q) := \min_{F }  \textnormal{cost}(F) = \min_{F } \max_{(i, j) \in F} d(P(i), Q(j)).
\]

\subparagraph{Free-space matrix.}
Given $\mathcal{T}$ and  a value $\Delta$, the \emph{Free-space matrix} $M_\Delta(\mathcal{T}, \mathcal{T})$ is an $n \times n$ $01$-matrix where for all $(i, j) \in [n] \times [n]$ the matrix has a zero at position $i$, $j$ if $d(\mathcal{T}(i), \mathcal{T}(j)) \leq \Delta$.  We use matrix notation where $(i, j)$ denotes row $i$ and column $j$.
In the plane, $(0, 0)$ denotes the top left corner.
We also use $M_\Delta(\mathcal{T}, \mathcal{T})$ to denote the digraph where the vertices are all $(i, j) \in [n] \times [n]$ and there exists an edge from $(i, j)$ to $(s, k)$ whenever (see Figure~\ref{fig:matrix}):
\begin{itemize}
    \item $s \in \{ i, i -1 \}$ and $k \in \{ j, j - 1 \}$, and
    \item  $M_\Delta(\mathcal{T}, \mathcal{T})[i, j] = M_\Delta(\mathcal{T}, \mathcal{T})[s, k] = 0$.
\end{itemize}

\noindent
Alt and Godeau~\cite{alt1995computing} show that for any two subtrajectories $\mathcal{T}[a, b]$  and $\mathcal{T}[c, d]$: $\fd(T[a, b], T[c, d]) \leq \Delta$ if and only if there exists a directed path from $(b, d)$ to $(a, c)$ in this graph. 

\begin{definition}
\label{def:minrow}
    Given $\Delta$ and $\mathcal{T}$, we say that a vertex $(b, x)$ in $M_\Delta(\mathcal{T}, \mathcal{T}$) can \emph{reach} row $a$ if there exists an integer $x' \leq x$ such that there is a directed path from $(b, x)$ to $(a, x')$. We denote by $\minrow(b, x)$ the minimum integer $a'$ such that $(b, x)$ can reach row $a'$. 
    \end{definition}

\begin{definition}[Figure~\ref{fig:pathlet}]
    A \pathlet  is any pair $(P, {\fancy{P}})$ where $P$ is a subtrajectory and ${\fancy{P}}$ is a set subtrajectories with $P \in {\fancy{P}}$. We require that for all $\mathcal{T}[a, b], \mathcal{T}[c, d] \in \mathcal{T}$, $[a, b] \cap [c, d] = \emptyset$.     
    We define its \emph{length} $|P|$, \emph{cardinality} $|{\fancy{P}}|$, and \emph{coverage} $\Cov(P, {\fancy{P}}) =  \bigcup_{S \in {\fancy{P}}}$ S. 
\end{definition}

\noindent
Note that we follow~\cite{agarwal2018subtrajectory, buchin2017clustering, buchin2020improved, buchin2011detecting} and consider only \emph{disjoint} clusterings, meaning $\forall (P, {\fancy{P}}') \in C$, $\forall \mathcal{T}[a, b], \mathcal{T}[c, d] \in {\fancy{P}}$, the intervals $[a, b]$ and $[c, d]$ are disjoint.
For brevity, we say that a $\Delta$-\pathlet is any \pathlet $(P, {\fancy{P}})$ where $\forall P' \in {\fancy{P}} : \fd(P, P') \leq \Delta$.

\begin{definition}
A \emph{clustering} $C$ is any set of \pathlets. 
Its \emph{coverage} is $\Cov(C) := \bigcup\limits_{(P, {\fancy{P}}) \in C} \Cov(P, {\fancy{P}})$. 
\end{definition}

\noindent
AFMNPT~\cite{agarwal2018subtrajectory} propose a cost function to determine the quality of a clustering $C$.
We distinguish between \emph{$k$-centre} and \emph{$k$-means} clustering: 

\begin{definition}
\label{def:score}
Given constants $(c_1, c_2, c_3)$ , 
$\texttt{Score}(C)$ is a sum that weights three terms:
    \begin{enumerate}
    \item The number of \pathlets in $C$: $c_1 |C|$.
\item The chosen cost function for the clustering. I.e.,
   
    $c_2 \max\limits_{(P, {\fancy{P}}) \in C} \centercost(P, {\fancy{P}}) = c_2 \max\limits_{(P, {\fancy{P}}) \in C} \max\limits_{P' \in {\fancy{P}}} \fd(P, P')$ when doing $k$-centre clustering, or 
    
         $c_2 \sum\limits_{(P, {\fancy{P}}) \in C} \means(P, {\fancy{P}})  = c_2 \sum\limits_{(P, {\fancy{P}}) \in C} \sum\limits_{P' \in {\fancy{P}}} \fd(P, P')$ when doing $k$-means clustering. 
\item The fraction of uncovered points: $c_3 \frac{| \mathcal{X} \backslash \Cov(C) |}{|\mathcal{X}|}$.
\end{enumerate}
\end{definition}

\noindent
We also consider unweighted  metrics:
 the maximal and average \frechet distance between subtrajectories in a \pathlet, and
the maximal and average cardinality of the \pathlets.

\section{Data}

We use five real-world data sets (see Table~\ref{tab:data_sets}). For some of these data sets we include a subsampled version that only retains a $\frac{1}{c}$ fraction of the trajectories.

\begin{itemize}
    \item The first three data sets are Athens-small, Chicago and Berlin, consisting of GPS samples from car traffic across these cities.  These data sets are widely used
benchmark tests for road map construction~\cite{ahmed2015comparison, buchin2017clustering, wang2015efficient, huang2018automatic}. 
\item 
The fourth data set is the \emph{Drifter} data set that was used for subtrajectory clustering by Driemel and Conradi~\cite{conradi2023finding}. It consists of GPS data from ocean drifters.
\item The fifth data set is the UnID data set, which is the smallest data set used in~\cite{9827305}. 
\end{itemize}

\noindent
Unfortunately, we cannot use any of the data sets in~\cite{agarwal2018subtrajectory}. Most of these are no longer available.  Those that are available are far too large. It is unknown how~\cite{agarwal2018subtrajectory} preprocessed these. 

\subparagraph{Synthetic data.}
All competitor algorithms do not terminate on the larger real-world data sets, even when the time limit is 24 hours. Thus, we primarily use these sets to  compare algorithmic efficiency. 
Our qualitative analysis uses the subsampled data sets and synthetic data. We generate data by creating a random domain $X$. 
Then, for some parameter $c$, we randomly select $c \%$ of the domain and invoke an approximate TSP solver. We add the result as a trajectory. This results in a collection of trajectories with many similar subtrajectories. 
Synthetic-A has a random domain of 200 vertices. Synthetic-B has a random domain of 100 vertices. For both of these, we consider $c \in \{50, 90, 95\}$.

\vspace{1cm}

\begin{table}[h]
    \centering \ra{1.1}
    \begin{tabular}{@{}llrrl@{}}
    \toprule
      \textbf{Name} & \textbf{Real world} & \textbf{$\#$ of trajectories} & \textbf{$\#$ of vertices} & \textbf{Subsampled}  \\
      \midrule
        Athens-small & yes & 128           & 2840 & no \\
        Chicago-4 & yes & 222         & 29344 & $\frac{1}{4}$'th  \\
        Chicago & yes & 888         & 118360 & no \\
        Berlin-10 & yes & 2717         & 19130 & $\frac{1}{10}$'th  \\
        Berlin & yes & 27188        & 192223 & no \\
        Drifter & yes & 2167       & 1931692 & no \\
        UniD & yes & 362        & 214077 & no \\
        Synthetic-A & no & 50 & 5000 - 9500 & no\\
        Synthetic-B & no & 100 &  5000 - 9500 & no \\
        \bottomrule
    \end{tabular}
    \caption{Our data sets, their sizes, and whether we subsampled a fraction of their trajectories.}
        \label{tab:data_sets}
    \label{tab:my_label}
\end{table}

\clearpage 

\section{\texorpdfstring{Introducing and solving SC($m$, $\ell$, $\Delta$, $\mathcal{T}$)}{Introducing and solving SC(m, l, D, T)} }
\label{main:BBGLL}

BBGLL~\cite{buchin2011detecting} introduce the problem \textbf{SC($m$, $\ell$, $\Delta$, $\mathcal{T}$)}.
The input are integers $m$, $\ell$, $\Delta$ and a set of trajectories $\mathcal{T}$.  The goal is to find a $\Delta$-\pathlet $(P, {\fancy{P}})$ where $|{\fancy{P}}| \geq m$ and $|P| = \ell$.  
The problem statement may be slightly altered by maximising $m$ or $\ell$ instead. This yields the problems \textbf{SC($\max$, $\ell$, $\Delta$, $\mathcal{T}$)} and \textbf{SC($m$, $\max$, $\Delta$, $\mathcal{T}$)}.
There exist three single-core implementations of these algorithms (the \emph{MOVETK library}~\cite{movetk},  \emph{map-construct} \cite{buchin2017clustering} and its improvement: \emph{map-construct-rtrees}~\cite{buchin2020improved} a version of the map-construct algorithm that uses r-trees). 
\cite{movetk} uses $O(n^2)$ space and $O(n^2 + n m \ell)$ time.
Both \cite{buchin2017clustering, buchin2020improved} use $O(n \ell)$ space and $O(n^2 + n m \ell)$ time, and~\cite{buchin2020improved} offers running time improvements over~\cite{buchin2020improved}. 
We note that the latter implement a semi-continuous Fréchet distance variant instead of the discrete Fréchet distance. 
On the experimental input, switching between these metrics does not change the algorithmic output. 
We do not compare to the existing GPU-based implementation~\cite{gudmundsson2012gpu}.

  \begin{figure}[b]
        \centering
        \includegraphics[]{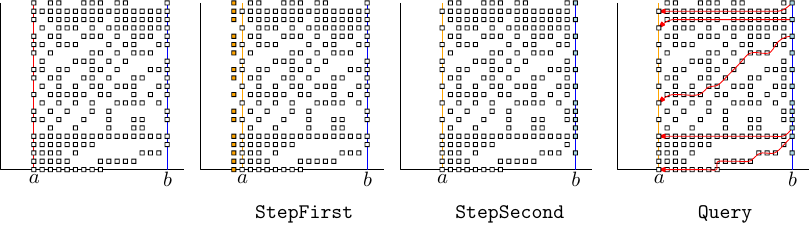}
        \caption{ An overview of the algorithm in~\cite{buchin2011detecting}. We illustrate zeroes in $M_\Delta(\mathcal{T}, \mathcal{T})$ with a square. \vspace{-0.5cm}}
        \label{fig:BBGLL}
    \end{figure}

\subparagraph{Contribution.} We sketch how we optimise the BBGLL algorithm in Appendix~\ref{sec:BBGLL}. 
Considers two integers $a, b$ with $\ell = b - a$ (Figure~\ref{fig:BBGLL}).
It alternates between
\texttt{StepFirst} which increments $a$, \texttt{StepSecond} which increments $b$, and $\texttt{Query}$ that given $(\Delta, \ell =  b - a, m)$ aims to find a $\Delta$-\pathlet $(P, {\fancy{P}})$ with $P = \mathcal{T}[a, b]$ and $|{\fancy{P}}| = m$.
We apply two ideas:

\begin{enumerate}
    \item \textbf{Windowing.} By~\cite{buchin2011detecting}, it suffices to keep only $M_\Delta(\mathcal{T}[a, b], \mathcal{T})$ in memory. Furthermore, one may store only the $z' \in O(n \ell)$ zeroes in this sub-matrix. This is implemented in~\cite{buchin2017clustering, buchin2020improved}. By cleverly abandoning memory, our \texttt{StepFirst} deletes a row in constant time. The \texttt{StepSecond} function is then the bottleneck, as it adds a new row.
    \item \textbf{Row generation.} In~\cite{buchin2011detecting}, the  \texttt{StepSecond} and \texttt{Query} functions iterate over all $n$ cells in the bottommost row of $M_\Delta(\mathcal{T}, \mathcal{T})$. By applying a range searching data structure, we instead only iterate over all zeroes in this row.
\end{enumerate}

\noindent
 We note that map-construct-rtree~\cite{buchin2020improved} uses r-trees for row generation (we refer to the journal version~\cite{buchin2025roadster}). 
 However, we believe that the size of their r-tree implementation dominates the time and space used by their algorithm in our experiments. In Appendix~\ref{sec:BBGLL}, we describe how to use a light-weight range tree to implement this principle.

\subparagraph{Experiments.}
We compare the implementations across our data sets for a variety of choices of $\Delta$ and $m$. We use the implementations of \textbf{SC}($m$, $\max$, $\Delta$, $\mathcal{T}$) so that our experiments have to fix fewer parameters. 
Figures~\ref{fig:bar_bbgll_time} and~\ref{fig:bar_bbgll_memory} shows the time and space usage across a subset of our experiments.
The data clearly shows that our implementation is significantly more efficient than the previous single-core implementations.
It is often more than a factor 1000 more efficient in terms of runtime and memory usage.
We do not provide a more extensive analysis of these results since this is not our main contribution.

    \begin{figure}[H]
    \hspace*{-1.45in}
    \includegraphics[width = 1.5\textwidth]{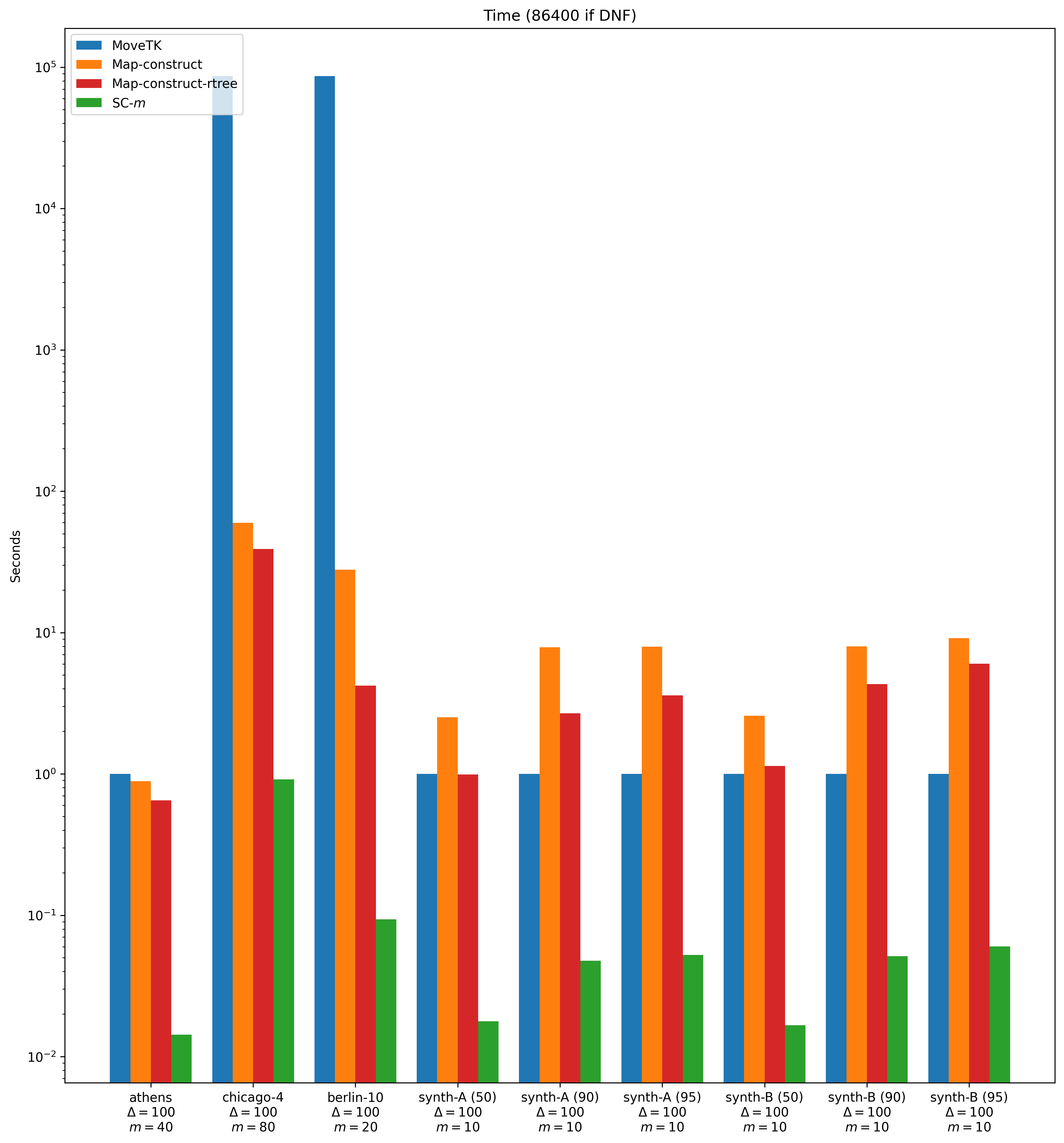}

    \caption{$\blacksquare$ BBGLL implementation comparison $\blacksquare$ Running time  $\blacksquare$ Logarithmic scaling}
    \label{fig:bar_bbgll_time}
\end{figure}

\begin{figure}[H]
\hspace*{-1.5in}
    \includegraphics[width=1.5\textwidth]{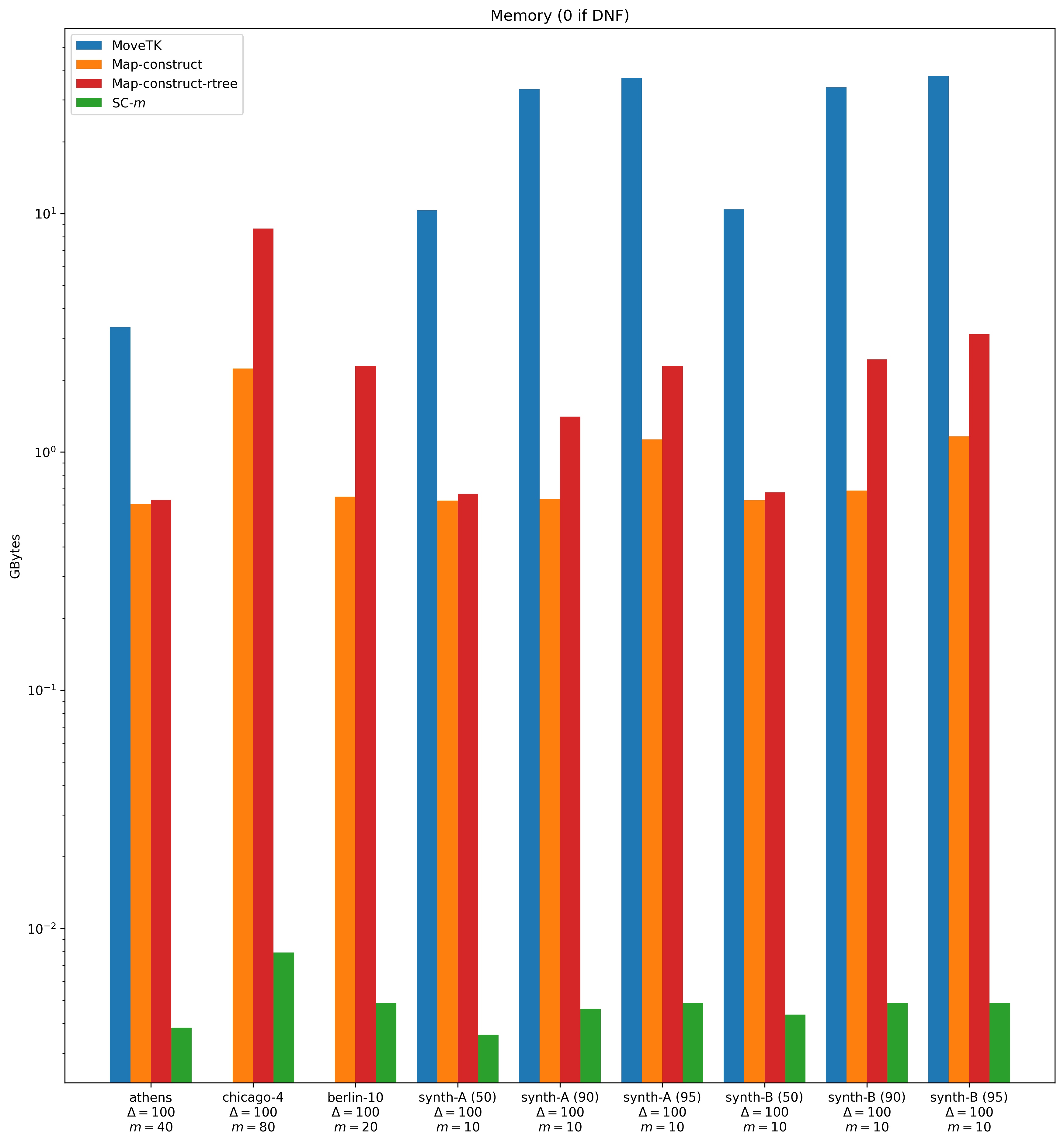}
    \caption{$\blacksquare$ BBGLL implementation comparison $\blacksquare$ Memory usage  $\blacksquare$ Logarithmic scaling  }
     \label{fig:bar_bbgll_memory}
\end{figure}

\newpage

\section{\texorpdfstring{Greedy clustering using SC($\max$, $\ell$, $\Delta$, $\mathcal{T}$)}{Greedy clustering using SC(max, l, D, T)}}
\label{sec:clustering}

Our primary contribution is that we propose to use \textbf{SC($\max$, $\ell$, $\Delta$, $\mathcal{T}$)} in a greedy clustering algorithm. 
On a high level, our clustering algorithm does the following:

\begin{itemize}
    \item We first choose the scoring function and constants (see Definition~\ref{def:score}).
    \item    Our input is the set of trajectories $\mathcal{T}$, some $\Delta$, some $\ell$, and an empty clustering $C$. 
    \item 
We obtain a \pathlet $(P, {\fancy{P}}) =$ \textbf{SC($\max$, $\ell$, $\Delta$, $\mathcal{T}$)}, and add it to $C$.
\item 
  We then remove all subtrajectories in ${\fancy{P}}$  from $\mathcal{T}$, and recurse. 
  \item 
We continue this process until the addition of $(P, {\fancy{P}})$ to $C$ no longer increases $\texttt{Score}(C)$.
\end{itemize}

\noindent
The implementation details depend on whether we are doing $k$-centre or $k$-median clustering.

\subsection{\texorpdfstring{$k$-centre clustering.}{k-centre clustering}}

Given are $(c_1, c_2, c_3)$ and a set of trajectories $\mathcal{T}$. 
We choose some integer $\ell$. 
The $k$-centre clustering scoring function weighs the value:  $\max\limits_{(P, {\fancy{P}}) \in C} \centercost(P, {\fancy{P}}) = \max\limits_{(P, {\fancy{P}}) \in C} \max\limits_{P' \in {\fancy{P}}} \fd(P, P')$.
\vspace{0.5cm}

\noindent
If $C$ is a set of $\Delta$-\pathlets (we assume there exists a \pathlet in $C$ whose radius is $\Delta$) then:
\[
\texttt{Score}(C) =  c_1 \cdot |C| + c_2 \cdot \Delta + c_3 \cdot \frac{|\mathcal{X} \backslash \Cov(C) |}{|\mathcal{X}|}. 
\]

\noindent
Let $C$ be a set of $\Delta$-\pathlets and $(P, {\fancy{P}})=$ \textbf{SC(}$\max$, $\ell$, $\Delta$, $\mathcal{T} - \Cov(C)$).
Per construction of our algorithm, $\Cov(P, {\fancy{P}}) \cap \Cov(C) = \emptyset$ and so the contribution of  $(P, {\fancy{P}})$ to \texttt{Score($C \cup \{ (P, {\fancy{P}})\}$)} equals $c_1 - c_3 \cdot \frac{|\Cov(P, {\fancy{P}})|}{|\mathcal{X}|}$.
We obtain $|\Cov(P, {\fancy{P}})|$ using constant additional time throughout our algorithm.
This leads to the following algorithm:

\begin{algodef}[ SC-$\ell$($(c_1, c_2, c_3)$, $\mathcal{T}$) under $k$-centre clustering]
\label{algodef:k-centre-length}
    We run Algorithm~\ref{alg:k-centre} $(c_1, c_2, c_3)$ for exponentially increasing $\Delta \in (2, 4, \ldots )$  where the subroutine $f$ is \textbf{SC}($\max$, $\ell$, $\Delta$, $\mathcal{T}$). We then output the maximum-score clustering $C$ over all runs.
\end{algodef}

 \begin{algorithm}[h]
    \caption{\texttt{k\_centre}(\texttt{vector} $(c_1, c_2, c_3)$, \texttt{value} $\Delta$, \texttt{trajectory} $\mathcal{T}$, \texttt{subroutine} $f$) }
    \label{alg:k-centre}
    \begin{algorithmic}
    \State $C \gets \emptyset$, $T \gets \mathcal{T}$.
    \While{$|T| > 0$}
    \State $(P, {\fancy{P}}) \gets $ the result of applying $f$ to ($\Delta$, $T$)
    \If{ \texttt{Score}($C \cup \{ (P, {\fancy{P}}) \}$) $>$ \texttt{Score}($C$) }
    \State $C$.add($(P, {\fancy{P}})$).
    \State $T \gets T - {\fancy{P}}$.
    \Else
    \State Break.
    \EndIf
    \EndWhile
    \State \Return $C$
    \end{algorithmic}
  \end{algorithm}

\noindent
We may also fix $m$ instead of $\ell$. Or, formally:

\begin{algodef}[ SC-$m$($(c_1, c_2, c_3)$, $\mathcal{T}$) under $k$-centre clustering ]
\label{algodef:k-centre-size}
         We run Algorithm~\ref{alg:k-centre} $(c_1, c_2, c_3)$ for exponentially increasing $\Delta \in (2, 4, \ldots )$  where the subroutine $f$ is \textbf{SC}($m$, $\max$, $\Delta$, $\mathcal{T}$). We then output the maximum-score clustering $C$ over all runs.
\end{algodef}

\subsection{\texorpdfstring{$k$-means clustering.}{k-means clustering} }

Our algorithm for $k$-means clustering is less straightforward. 
To get good performance we must iterate over different choices of $\Delta$ across the algorithm. Consider a clustering $C$ and $(P, {\fancy{P}})=$ \textbf{SC($\max$, $\ell$, $\Delta$, $\mathcal{T} - \Cov(C)$)} for some $\Delta$. 
Under $k$-centre clustering, the more points a $\Delta$-\pathlet covers, the better its scoring contribution. 
For $k$-means clustering this is no longer true as the contribution of $(P, {\fancy{P}})$ to \texttt{Score($C \cup \{ (P, {\fancy{P}})\}$)} equals:

\[
\textnormal{Contribution of } (P, {\fancy{P}}) \textnormal{ to } \texttt{Score}(C) = c_1 - c_3 \cdot \frac{|\Cov(P, {\fancy{P}})|}{|\mathcal{X}|} + \sum_{P' \in {\fancy{P}}} c_2 \cdot \fd(P, P').
\]

\noindent

\noindent
We abuse the fact that subtrajectories share no vertices to simplify the above expression to:

\[
\textnormal{Contribution of } (P, {\fancy{P}}) \textnormal{ to } \texttt{Score}(C) =  
c_1 + \sum_{P' \in {\fancy{P}}} \left( c_2 \cdot \fd(P, P') - c_3  \frac{|P'|}{|\mathcal{X}|} \right).
\]

\noindent
Consider a set of `candidate' \pathlets $S$.
To select the best \pathlet from $S$, we use the same greedy set cover technique as~\cite{agarwal2018subtrajectory}. We assign to every \pathlet a evaluation value. This is the ratio between how much it reduces the score, and how much it increases the score. 
I.e.,
\[
\texttt{evaluation}(P, {\fancy{P}}) = \frac{\sum\limits_{P' \in {\fancy{P}}  } c_3 \frac{|P'|}{|\mathcal{X}| } }{c_1 + \sum\limits_{P' \in {\fancy{P}}} c_2 \cdot \fd(P, P')}
\]

\noindent
We then select the \pathlet with the best evaluation among all $(P, {\fancy{P}}) \in S$. 
\noindent

\begin{algodef}[ SC-$\ell$($(c_1, c_2, c_3)$, $\mathcal{T}$) under $k$-means clustering]
    \label{algodef:k-means-length}
    We run Algorithm~\ref{alg:k-means} where the subroutine $f$ is \textbf{SC}($\max$, $\ell$, $\Delta$, $\mathcal{T}$). 
\end{algodef}

\begin{algodef}[ SC-$m$($(c_1, c_2, c_3)$, $\mathcal{T}$) under $k$-means clustering ]
    \label{algodef:k-means-size}
     We run Algorithm~\ref{alg:k-means} where the subroutine $f$ is \textbf{SC}($m$, $\max$, $\Delta$, $\mathcal{T}$). 
\end{algodef}

 \begin{algorithm}[H]
    \caption{\texttt{k\_means}(\texttt{vector} $(c_1, c_2, c_3)$, \texttt{trajectory} $\mathcal{T}$, \texttt{subroutine} $f$) }
    \label{alg:k-means}
    \begin{algorithmic}
    \State $C \gets \emptyset$, $T \gets \mathcal{T}$.
    \While{$|T| > 0$}
    \State $(P, {\fancy{P}}) $ $\gets$ null
    \For{$\Delta \in (2, 4, \ldots )$}
    \State $(Q, {\fancy{Q}}) \gets $ the result of applying $f$ to ($\Delta$, $\mathcal{T}$)
    \If{ $\texttt{evaluation}(Q, {\fancy{Q}}) > \texttt{evaluation}(P, {\fancy{P}})$ }
    \State $(P, {\fancy{P}}) \gets (Q, {\fancy{Q}})$. 
    \EndIf
    \EndFor
    \If{ \texttt{Score}($C \cup \{ (P, {\fancy{P}}) \}$) $>$ \texttt{Score}($C$)) }
    \State $C$.Add($(P, {\fancy{P}})$).
    \State $T \gets T - {\fancy{P}}$.
    \Else
    \State Break.
    \EndIf
    \EndWhile
    \State \Return $C$
    \end{algorithmic}
  \end{algorithm}

\clearpage
\section{A new algorithm as a subroutine}
\label{sec:pareto}

Our previous clusterings~\ref{algodef:k-centre-length}-\ref{algodef:k-means-size} fix a parameter $\ell$ (or $m$) and subsequently  \emph{only} produce \pathlets $(P, {\fancy{P}})$ of maximal cardinality with  $|P| = \ell$ (or of maximum length with $| {\fancy{P}}| = m$). 

This approach has three clear downsides:

\begin{enumerate}
    \item To find a suitable $\ell$ (or $m$) you have to train the algorithm on a smaller data set. 
    \item Using multiple values of $\ell$ (or $m$) to broaden the search space is computationally expensive. 
    \item Even if we fix constantly many choices for $\ell$ and $m$, the corresponding algorithm has a heavily restricted solution space which affects the quality of the output. 
\end{enumerate}

\noindent
To alleviate these downsides, we design a new algorithm based on the following observation: 

\begin{definition}
    For fixed $(\Delta, \mathcal{T})$, we define a bivariate Boolean function $\theta(m, \ell)$ where $\theta(m, \ell)$ is $1$ if and only if there exists a $\Delta$-\pathlet $(P, {\fancy{P}})$ of $\mathcal{T}$ with $m = |{\fancy{P}}|$ and $\ell = |P|$.  
\end{definition}

\noindent
The functions \textbf{SC($\max$, $\ell$, $\Delta$, $\mathcal{T}$)} and
\textbf{SC($m$, $\max$, $\Delta$, $\mathcal{T}$)} always output a $\Delta$-\pathlet $(P, {\fancy{P}})$ where $\theta(|P|, |{\fancy{P}}|)$ lies on the Pareto Front of $\theta(\cdot \, , \,  \cdot)$. 
If \pathlets on the Pareto front of $\theta(\cdot \, , \,  \cdot)$ are good candidates to include in a clustering, then it may be worthwhile to directly compute an approximation of this Pareto front as a candidate set:

\begin{definition}
    Given $(\Delta, \mathcal{T})$, a set of $\Delta$-\pathlets $S$ is a $2$-approximate Pareto front if for every $\Delta$-\pathlet $(P, {\fancy{P}})$ there exists a $(Q, {\fancy{Q}}) \in S$ with:
    $|{\fancy{P}}| \leq |{\fancy{Q}}| $ and $|P| \leq 2 |Q|$.
\end{definition}

\noindent
Our new algorithm,  \textbf{PSC($\Delta$, $\mathcal{T}$)}, \emph{iterates} over a $2$-approximate Pareto front. 
Storing this front explicitly uses $O(n^2 \log n)$ space. Instead, we generate these clusters on the fly.  

\subparagraph{High-level description.}
Let $B(\mathcal{T})$ be a balanced binary tree over $\mathcal{T}$.
Let $S(\Delta, \mathcal{T})$ start out as empty. Our algorithm \textbf{PSC}$(\Delta, \mathcal{T})$  (defined by Algorithm~\ref{alg:PSC}) does the following:

\begin{itemize}
    \item For each inner node in $B$, denote by  $\mathcal{T}[a, b]$ the corresponding subtrajectory of $\mathcal{T}$.
    \item For every prefix or suffix $P$ of $\mathcal{T}[a, b]$, let  $(P, {\fancy{P}})$ be a $\Delta$-\pathlet  maximising $|{\fancy{P}}|$.  
    \item We add $(P, {\fancy{P}})$ to $S(\Delta, \mathcal{T})$.
\end{itemize}

\subparagraph{Computing all prefix \pathlets for $\mathcal{T}[a, b]$.}
Given a trajectory $\mathcal{T}[a, b]$ and our current set $S(\Delta, \mathcal{T})$, we simultaneously compute for all prefixes $P$ of $\mathcal{T}[a, b]$ a maximum-cardinality \pathlet $(P, {\fancy{P}})$ and add this \pathlet to $S(\Delta, \mathcal{T})$.
We define an analogue of Definition~\ref{def:minrow}:

\begin{definition}
Consider row $a$ in $M_\Delta(\mathcal{T}, \mathcal{T})$. 
    For $(x, y) \in M_\Delta(\mathcal{T}, \mathcal{T})$ let $\maxcol_a(x, y)$ be the maximum column $d \leq y$ such that there exists a directed path from $(x, y)$ to $(a, d)$.
\end{definition}

\noindent
Given $\mathcal{T}[a,b]$ we initialize our algorithm by storing integers $a$ and $c \gets a$. 
We increment $c$ and maintain for all $y \in [n]$ the value $\maxcol_a(c, y)$. This uses $O(n)$ space, as does storing row $c$. 
Given this data structure, we alternate between a \texttt{StepSecond} and \texttt{Query} subroutine:

  \vspace{0.5cm}
\texttt{StepSecond:} We increment $c$ by $1$. 
For each $y \in [n]$, we first check whether $d(\mathcal{T}(c), \mathcal{T}(y)) > \Delta$.
If it is then $(c, y)$ is an isolated vertex in the graph $M_\Delta(\mathcal{T}, \mathcal{T})$. Thus, there exists no directed path from $(c, y)$ to any vertex in row $a$ and we set $\maxcol_a(c, y) = - \infty$. 
Otherwise, $\maxcol_a(c, y) = \max \{ \maxcol_a(\alpha, \beta) \mid  \alpha \in \{ c, c - 1 \}, \beta \in \{ y, y - 1 \}   \}$ which we compute in $O(1)$ time. 

\texttt{Query:}
For a prefix $P = \mathcal{T}[a, c]$ of $\mathcal{T}[a, b]$ we compute a maximum-cardinality $\Delta$-\pathlet $(P, {\fancy{P}})$ using the values $\maxcol_a(c, y)$ (see Algorithm~\ref{alg:query_fast}).

 \begin{algorithm}[b]
    \caption{\textbf{PSC}($\Delta$, $\mathcal{T}$)}
    \label{alg:PSC}
    \begin{algorithmic}
    \State $S(\Delta, \mathcal{T}) \gets \emptyset$.
    \State $B(\mathcal{T}) \gets$ balanced binary tree over $\mathcal{T}$.
    \ForAll{$T[a, b] \in B(\mathcal{T}) $}
    \ForAll{$c \in [a, b]$}
    \State \texttt{StepSecond}($\Delta$, $\mathcal{T}$, $c$).
    \State \texttt{Query}($T[a, c]$, $\Delta$, $\mathcal{T}$, $S(\Delta, \mathcal{T})$).
    \EndFor
    \ForAll{$c \in [a, b]$ in decreasing order}
    \State \texttt{StepSecond}($\Delta$, $\mathcal{T}$, $c$).
    \State \texttt{Query}'($T[c, b]$, $\Delta$, $\mathcal{T}$, $S(\Delta, \mathcal{T})$).
    \EndFor
    \EndFor
    \State \Return $S(\Delta, \mathcal{T})$
    \end{algorithmic}
  \end{algorithm}

Note that Algorithm~\ref{alg:query_fast} is essentially the same as the algorithm in~\cite{buchin2011detecting} that we described in Section~\ref{sec:clustering}. 
Thus, it computes a maximum-cardinality \pathlet.
The key difference is that we, through $\maxcol_a(b,j_2)$,  avoid walking through the free-space matrix in $O(\ell)$ time to compute $j_1$.  
We denote by  \texttt{Query}'($T[c, b]$, $\Delta$, $\mathcal{T}$, $S(\Delta, \mathcal{T})$) the function which computes for a suffix $P = \mathcal{T}[c, b]$ of $\mathcal{T}[a, b]$ a maximum-cardinality \pathlet  $(P, {\fancy{P}})$ in an analogous manner. 

\begin{lemma}
    The set $S(\Delta, \mathcal{T})$ is a $2$-approximate Pareto front. 
\end{lemma}

\begin{proof} 
    Consider any $\Delta$-\pathlet $(P, {\fancy{P}})$.
    Let $P = \mathcal{T}[a, b]$.
    Let $\cT[i, k]$ be the minimal subtrajectory in our tree that has $\cT[a, b]$ as a subtrajectory.
    Let the children of $\cT[i, k]$ in our tree be $\cT[i, j]$ and $\cT[j+1, k]$. Then, $a \le j < b$ by minimality of $\cT[i, k]$.
    Hence, $\cT[a, j]$ is a suffix of $\cT[i, j]$, and $\cT[j+1, b]$ is a prefix of $\cT[j+1, k]$.
    There must exist a trajectory $P' \in \{ \mathcal{T}[a, j],  \mathcal{T}[j+1, b]) \}$ which contains at least half as many vertices as $T[a, b]$.
    Since $P'$ is either a prefix or suffix of a subtrajectory in our tree, we must have added some maximum-cardinality \pathlet $(P', {\fancy{P}}')$ to $S(\Delta, \mathcal{T})$.
    Moreover, since $P'$ is a subtrajectory of $P$, $|{\fancy{P}}'| \geq | {\fancy{P}}|$ (since any discrete walk from row $b$ to row $a$ is also a discrete walk in the Free-space matrix from row $d$ to row $c$ for $[c, d] \subseteq [a, b]$).
\end{proof}

 \begin{algorithm}[h]
    \caption{Query($P = \mathcal{T}[a, c]$, $\Delta$, $\mathcal{T}, S(\Delta, \mathcal{T}))$}
    \label{alg:query_fast}
    \begin{algorithmic}
    \State $(c, j_2) \gets M_\Delta(\mathcal{T}, \mathcal{T})[c, n]$, ${\fancy{P}} \gets \emptyset$ 
    \While{$j_2 > 0$} 
    \If{$\maxcol_a(c, j_2) \neq -\infty $}
    \State $j_1 \gets \maxcol_a(b, j_2)$.
    \State Add    
    $\mathcal{T}[j_1, j_2]$ to  ${\fancy{P}}$.
    \State $(c, j_2) \gets (c, j_1 - 1)$.
    \Else 
    \State $(c, j_2) \gets (c, j_2 - 1)$. 
    \EndIf
    \EndWhile
    \State Add $(P, {\fancy{P}})$ to $S(\Delta, \mathcal{T})$.
    \end{algorithmic}
  \end{algorithm}

\begin{theorem}
The algorithm PSC($\Delta$, $\mathcal{T}$) iterates over a $2$-approximate Pareto front of $\mathcal{T}$ in $O(z \log^4 n)$ time using $O(n \log n$ space. Here, $z$ denotes the number of zeroes in $M_\Delta(\mathcal{T}, \mathcal{T})$. 
\end{theorem}

\begin{proof}
Consider the balanced binary tree $B(\mathcal{T})$.
Each level of depth of $B(\mathcal{T})$  corresponds to a set of subtrajectories that together cover $\mathcal{T}$. 
For every such subtrajectory $\mathcal{T}[a, b]$, both \texttt{StepSecond} and \texttt{Query} get invoked for each row $c \in [a, b]$. Since the tree has depth $\log n$, it follows that we invoke these functions $O(\log n)$ times per row $c$ of ${M}_\Delta(\mathcal{T}, \mathcal{T})$ (i.e. $O(n \log n)$ times in total). 
Naively, both the \texttt{StepSecond} and \texttt{Query} functions iterate over all columns in their given row $c$. 
However, just as in Section~\ref{sec:clustering}, we note that our functions skip over all $(j_2, c)$ for which $d(\mathcal{T}(j_2), \mathcal{T}(c)) > \Delta$. 
It follows that we may use our row generation trick to spend $O(Z_c \log^3 n)$ time instead ($Z_c$ denotes {the number of} zeroes in row $c$). 
Since each row $c$ gets processed $O(\log n)$ times, this upper bounds our running time by $O(z \log^4 n)$. 

For space, we note that invoking \texttt{StepSecond} and \texttt{Query} for a row $c$ need only rows $c$ and $c-1$ in memory. So the total memory usage is $O(n)$. 
Explicitly storing $S(\Delta, \mathcal{T})$ takes $O(n^2 \log n)$ space.
However, since we only iterate over $S(\Delta, \mathcal{T})$, the total space usage remains $O(n)$. 
The range tree therefore dominates the space with its $O(n \log n)$ space usage.
\end{proof}

\subsection{\texorpdfstring{Greedy clustering using \textbf{PSC($\Delta$, $\mathcal{T}$)}}{Greedy clustering using PSC(D, T)}}

We reconsider Algorithms~\ref{alg:k-centre} and \ref{alg:k-means} in Section~\ref{sec:clustering} for $k$-centre and $k$-means clustering. 
As written, these algorithms invoke a subroutine $f$ that returns a single \pathlet ($P$, ${\fancy{P}}$).
Then, the algorithm computes whether it wants to add ($P$, ${\fancy{P}}$) to the current clustering.
We can adapt these algorithms to this new function.
For $f$, we invoke \textbf{PSC}($\Delta$, $\mathcal{T}$). 
This function iterates over a set of \pathlets  ($P$, ${\fancy{P}}$).
For each \pathlet ($P$, ${\fancy{P}}$) that we encounter, we let either Algorithm~\ref{alg:k-centre} or \ref{alg:k-means} compute whether it wants to add ($P$, ${\fancy{P}}$) to the current clustering.

\section{Experiments}
\label{sec:experiments}

We propose three new greedy clustering algorithms for $k$-centre clustering:

\begin{enumerate}
    \item \texttt{SC-}$\ell$ runs Algorithm~\ref{alg:k-centre}, for exponentially scaling $\Delta$, with $f \gets $\textbf{SC($\max$, $\ell$, $\Delta$, $\mathcal{T}$)},
    \item \texttt{SC-}$m$ runs Algorithm~\ref{alg:k-centre}, for exponentially scaling $\Delta$, with $f \gets $ \textbf{SC($m$, $\max$, $\Delta$, $\mathcal{T}$)}, and
    \item \texttt{PSC} runs Algorithm~\ref{alg:k-centre}, for exponentially scaling $\Delta$,  with $f \gets $ $\textbf{PSC}(\Delta, \mathcal{T}$).
\end{enumerate}

\noindent
Our algorithms \texttt{SC-}$\ell$ and \texttt{SC}-$m$ need to fix a parameter $\ell$ and $m$ respectively. 
We choose $\ell \in \{ 4, 8, 16, 32, 64 \}$ and $m \in \{ 2, 4, 8, 16 \}$ and do a separate run for each choice. 
Thus, together with \texttt{PSC}, we run ten different algorithm configurations. 
We additionally compare to the subtrajectory clustering algorithm from~\cite{buchin2020improved}, which is oblivious to the scoring function.

\begin{itemize}
    \item \texttt{Map-construction} runs the algorithm in~\cite{buchin2020improved}. We terminate this map construction algorithm as soon as it has constructed a clustering of $\mathcal{T}$. 
\end{itemize}

\subparagraph{$k$-means clustering.}
For $k$-means clustering, we obtain three algorithms by running Algorithm~\ref{alg:k-means} instead. Note that our clustering algorithms invoked Algorithm~\ref{alg:k-centre} across several runs for exponentially scaling $\Delta$. Our clustering algorithms do not invoke Algorithm~\ref{alg:k-means} in this manner. Instead, Algorithm~\ref{alg:k-means} considers exponentially scaling $\Delta$'s internally.  Our empirical analysis will show that this change has significant ramifications for the clustering's performance on large input.

We again obtain ten algorithm configurations. In addition, we run the map construction algorithm from~\cite{buchin2020improved} and the $k$-means subtrajectory clustering algorithms from~\cite{agarwal2018subtrajectory}:

\begin{itemize}
    \item \texttt{Envelope} runs the algorithm from~\cite{agarwal2018subtrajectory} with the scoring parameters $(c_1, c_2, c_3)$. 
\end{itemize}

\subparagraph{Experimental setup.}
We design different experiments to measure performance and the scoring function.  For performance, we run fourteen experiments across the five real-world data sets.
Seven experiments for $k$-centre clustering, and seven for $k$-means clustering. 
Each experiment ran three times, using a different vector $(c_1, c_2, c_3)$. 
Each experiment runs $11$ programs for $k$-centre clustering, or $12$ for $k$-means clustering. If a program fails to terminate after 24 hours, we do not extract any output or memory usage. 

We restrict scoring measurements to the three smaller or subsampled real-world data sets. This is necessary to make \texttt{Map-construct} and \texttt{Envelope} terminate.  We add six synthetic data sets. We only compare under $k$-means clustering since only our algorithms can optimise for $k$-centre clustering.  
The experiments were conducted on a machine with a 4.2GHz AMD Ryzen 9 7950X3D and 128GB memory.
Over 90 runs hit the timeout. 
Running the total set of experiments sequentially takes over 100 days on this machine. 

The full results per experiment are organised in tables in Appendix~\ref{app:raw}.
In this section, we restrict our attention to only three plots.

\subparagraph{Choosing $(c_1, c_2, c_3)$.}
Recall that $c_1$ weighs the number of clusters in the clustering and that $c_3$ weighs the fraction of uncovered points. 
We follow the precedent set in~\cite{agarwal2018subtrajectory} and choose $c_1 = 1$, and $c_3$ equal to the number of trajectories in $\mathcal{T}$. 
Finally, $c_2$ weighs the \frechet distance across clusters. 
For each data set, we first find a choice for $c_2$ such the algorithms lead to a non-trivial clustering. We then also run the algorithms using the scoring vectors $(c_1, 10 \cdot  c_2, c_3)$ and $(c_1, 0.1 \cdot  c_2, c_3)$, for a total of three scoring vectors per data set. 

\subparagraph{Comparing quality.}
To compare solution quality, we measure the scoring function and other qualities such as: the total number of clusters in the clustering, the maximum \frechet distance from a subtrajectory to its centre, and the average \frechet distance across clusters. Our competitor algorithms frequently time out on the real-world data sets.
Hence, the qualitative comparison also makes use of our synthetic data set. We use three different configurations of our synthetic set to create input of different size. We choose to compare only under $k$-means clustering for a more fair comparison to other algorithms. Indeed, the algorithm by~\cite{agarwal2018subtrajectory} was designed for the $k$-means metric (\cite{buchin2020improved} is scoring function oblivious). 
Consequently, comparing under $k$-centre would only give us an unfair advantage.

\subsection{Results for $k$-means clustering}
The bar plot in Figure~\ref{fig:bar_kmeans_score} visualizes how the algorithms score. 
Figures \ref{fig:bar_kmeans_time} and \ref{fig:bar_kmeans_memory} show, respectively, a logarithmic-scale plot of the time and space usage. 
These plots use the `medium' choice for $c_2$ for each data set. 
Note that our choice of $c_2$ is very similar across the data sets in the figure.
This is because \emph{Athens}, \emph{Berlin} and \emph{Chicago} come from the same source, and thereby have the same scaling.
We subsequently designed our synthetic data to match this scaling.
Plots for the other scoring vector choices are contained in Appendix~\ref{app:bar_plots}. 

\subparagraph{Clustering quality.}
Figure~\ref{fig:bar_kmeans_score} shows that \texttt{Map-construct} (which is oblivious to the scoring function) scores poorly on all inputs.  One reasons for this is that it always outputs a large number of \pathlets.  
The \texttt{Envelope} algorithm, which was specifically designed for this scoring function, almost always obtains the best score. 

For any data set, for any scoring vector, we always find a choice for $\ell$ (or $m$) where \texttt{SC}-$\ell$ and \texttt{SC}-$m$ obtain a comparable score to \texttt{Envelope}.
The algorithm \texttt{PSC} scores best out of all newly proposed clustering algorithms.
Across all data sets and scoring functions it frequently competes with, and sometimes even improves upon, the score of \texttt{Envelope}.

\begin{figure}[H]
    \hspace*{-1.45in}
    \includegraphics[width=1.5\textwidth]{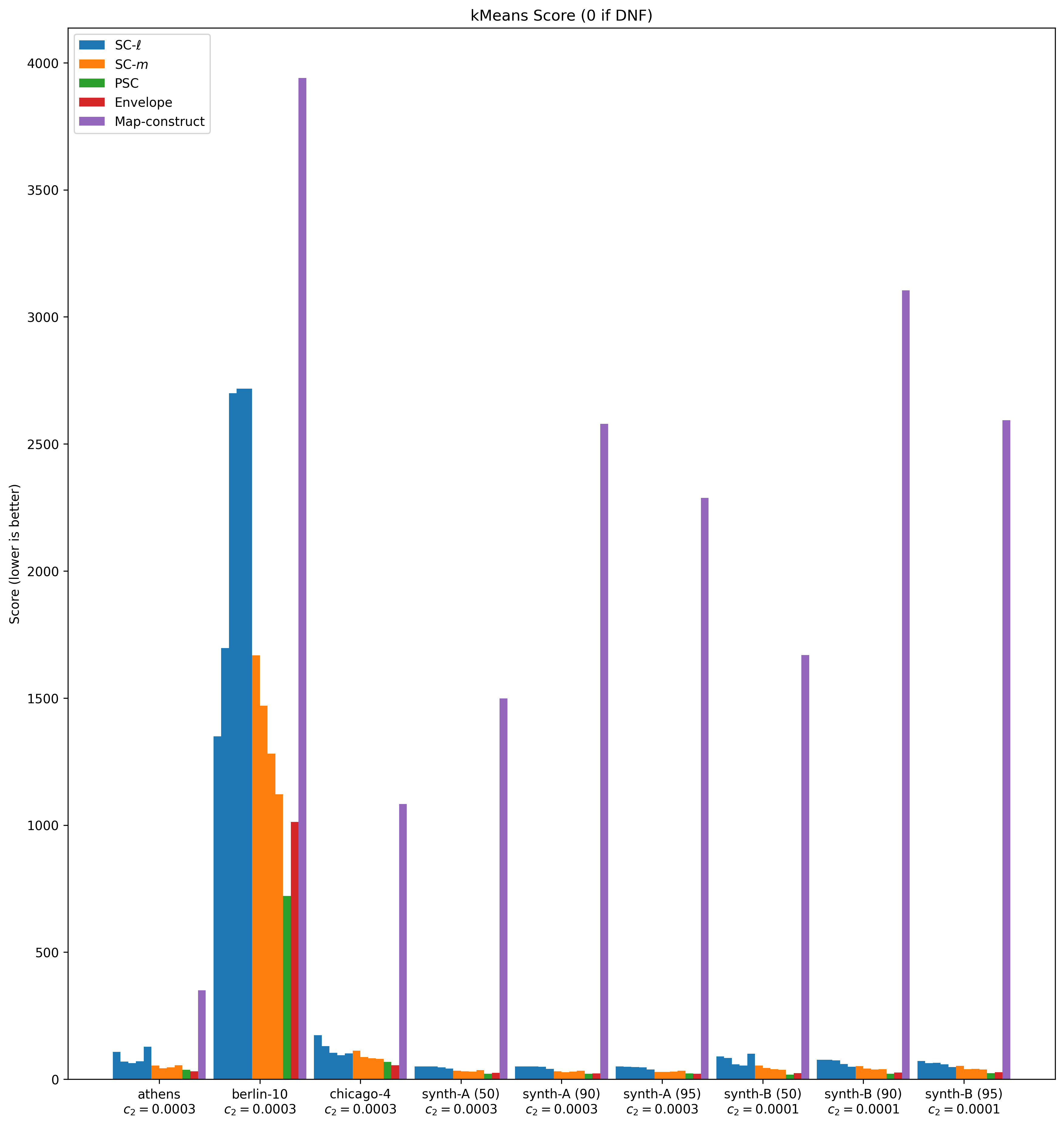}
    \caption{$\blacksquare$ $k$-means clustering $\blacksquare$ $k$-means score $\blacksquare$}
    \label{fig:bar_kmeans_score}
\end{figure}

\subparagraph{Other quality metrics. }
We briefly note that the full data in the appendix also considers other quality metrics: the number of \pathlets and the maximum and average \frechet distance within a \pathlet.
Perhaps surprisingly, \texttt{Map-construct} does not perform better on the latter two metrics even though it outputs many \pathlets.  
\texttt{Envelope} scores worse on these metrics across nearly all data sets and scoring vectors. However, it stays within a competitive range.

\subparagraph{Runtime comparison.}
Figure~\ref{fig:bar_kmeans_time} shows the runtime performance in logarithmic scale. 
On the \emph{Berlin} data set, no algorithm terminates in time.
On \emph{Drifter} and \emph{UniD}, only \texttt{PSC} terminates within the time limit. In the figure, we consider the synthetic data sets their place.

\texttt{Envelope} may produce high-quality clusterings, but takes significant time to do so. 
Despite using considerably less memory than \texttt{Map-construct}, it is consistently more than a factor two slower.
Our algorithms \texttt{SC}-$\ell$ and \texttt{SC}-$m$ can always find a choice of $\ell$ or $m$ such that previous algorithms are orders of magnitude slower. 
Our \texttt{PSC} algorithm is always orders of magnitude faster than the existing clustering algorithms. 
This is especially true on larger data sets, which illustrates that our algorithms have better scaling. 

When we compare our own algorithms, we note that we can frequently find a choice for $\ell$ such that \texttt{SC}-$\ell$ is significantly faster than \texttt{SC}-$m$ and \texttt{PSC}.
However, when considering the runtime  to quality ratio, \texttt{PSC} appears to be a clear winner.

\subparagraph{Memory consumption.}
Figure~\ref{fig:bar_kmeans_memory} shows memory performance on a logarithmic scale. The figure excludes \emph{Drifter}, \emph{UniD} and \emph{Berlin} for the same reason as when comparing running times. 
We record zero memory usage whenever an algorithm times out.

We observe that \texttt{Map-construct} uses an order of magnitude more memory than its competitors. This is because it stores up to $O(n^2)$ candidate \pathlets
in memory, together with the entire free-space matrix $M_\Delta(\mathcal{T}, \mathcal{T})$ for various choices of $\Delta$. \texttt{Envelope} also has a large memory footprint and uses at least a factor $100$ more memory on all of the inputs. 

When comparing our own implementations, we observe that \texttt{SC}-$\ell$ and \texttt{PSC} have a very small memory footprint.
These algorithms use $O(n \ell)$ and $O(n)$ space. Since we fix $\ell$ to be a small constant, both algorithms use near linear space.  
The space usage of \texttt{SC}-$m$ is $O(n \ell)$ where $\ell$ is the length of the longest subtrajectory $P$ such that there exists a $\Delta$-\pathlet $(P, {\fancy{P}})$ with $|{\fancy{P}}| \geq m$. 
The data shows that for many choices of $(m, \Delta)$, this $O(n \ell)$ space is considerable. Yet, it is still considerably lower than pre-existing algorithms.

\subsection{Results for $k$-centre clustering}
When comparing score, runtime, or memory usage under $k$-centre clustering, we reach the same conclusions. 
The only difference is that  the performance of \texttt{SC}-$m$ becomes more competitive to \texttt{SC}-$\ell$ and \texttt{PSC} on large data sets. 
This is because \textbf{SC}($m$,  $\max$, $\Delta$, $\mathcal{T})$ uses much time and space whenever both the data set and $\Delta$ are large. 
Our $k$-means clustering meta-algorithm considers a large $\Delta$ each round, and invokes \textbf{SC}($m$,  $\max$, $\Delta$, $\mathcal{T})$ every round.

Our $k$-centre clustering meta-algorithm has a separate run for each $\Delta$. If $\Delta$ is large, it terminates quickly because it adds very high-cardinality \pathlets to the clustering. 
This speeds up all our implementations, but \texttt{SC}-$m$ the most.

\section{Conclusion}

We proposed new subtrajectory clustering algorithms.
Our first two approaches use the algorithms \textbf{SC}($\max$, $\ell$, $\Delta$, $\mathcal{T}$) and \textbf{SC}($m$, $\max$, $\Delta$, $\mathcal{T}$) by BBGLL~\cite{buchin2011detecting}. 
We use these algorithms as a subroutine in a greedy clustering algorithm to cluster a set of trajectories $\mathcal{T}$.

Consider the Boolean function $\theta(\ell, m)$ that outputs \texttt{true} if there exists a $\Delta$-\pathlet $(P, \mathcal{P})$ of $\mathcal{T}$ with $|P| = \ell$ and $|\mathcal{P}| = m$. We observe that these functions always output a point on the Pareto front of this function. 
We show a new algorithm \textbf{PSC}($\Delta$, $\mathcal{T}$) that, instead of returning a single \pathlet on this Pareto front, iterates over a $2$-approximation of this Pareto front instead.
Intuitively, this function iterates over a broader set of candidate clusters that have comparable quality to any output of \textbf{SC}($\max$, $\ell$, $\Delta$, $\mathcal{T}$) and \textbf{SC}($m$, $\max$, $\Delta$, $\mathcal{T}$). 
We create a greedy clustering algorithm \texttt{PSC} using this algorithm as a subroutine. 

Our analysis shows that our new clustering algorithms significantly outperform previous algorithms in running time and space usage. 
We primarily use the scoring function from~\cite{agarwal2018subtrajectory} to measure clustering quality and observe that our algorithms give competitive scores. 
For other quality metrics, our algorithms often perform better than their competitors. 
We observe that \texttt{PSC} is the best algorithm when comparing the ratio between score and performance.

\newpage

\begin{figure}[H]
\hspace*{-1.5in}
    \includegraphics[width=1.5\textwidth]{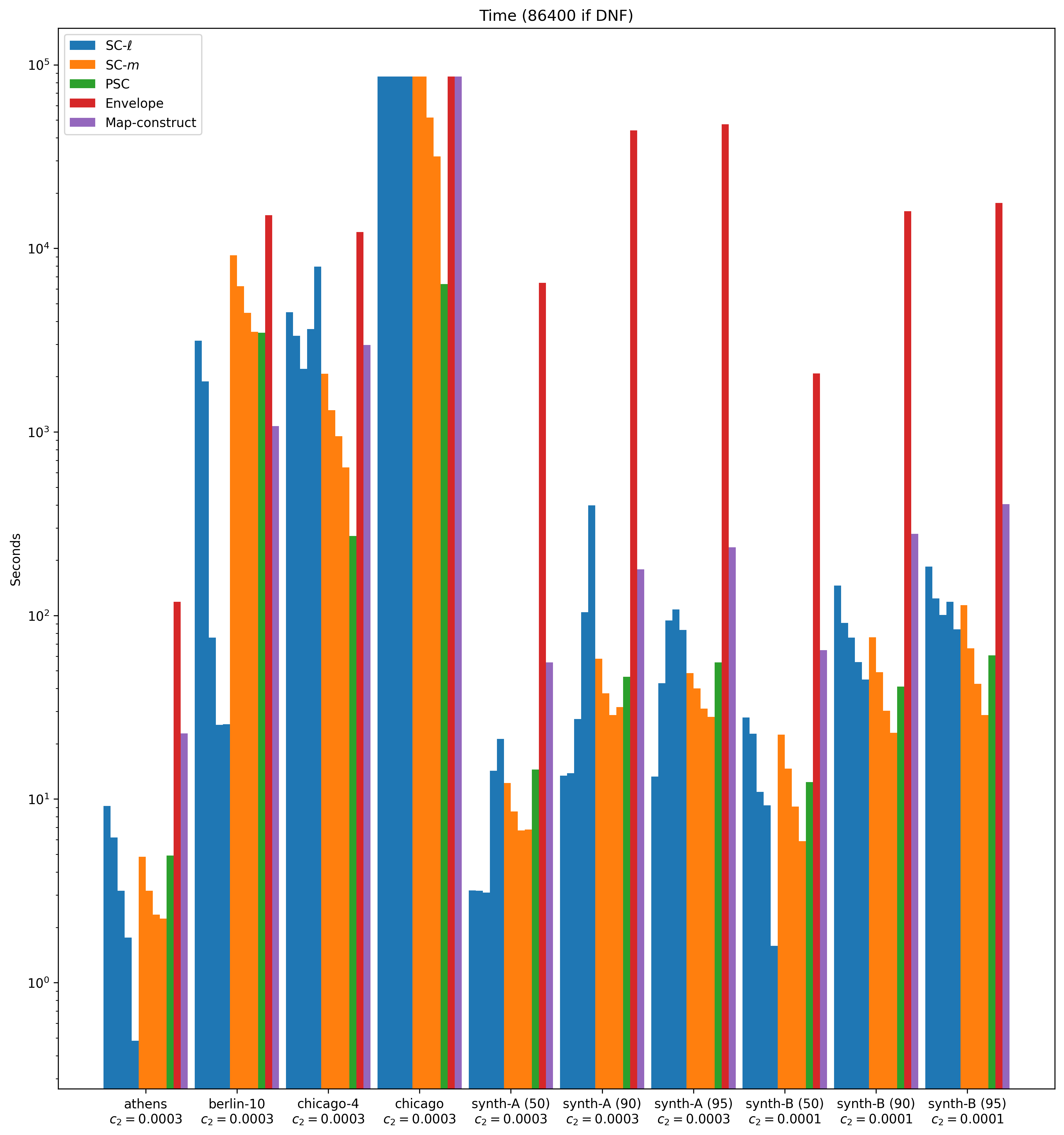}
    \caption{$\blacksquare$ $k$-means clustering $\blacksquare$ Running time  $\blacksquare$ Logarithmic scaling}
    \label{fig:bar_kmeans_time}
\end{figure}

\begin{figure}[H]
\hspace*{-1.45in}
    \includegraphics[width=1.5\textwidth]{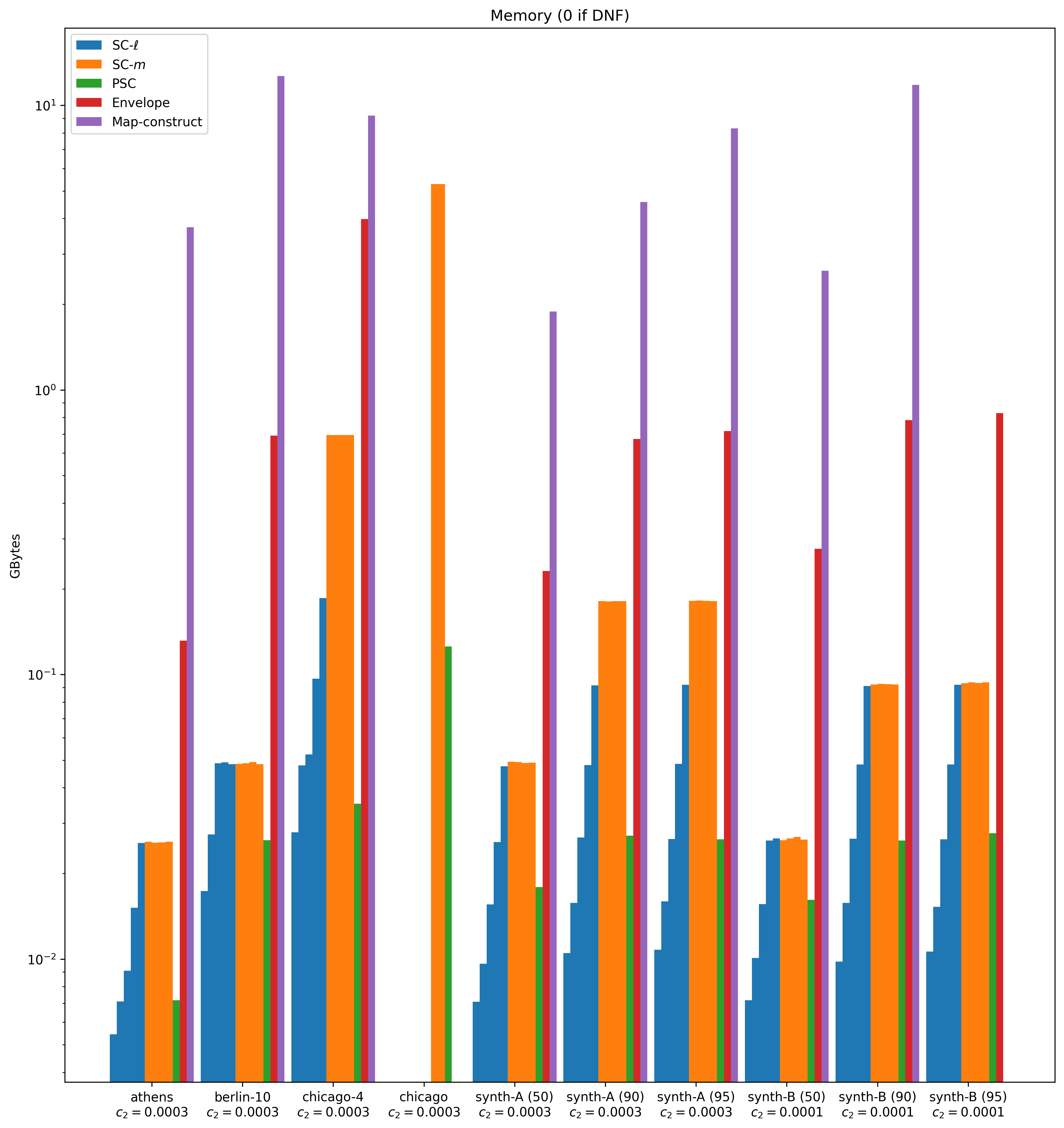}
    \caption{$\blacksquare$ $k$-means clustering $\blacksquare$ Memory usage  $\blacksquare$ Logarithmic scaling \\
    Whenever a run did not terminate, we display a memory usage of zero. }
     \label{fig:bar_kmeans_memory}
\end{figure}

\bibliography{refs.bib}

\appendix

\newpage
\section{\texorpdfstring{Introducing and solving SC($m$, $\ell$, $\Delta$, $\mathcal{T}$)}{Introducing and solving SC(m, l, D, T)} }
\label{sec:BBGLL}

BBGLL~\cite{buchin2011detecting} introduce the problem \textbf{SC($m$, $\ell$, $\Delta$, $\mathcal{T}$)}.
The input are integers $m$, $\ell$, $\Delta$ and a set of trajectories $\mathcal{T}$.  The goal is to find a $\Delta$-\pathlet $(P, {\fancy{P}})$ where $|{\fancy{P}}| \geq m$ and $|P| = \ell$.  
The problem statement may be slightly altered by maximizing $m$ or $\ell$ instead, giving the problems \textbf{SC($\max$, $\ell$, $\Delta$, $\mathcal{T}$)} and \textbf{SC($m$, $\max$, $\Delta$, $\mathcal{T}$)}.
They present a theoretical algorithm to solve $\textbf{SC($m$, $\ell$, $\Delta$, $\mathcal{T}$)}$ in $O(n m \ell)$ time and $O(n \ell)$ space.
This algorithm can be adapted to solve the other two problem variants too.

There exist two single-core implementations of this algorithm (the MOVETK library~\cite{movetk}, and the one by~\cite{buchin2017clustering, buchin2020improved}). 
\cite{movetk} uses $O(n^2)$ space and $O(n^2 + n m \ell)$ time.
We conjecture that \cite{buchin2017clustering, buchin2020improved} use $O(n \ell)$ space and $O(n^2 + n m \ell)$ time.
There also exists a GPU-based implementation of this algorithm~\cite{gudmundsson2012gpu}. We do not consider this approach since we have no way to fairly compare single- and multi-core performance.
We provide an optimised implementation of the BBGLL algorithm that significantly improves upon previous single-core implementations. We are unsure how many of our optimisations are included in~\cite{buchin2017clustering, buchin2020improved}. Theoretically, we slightly improve the asymptotic running time compared to~\cite{buchin2011detecting}, as we use $O( n m \ell + (n + z) \log^3 n)$ time instead of $O(n m \ell + n^2)$ time (where $z$ denotes the number of zeroes in $M_\Delta(\mathcal{T}, \mathcal{T})$).

\subparagraph{The existing implementations.}
We explain the implementation of \textbf{SC($m$, $\ell$, $\Delta$, $\mathcal{T}$)} by \cite{movetk}. 
As a preprocessing step, store the matrix $M_\Delta(\mathcal{T}, \mathcal{T})$ where each cell $(i, j)$ stores $\minrow(i, j)$ ~(Definition~\ref{def:minrow}).
After preprocessing, the algorithm has at all times in memory an input trajectory $T \in \mathcal{T}$, and two integers $\alpha \leq \beta \leq |T|$ where $\beta = \alpha  + \ell - 1$. We use the notation $P = T[\alpha, \beta] = \mathcal{T}[a, b]$. 
The algorithm alternates between three functions: \\

 \texttt{StepFirst}: increment $\alpha$ (and thus $a$).\\

\texttt{StepSecond}: if $\beta < |T|$, increment $\beta$ (and thus $b$).

Else let $T'$ succeed $T$ in $\mathcal{T}$, set $\alpha \gets 1$, $\beta \gets \ell$, $T \gets T'$, $a \gets b + 1$ and $b \gets a + \ell$. \\

 \texttt{Query}: given $P = \mathcal{T}[a, b]$, greedily compute an $m$-cardinality $\Delta$-\pathlet $(P, {\fancy{P}})$ as {described in Algorithm~\ref{alg:SCquery}}.

 \begin{algorithm}[h]
    \caption{\texttt{Query}($P = \mathcal{T}[a, b]$, $m$, $\ell$, $\Delta$, $\mathcal{T})$ by~\cite{buchin2011detecting}}
   \label{alg:SCquery}
    \begin{algorithmic}
    \State $(b, j_2) \gets (b, n)$ and ${\fancy{P}} \gets \emptyset$. 
    \While{$|{\fancy{P}}| < m$ and $j_2 > 0$} 
    \If{$\minrow(b, j_2) \leq a$}
    \State Find the maximum $j_1$ such that there exists a path from $(b, j_2)$ to $(a, j_1)$. 
    \State This takes $O(\ell)$ time by greedily stepping to the right-most neighbour 
    $(b', j')$ where,
    \State $\minrow(b', j') \leq a$.
    \State Add    
    $\mathcal{T}[j_1, j_2]$ to ${\fancy{P}}$ and set $(b, j_2) \gets (b, j_1 - 1)$.
    \Else 
    \State $(b, j_2) \gets (b, j_2 - 1)$. 
    \EndIf
    \EndWhile
    \If{$|{\fancy{P}}| \geq m$}
   \State  \Return $(P, {\fancy{P}})$
    \EndIf
    \end{algorithmic}
  \end{algorithm}

 \subparagraph{Algorithm analysis.}
 Computing $M_\Delta(\mathcal{T}, \mathcal{T})$ and the associated integers in each cell uses $O(n^2)$ space and can be computed in $O(n^2)$ time using a sweepline. 

  The \texttt{StepFirst} and \texttt{StepSecond} functions take constant time. 
 The if-statement can be true at most $m$ times and then takes  $O(\ell)$ time. The else statement may occur at most $O(n)$ times before we get a point that is out of bounds, therefore the total query time  is at most $O(n + m \ell)$. 
 BBGLL~\cite{buchin2011detecting} prove that one can solve SC($m$, $\ell$, $\Delta$, $\mathcal{T}$) by invoking \texttt{StepFirst}, \texttt{StepSecond} and \texttt{Query} $O(n)$ times for a total runtime of $O(n^2 + n m \ell)$. 
 
\subsection{Two performance improvement techniques}
\label{sub:improvement}

We discuss two techniques that improve the performance of our implementation. 

\subparagraph{Windowing. }
The implementation in~\cite{movetk} computes the matrix $M_\Delta(\mathcal{T}, \mathcal{T})$ in $O(n^2)$ time and space. 
The original publication BBGLL~\cite{buchin2011detecting} already describes how to reduce this space usage to $O(n \ell)$ by expanding the \texttt{StepFirst} and \texttt{StepSecond} functions, storing only the $\ell \times n$ submatrix/subgraph $M_\Delta(P = \mathcal{T}[a, b], \mathcal{T})$. 
We store only the $z' \in O(n \ell)$ non-isolated vertices in this subgraph. What remains is to dynamically maintain  $M_\Delta(P = \mathcal{T}[a, b], \mathcal{T})$: \\

\texttt{StepFirst:}
Before incrementing $a$, remove the top row from $M_\Delta(P, \mathcal{T})$ from memory. \\ 

\texttt{StepSecond:}
After incrementing $b$, we compute the new row in the graph $M_\Delta(P, \mathcal{T})$ in $O(n)$ time. 
Iterate over all $n$ columns from low to high.  For each point $(b, j)$ there are at most $3$ cells with an in-edge to $(b, j)$. Moreover, if $M_\Delta(b, j) = 0$ then $\minrow(b, j)$ is the minimum of  $\minrow(b', j')$ for $b' \in \{b, b-1 \}$ and $j' \in \{j, j-1 \}$, which we access in $O(1)$ time.

    \subparagraph{Row generation.}
    \texttt{StepSecond} and \texttt{Query} now share a bottleneck that we will call row generation. 
    \texttt{StepSecond} iterates over all columns in a new row. 
    \texttt{Query} (through the else-statement) iterates over $O(n)$ columns in the bottom-most row. Both subroutines can be adjusted to have an output-sensitive running time instead through the following structure:

\begin{itemize}
    \item    Construct a binary tree over $\mathcal{T}$ by recursively {splitting} the concatenated trajectory into two roughly equal-sized children. 
    Store for each node all vertices in the corresponding subtrajectory in a range tree~\cite{DBLP:books/lib/BergCKO08}. 
    \item     Given $(\Delta, T', q)$ with $\Delta \in \mathbb{R}$, $T'$ a subtrajectory corresponding to a node in our tree, and a query point $q$, the corresponding range tree can return the furthest point along $T'$ that is within distance $\Delta$ of $q$ in $O(\log^2 n)$ total time. 
\end{itemize}
 
\noindent     
We use this data structure to speed up both functions.  \\

    \texttt{StepSecond:} 
    For a given integer $b$, we need to compute for all $j \in [n]$ the value $\minrow(b, j)$. 
    Note that $\minrow(b, j) < b$ only if $M_\Delta(T, \mathcal{T})[b, j] = 0$ (i.e., $\mathcal{T}(j)$ is within distance $\Delta$ of $T(b)$).   
    We generate all such $(b, j)$ from high to low: 
    Let us have added $(b, j)$ to the output. 
    The trajectory $\mathcal{T}[1, j - 1]$ consists of $O(\log n)$ subtrees of our binary tree over $\mathcal{T}$. We query our each subtree with $q = T(b)$ get the furthest point along the subtree  within distance $\Delta$ of $T(b)$ in $O(\log^3 n)$ time. Thus, we add $k$ points to our graph using $O(k \log^3 n)$ total time.  \\

    \texttt{Query:} 
    If in the row generation we carefully record all $(b, j)$ with $\minrow(b, j) < a$, then we can avoid Case $2$. This improves the run time of one \texttt{Query} call to $O(m \ell)$. \\

\noindent   
    The \texttt{Query} function gets invoked $O(n)$ times. The \texttt{StepSecond} functions collectively use $O(n + z)$ range queries. 
    Overall, this reduces the total running time to $O( n m \ell + (n+z) \log^3 n)$.

    \subsection{Experimental evaluation}
    \label{sub:comparison}

Our experiments fix a data set $\mathcal{T}$ and some $\ell$ and $\Delta$. 
We then run \textbf{SC}($m$, $\max$, $\Delta$, $\mathcal{T}$).
Consequently, we have to fix one fewer parameter compared to running \textbf{SC}($m$, $\ell$, $\Delta$, $\mathcal{T}$). 
For each data set, we find some $(m_0, \Delta_0)$ that yields a reasonable output, then then consider $m \in \{m_0/2, m_0, 2 m_0\}$ and $\Delta \in \{\Delta_0/2, \Delta_0, 2 \Delta_0\}$.
This way, we cover a wide range of parameters.
There now exist four single-core implementations of the algorithm by BBGLL~\cite{buchin2011detecting}:

\begin{enumerate}
    \item \texttt{SC-$m$} refers to the C++ implementation found in this paper. 
    \item \texttt{Map-construct}~(\cite{buchin2017clustering}) refers to the Java implementation in~\cite{buchin2017clustering},
    \item \texttt{Map-construct-rtree} refers to the Java implementation in~\cite{buchin2020improved}, and
    \item \texttt{Movetk}~\cite{movetk} refers to the C++ implementation in \cite{movetk}. 
\end{enumerate}

\subparagraph{Data sets.}
We compare these algorithms using the \emph{Athens-small}, \emph{Chicago-4}, \emph{Berlin-10} and \emph{Synthetic} data set.
All other data sets are so large that Movetk runs out of memory before terminating. 
The raw data is given in three tables that follow this section.
We discuss the results of these experiments through two bar plots that we show in the main body.

    \begin{figure} \begin{center}
    \begin{minipage}{0.5\textwidth}
    \begin{tabular}{@{}lrrrr@{}}        
        \toprule
        Name & Seconds & GBytes & $\Delta$ & $m$\\
        \midrule
        MoveTK & 1.00 & 3.21 & 50 & 20\\
Map-construct & 0.87 & 0.62 & 50 & 20\\
Map-construct-rtree & 0.54 & 0.61 & 50 & 20\\
        \medskip
SC-$m$ & 0.01 & 0.00 & 50 & 20\\
        MoveTK & 1.00 & 3.21 & 50 & 40\\
Map-construct & 0.78 & 0.61 & 50 & 40\\
Map-construct-rtree & 0.52 & 0.61 & 50 & 40\\
        \medskip
SC-$m$ & 0.01 & 0.00 & 50 & 40\\
        MoveTK & 1.00 & 3.21 & 50 & 80\\
Map-construct & 0.76 & 0.61 & 50 & 80\\
Map-construct-rtree & 0.52 & 0.60 & 50 & 80\\
        \medskip
SC-$m$ & 0.01 & 0.00 & 50 & 80\\
        MoveTK & 1.00 & 3.34 & 100 & 20\\
Map-construct & 0.92 & 0.62 & 100 & 20\\
Map-construct-rtree & 0.66 & 0.62 & 100 & 20\\
        \medskip
SC-$m$ & 0.01 & 0.00 & 100 & 20\\
        MoveTK & 1.00 & 3.34 & 100 & 40\\
Map-construct & 0.89 & 0.61 & 100 & 40\\
Map-construct-rtree & 0.65 & 0.63 & 100 & 40\\
        \medskip
SC-$m$ & 0.01 & 0.00 & 100 & 40\\
        MoveTK & 1.00 & 3.34 & 100 & 80\\
Map-construct & 0.85 & 0.63 & 100 & 80\\
Map-construct-rtree & 0.66 & 0.62 & 100 & 80\\
        \medskip
SC-$m$ & 0.02 & 0.00 & 100 & 80\\
        MoveTK & 1.00 & 3.59 & 200 & 20\\
Map-construct & 1.22 & 0.62 & 200 & 20\\
Map-construct-rtree & 1.00 & 0.66 & 200 & 20\\
        \medskip
SC-$m$ & 0.02 & 0.00 & 200 & 20\\
        MoveTK & 1.00 & 3.59 & 200 & 40\\
Map-construct & 1.12 & 0.62 & 200 & 40\\
Map-construct-rtree & 0.99 & 0.66 & 200 & 40\\
        \medskip
SC-$m$ & 0.02 & 0.00 & 200 & 40\\
        MoveTK & 1.00 & 3.59 & 200 & 80\\
Map-construct & 0.97 & 0.60 & 200 & 80\\
Map-construct-rtree & 0.98 & 0.63 & 200 & 80\\
        \medskip
SC-$m$ & 0.03 & 0.00 & 200 & 80\\
        \bottomrule
    \end{tabular}
\end{minipage}

        \caption{$\blacksquare$ \textbf{SC}($m$, $\max$, $\Delta$, $\mathcal{T}$) $\blacksquare$ Athens-small $\blacksquare$}

    \end{center} \end{figure}

\newpage

    \begin{figure} \begin{center}
    \begin{minipage}{0.5\textwidth}
    \begin{tabular}{@{}lrrrr@{}}
        \toprule
        Name & Seconds & GBytes & $\Delta$ & $m$\\
        \midrule
        MoveTK & --- & --- & 50 & 10\\
Map-construct & 27.18 & 0.65 & 50 & 10\\
Map-construct-rtree & 2.79 & 1.34 & 50 & 10\\
        \medskip
SC-$m$ & 0.05 & 0.00 & 50 & 10\\
        MoveTK & --- & --- & 50 & 20\\
Map-construct & 24.69 & 0.62 & 50 & 20\\
Map-construct-rtree & 2.75 & 1.32 & 50 & 20\\
        \medskip
SC-$m$ & 0.05 & 0.00 & 50 & 20\\
        MoveTK & --- & --- & 50 & 40\\
Map-construct & 22.15 & 0.63 & 50 & 40\\
Map-construct-rtree & 2.63 & 1.22 & 50 & 40\\
        \medskip
SC-$m$ & 0.06 & 0.00 & 50 & 40\\
        MoveTK & --- & --- & 100 & 10\\
Map-construct & 27.12 & 0.66 & 100 & 10\\
Map-construct-rtree & 4.03 & 2.32 & 100 & 10\\
        \medskip
SC-$m$ & 0.08 & 0.00 & 100 & 10\\
        MoveTK & --- & --- & 100 & 20\\
Map-construct & 27.89 & 0.65 & 100 & 20\\
Map-construct-rtree & 4.20 & 2.30 & 100 & 20\\
        \medskip
SC-$m$ & 0.09 & 0.00 & 100 & 20\\
        MoveTK & --- & --- & 100 & 40\\
Map-construct & 25.81 & 0.66 & 100 & 40\\
Map-construct-rtree & 4.22 & 2.27 & 100 & 40\\
        \medskip
SC-$m$ & 0.11 & 0.00 & 100 & 40\\
        MoveTK & --- & --- & 200 & 10\\
Map-construct & 33.43 & 0.71 & 200 & 10\\
Map-construct-rtree & 8.19 & 4.51 & 200 & 10\\
        \medskip
SC-$m$ & 0.18 & 0.01 & 200 & 10\\
        MoveTK & --- & --- & 200 & 20\\
Map-construct & 31.32 & 1.16 & 200 & 20\\
Map-construct-rtree & 7.76 & 4.50 & 200 & 20\\
        \medskip
SC-$m$ & 0.19 & 0.00 & 200 & 20\\
        MoveTK & --- & --- & 200 & 40\\
Map-construct & 30.23 & 0.70 & 200 & 40\\
Map-construct-rtree & 7.97 & 4.63 & 200 & 40\\
        \medskip
SC-$m$ & 0.21 & 0.01 & 200 & 40\\
        \bottomrule
    \end{tabular}
\end{minipage}

        \caption{$\blacksquare$ \textbf{SC}($m$, $\max$, $\Delta$, $\mathcal{T}$) $\blacksquare$ Berlin-10 $\blacksquare$\\
        \texttt{MoveTK} exceeds the 128GB memory limit in all nine runs.}

    \end{center} \end{figure}

\newpage

    \begin{figure} \begin{center}
    \begin{minipage}{0.5\textwidth}
    \begin{tabular}{@{}lrrrr@{}}
        \toprule
        Name & Seconds & GBytes & $\Delta$ & $m$\\
        \midrule
        MoveTK & --- & --- & 50 & 40\\
Map-construct & 64.03 & 1.17 & 50 & 40\\
Map-construct-rtree & 25.41 & 4.45 & 50 & 40\\
        \medskip
SC-$m$ & 0.48 & 0.01 & 50 & 40\\
        MoveTK & --- & --- & 50 & 80\\
Map-construct & 54.49 & 0.80 & 50 & 80\\
Map-construct-rtree & 24.37 & 2.27 & 50 & 80\\
        \medskip
SC-$m$ & 0.48 & 0.01 & 50 & 80\\
        MoveTK & --- & --- & 50 & 160\\
Map-construct & 42.75 & 0.73 & 50 & 160\\
Map-construct-rtree & 24.83 & 2.23 & 50 & 160\\
        \medskip
SC-$m$ & 0.58 & 0.01 & 50 & 160\\
        MoveTK & --- & --- & 100 & 40\\
Map-construct & 73.22 & 2.17 & 100 & 40\\
Map-construct-rtree & 38.80 & 9.12 & 100 & 40\\
        \medskip
SC-$m$ & 0.82 & 0.01 & 100 & 40\\
        MoveTK & --- & --- & 100 & 80\\
Map-construct & 59.63 & 2.24 & 100 & 80\\
Map-construct-rtree & 39.01 & 8.66 & 100 & 80\\
        \medskip
SC-$m$ & 0.92 & 0.01 & 100 & 80\\
        MoveTK & --- & --- & 100 & 160\\
Map-construct & 51.59 & 0.72 & 100 & 160\\
Map-construct-rtree & 38.71 & 8.44 & 100 & 160\\
        \medskip
SC-$m$ & 1.06 & 0.01 & 100 & 160\\
        MoveTK & --- & --- & 200 & 40\\
Map-construct & 87.29 & 4.52 & 200 & 40\\
Map-construct-rtree & 68.88 & 10.87 & 200 & 40\\
        \medskip
SC-$m$ & 1.80 & 0.02 & 200 & 40\\
        MoveTK & --- & --- & 200 & 80\\
Map-construct & 80.58 & 4.25 & 200 & 80\\
Map-construct-rtree & 70.14 & 10.64 & 200 & 80\\
        \medskip
SC-$m$ & 2.11 & 0.02 & 200 & 80\\
        MoveTK & --- & --- & 200 & 160\\
Map-construct & 59.23 & 2.17 & 200 & 160\\
Map-construct-rtree & 70.61 & 10.53 & 200 & 160\\
        \medskip
SC-$m$ & 2.53 & 0.02 & 200 & 160\\
        \bottomrule
    \end{tabular}
\end{minipage}

        \caption{$\blacksquare$ \textbf{SC}($m$, $\max$, $\Delta$, $\mathcal{T}$) $\blacksquare$ Chicago-4 $\blacksquare$\\
        \texttt{MoveTK} exceeds the 128GB memory limit in all nine runs.}

    \end{center} \end{figure}

\newpage 
 
\section{Bar plots for more variables}
\label{app:bar_plots}

For each data set we run all algorithms using three choices for the scoring vector $(c_1, c_2, c_3)$. In particular, we find for each data set some suitable triple $(c_1, c_2, c_3)$.
Then, we also consider the vectors $(c_1, 0.1 c_2, c_3)$ and $(c_1, 10 c_2, c_3)$. We call these three vectors the `small', `medium' and `large' scoring vectors. 
 Note that the choice of scoring vector matters even when we only compare algorithmic performance.
This is because all involved algorithms make different algorithmic choices depending on the scoring parameters that they receive. 
 \newpage 

\begin{figure}[H]
    \begin{adjustwidth}{-0.25\textwidth}{}
    \includegraphics[width=1.5\textwidth]{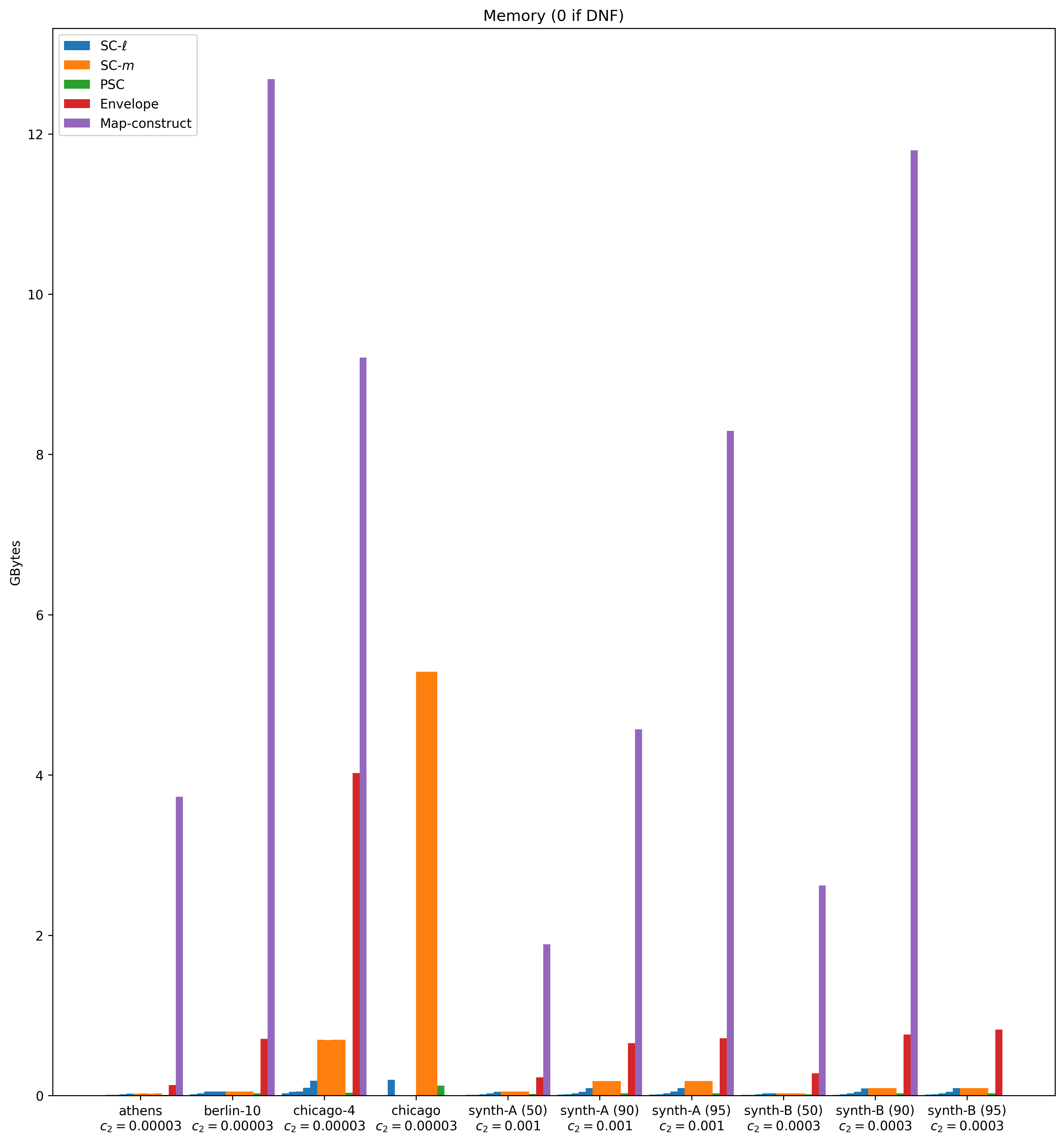}
    \end{adjustwidth}
    \caption{$\bullet$ $k$-means clustering $\bullet$ Memory usage  $\bullet$ Small $c_2$ $\bullet$}
\end{figure}

\begin{figure}[H]
    \begin{adjustwidth}{-0.25\textwidth}{}
    \includegraphics[width=1.5\textwidth]{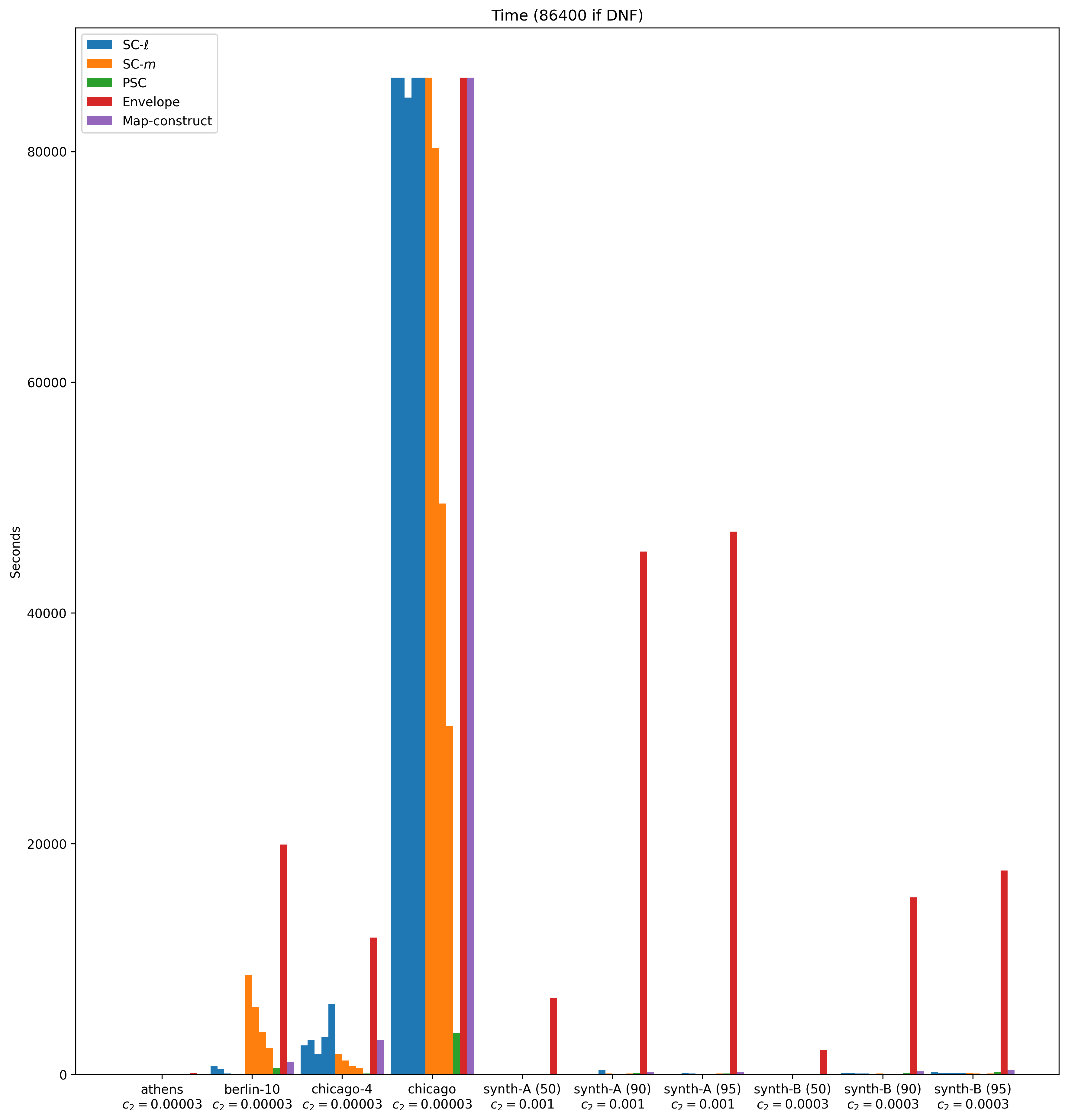}
    \end{adjustwidth}
    \caption{$\bullet$ $k$-means clustering $\bullet$ Running time  $\bullet$ Small $c_2$ $\bullet$}
\end{figure}

\begin{figure}[H]
    \begin{adjustwidth}{-0.25\textwidth}{}
    \includegraphics[width=1.5\textwidth]{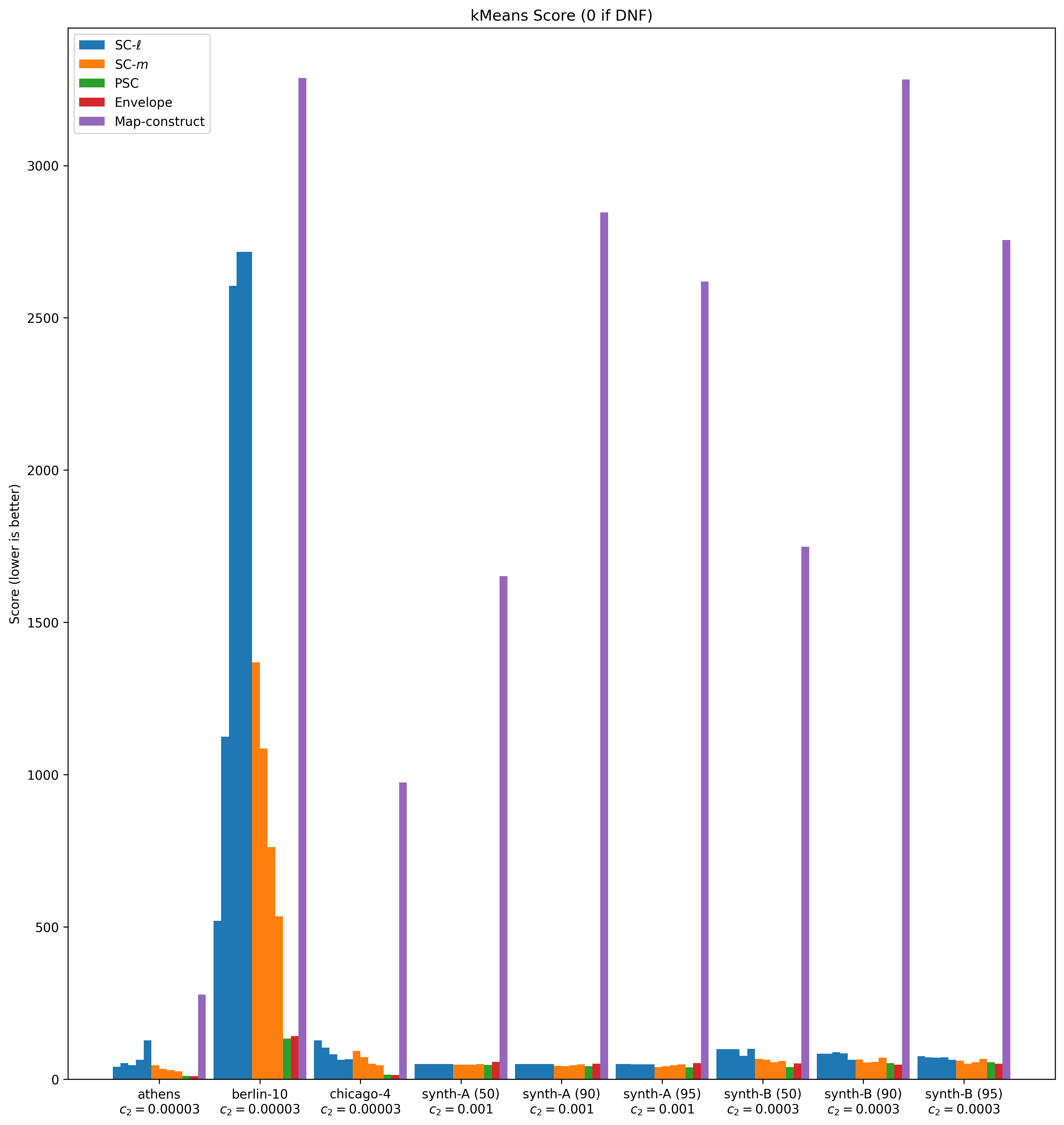}
    \end{adjustwidth}
    \caption{$\bullet$ $k$-means clustering $\bullet$ $k$-means score $\bullet$ Small $c_2$ $\bullet$}
\end{figure}

\begin{figure}[H]
    \begin{adjustwidth}{-0.25\textwidth}{}
    \includegraphics[width=1.5\textwidth]{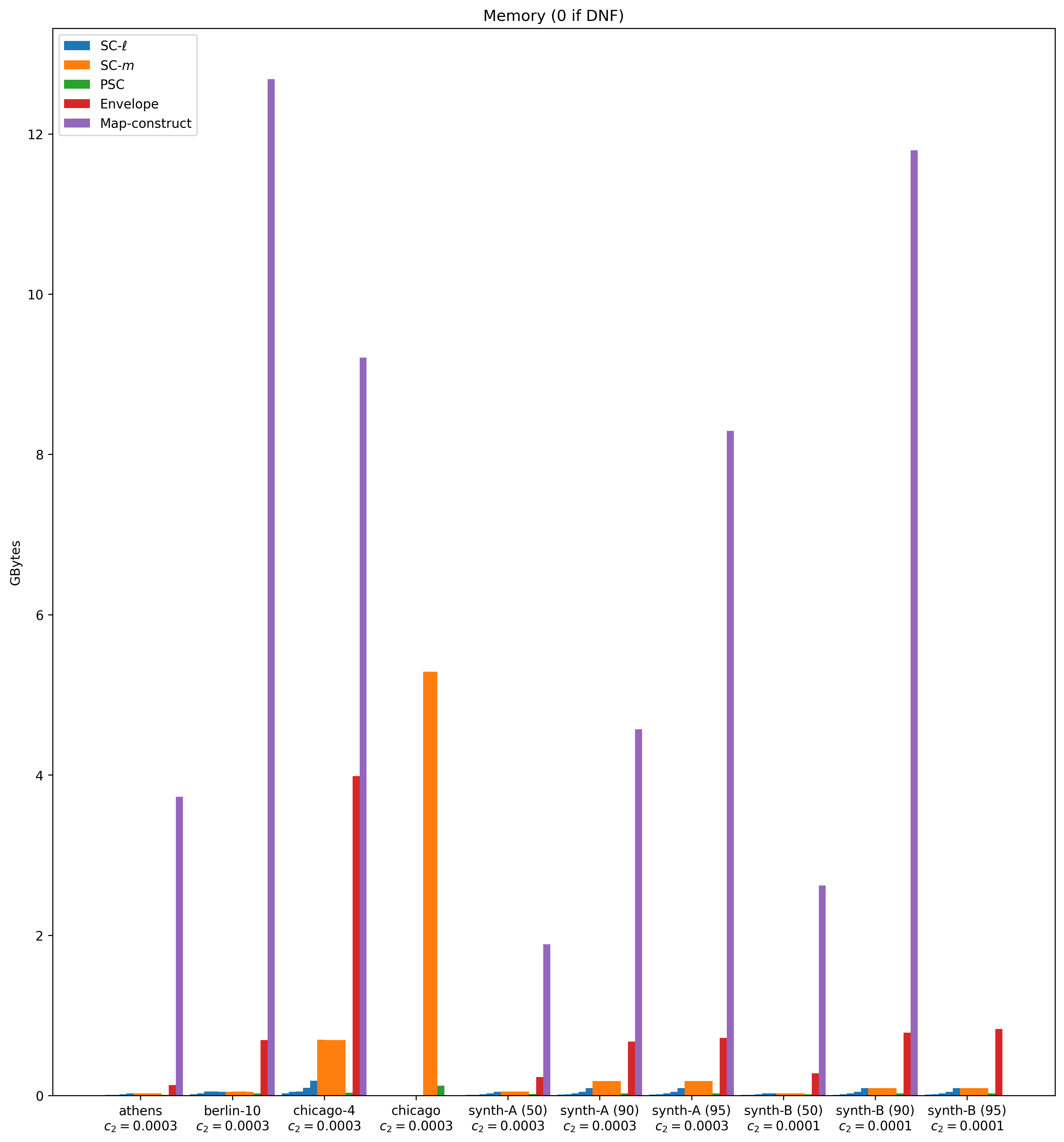}
    \end{adjustwidth}
    \caption{$\bullet$ $k$-means clustering $\bullet$ Memory usage  $\bullet$ Medium $c_2$ $\bullet$}
\end{figure}

\begin{figure}[H]
    \begin{adjustwidth}{-0.25\textwidth}{}
    \includegraphics[width=1.5\textwidth]{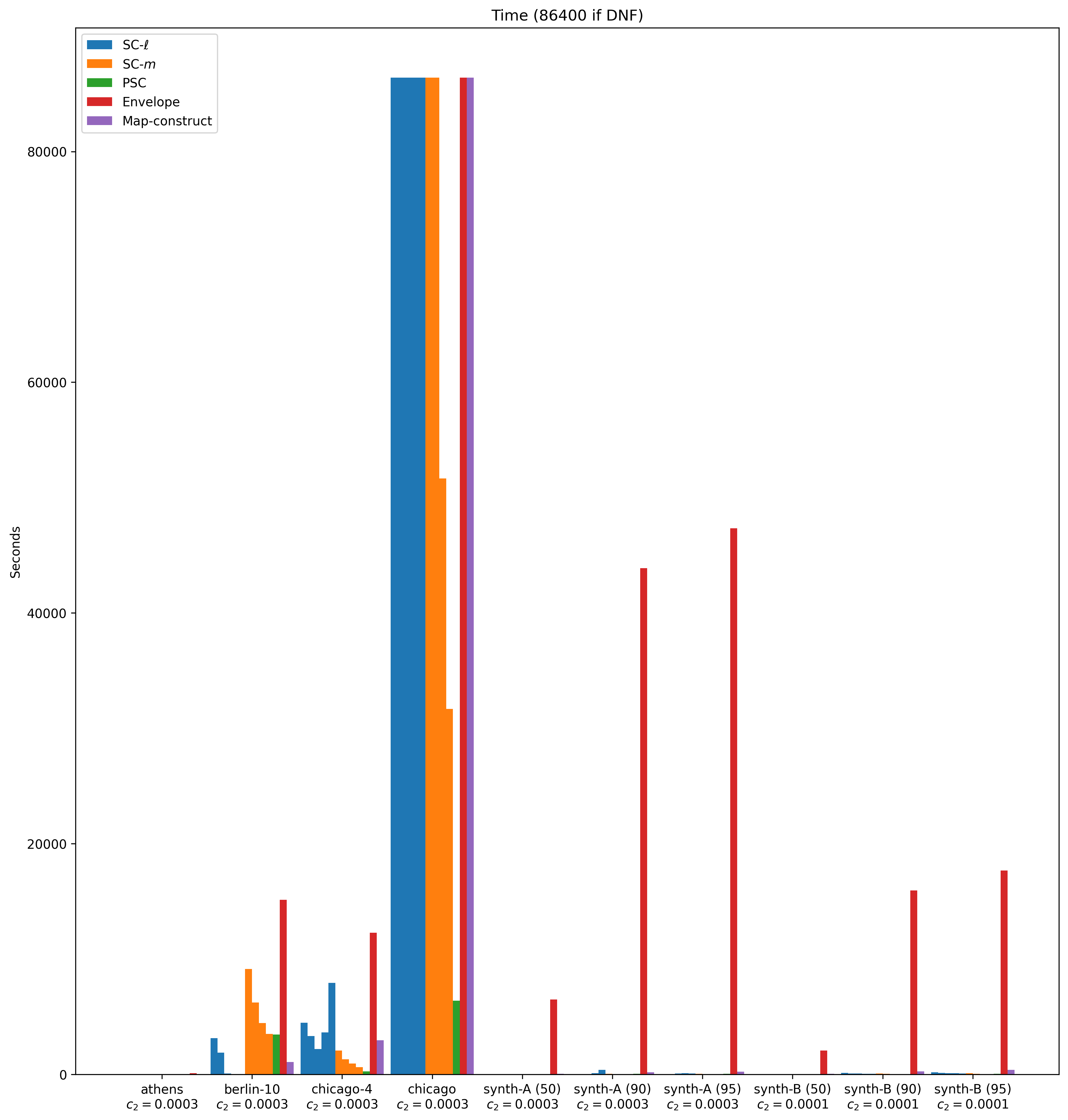}
    \end{adjustwidth}
    \caption{$\bullet$ $k$-means clustering $\bullet$ Running time  $\bullet$ Medium $c_2$ $\bullet$ }
\end{figure}

\begin{figure}[H]
    \begin{adjustwidth}{-0.25\textwidth}{}
    \includegraphics[width=1.5\textwidth]{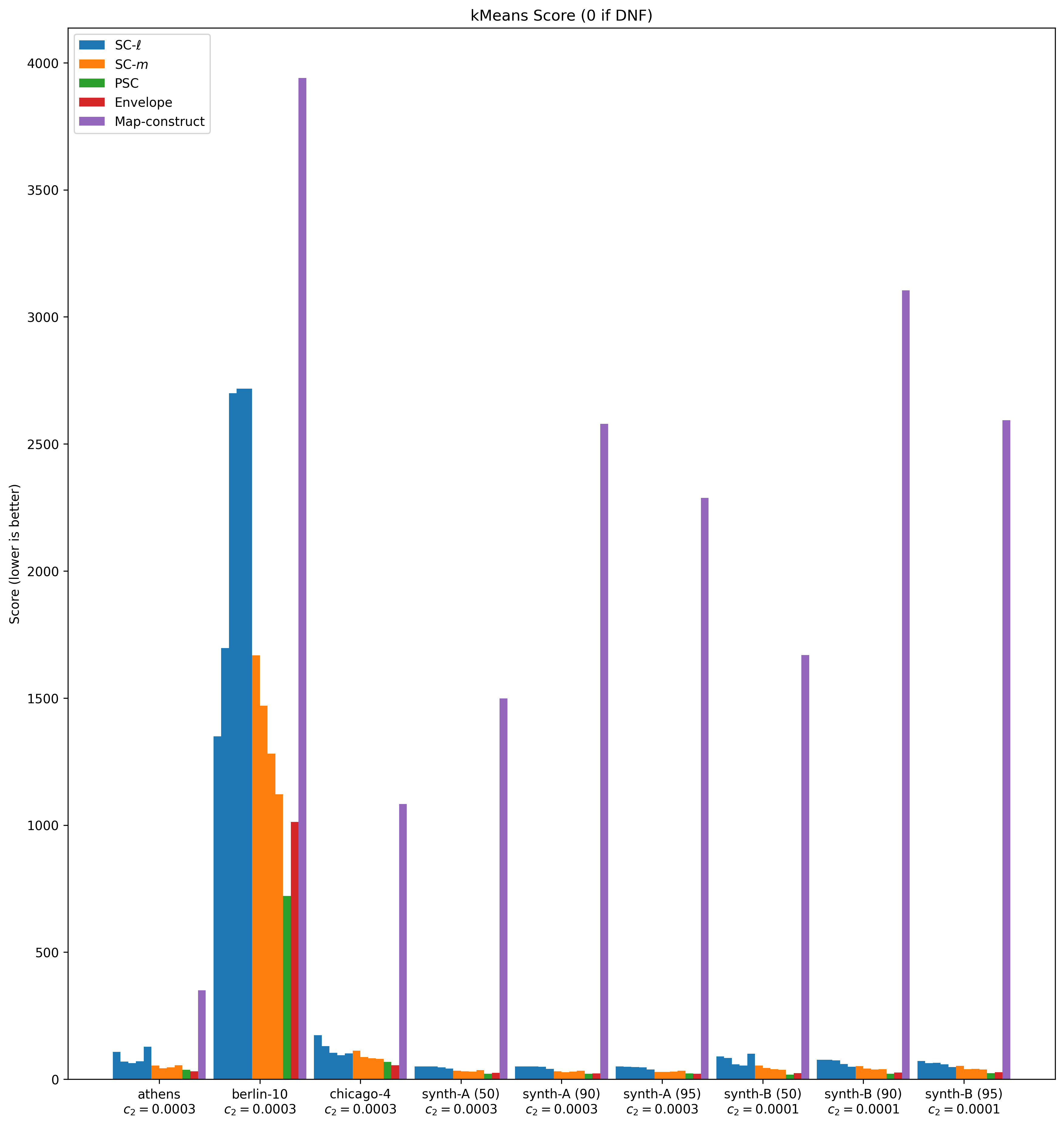}
    \end{adjustwidth}
    \caption{$\bullet$ $k$-means clustering $\bullet$ $k$-means score $\bullet$ Medium $c_2$ $\bullet$}
\end{figure}

\begin{figure}[H]
    \begin{adjustwidth}{-0.25\textwidth}{}
    \includegraphics[width=1.5\textwidth]{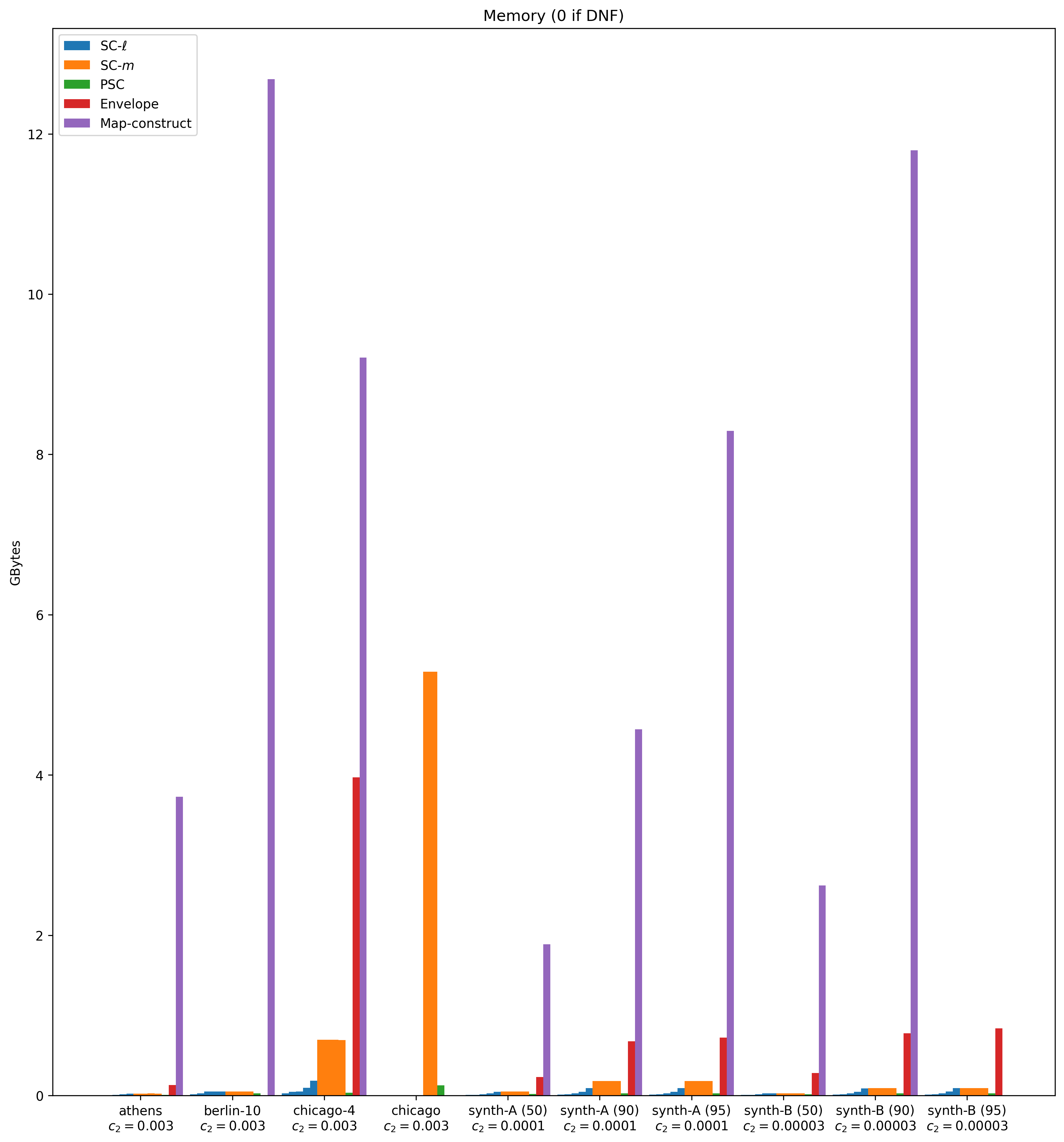}
    \end{adjustwidth}
    \caption{$\bullet$ $k$-means clustering $\bullet$ Memory usage  $\bullet$ Large $c_2$ $\bullet$}
\end{figure}

\begin{figure}[H]
    \begin{adjustwidth}{-0.25\textwidth}{}
    \includegraphics[width=1.5\textwidth]{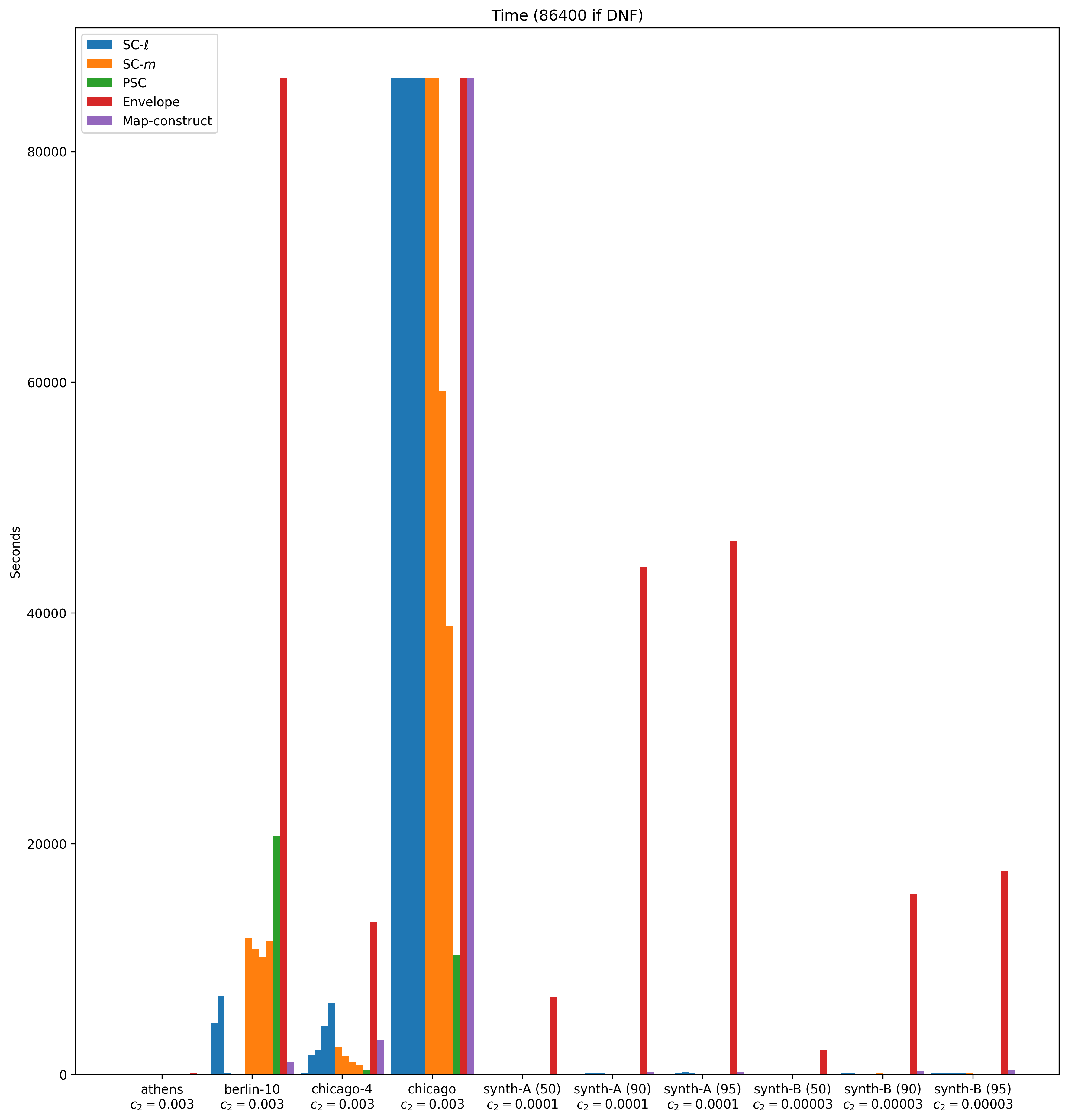}
    \end{adjustwidth}
    \caption{$\bullet$ $k$-means clustering $\bullet$ Running time  $\bullet$ Large $c_2$ $\bullet$}
\end{figure}

\begin{figure}[H]
    \begin{adjustwidth}{-0.25\textwidth}{}
    \includegraphics[width=1.5\textwidth]{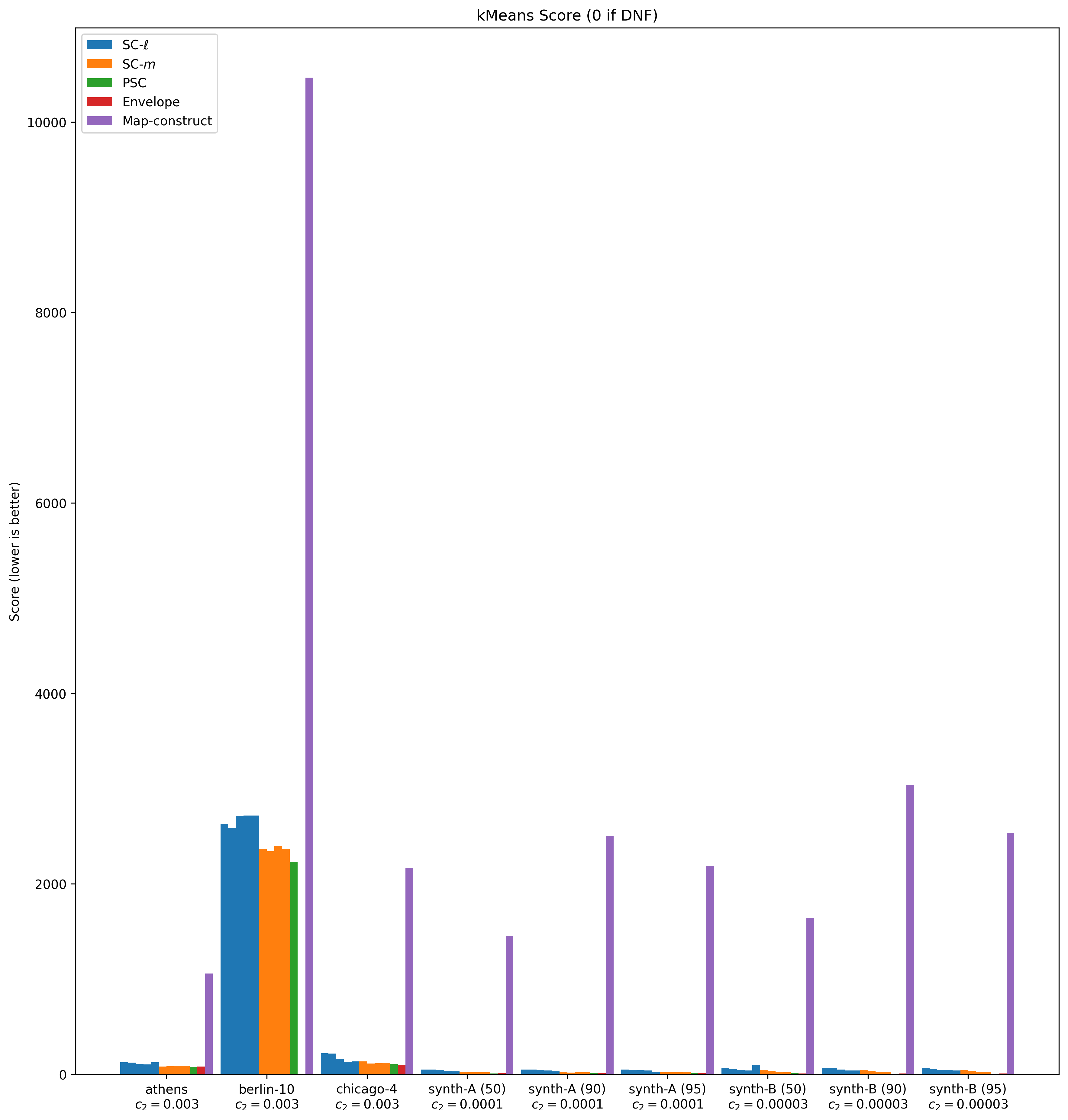}
    \end{adjustwidth}
    \caption{$\bullet$ $k$-means clustering $\bullet$ $k$-means score $\bullet$ Large $c_2$ $\bullet$}
\end{figure}

\section{Raw data}
\label{app:raw}

In the following tables, we compare our implementations \texttt{SC-}$\ast$ and \texttt{PSC}, to the implementations presented in the literature. Specifically the \emph{Envelope} implementation from \cite{agarwal2018subtrajectory} and the \emph{Map-construct} implementation from \cite{buchin2020improved}.

Each table specifies whether we do $k$-means clustering or $k$-centre clustering.
It then specifies the data set.
We list time and space consumption, and, for each clustering, its score using the scoring vector.
We also report additional data the reader may find interesting such as the number of clusters, or the maximum \frechet distance between a subtrajectory in a cluster and its centre.

We run our experiments on various data sets, including very large data. Consequently, some had timed out by our 24 hour deadline; for timed-out instances, we mark their row by $-$. In a few of the examples, the output is the empty clustering (and thus not all measures are well-defined); these can be recognised by the number of clusters being marked as $0$. 
For these rows, other metrics may display exceptional values. E.g., then the maximum \frechet distance is $-\infty$.

\begin{figure} \begin{center}
        \begin{minipage}{\textwidth}
    \adjustbox{max width=\textwidth}
    {
\begin{tabular}{@{}lrrcrrrrrr@{}}
        \toprule
        Name & Seconds & GBytes & $c$ & Clusters & Max Frechet & Avg Frechet & kMeans & Max Size & Avg Size\\
        \midrule
        SC-$\ell$-4 & 3.39 & 0.01 & (1.0, 3e-05, 128) & 15 & 1294.70 & 262.66 & 40.82 & 297 & 158.07\\
        SC-$\ell$-8 & 3.77 & 0.01 & (1.0, 3e-05, 128) & 16 & 1624.80 & 410.16 & 52.70 & 178 & 60.56\\
        SC-$\ell$-16 & 2.72 & 0.01 & (1.0, 3e-05, 128) & 11 & 1635.81 & 538.91 & 45.72 & 34 & 20.36\\
        SC-$\ell$-32 & 1.42 & 0.01 & (1.0, 3e-05, 128) & 4 & 818.49 & 514.21 & 63.86 & 40 & 23.00\\
        SC-$\ell$-64 & 0.48 & 0.03 & (1.0, 3e-05, 128) & 0 & -inf & nan & 128.00 & -inf & nan\\
        SC-$m$-2 & 4.44 & 0.03 & (1.0, 3e-05, 128) & 38 & 1584.88 & 399.76 & 46.26 & 3 & 2.89\\
        SC-$m$-4 & 3.06 & 0.03 & (1.0, 3e-05, 128) & 26 & 3274.57 & 485.05 & 34.50 & 5 & 4.77\\
        SC-$m$-8 & 2.04 & 0.03 & (1.0, 3e-05, 128) & 18 & 1638.00 & 736.98 & 30.05 & 9 & 8.06\\
        SC-$m$-16 & 1.48 & 0.03 & (1.0, 3e-05, 128) & 14 & 2540.57 & 666.06 & 25.57 & 17 & 14.64\\
        PSC & 2.38 & 0.01 & (1.0, 3e-05, 128) & 6 & 2888.73 & 840.06 & 11.46 & 35 & 26.67\\
        Envelope & 125.38 & 0.13 & (1.0, 3e-05, 128) & 4 & 3890.66 & 1086.80 & 9.77 & 58 & 43.00\\
        Map-construct & 22.77 & 3.73 & (1.0, 3e-05, 128) & 261 & 501.25 & 134.79 & 278.41 & 41 & 8.17\\
        \bottomrule
    \end{tabular}}
    \adjustbox{max width=\textwidth}
    {
    \begin{tabular}{@{}lrrcrrrrrr@{}}
        \toprule
        Name & Seconds & GBytes & $c$ & Clusters & Max Frechet & Avg Frechet & kMeans & Max Size & Avg Size\\
        \midrule
        SC-$\ell$-4 & 9.18 & 0.01 & (1.0, 0.0003, 128) & 18 & 204.72 & 79.56 & 107.25 & 100 & 32.61\\
        SC-$\ell$-8 & 6.18 & 0.01 & (1.0, 0.0003, 128) & 23 & 204.39 & 121.54 & 69.47 & 33 & 12.35\\
        SC-$\ell$-16 & 3.16 & 0.01 & (1.0, 0.0003, 128) & 13 & 407.84 & 174.47 & 63.94 & 27 & 10.69\\
        SC-$\ell$-32 & 1.76 & 0.02 & (1.0, 0.0003, 128) & 7 & 815.73 & 362.93 & 70.85 & 19 & 9.86\\
        SC-$\ell$-64 & 0.48 & 0.03 & (1.0, 0.0003, 128) & 0 & -inf & nan & 128.00 & -inf & nan\\
        SC-$m$-2 & 4.85 & 0.03 & (1.0, 0.0003, 128) & 39 & 816.49 & 293.29 & 54.50 & 3 & 2.72\\
        SC-$m$-4 & 3.17 & 0.03 & (1.0, 0.0003, 128) & 26 & 1323.31 & 298.50 & 42.90 & 5 & 4.27\\
        SC-$m$-8 & 2.35 & 0.03 & (1.0, 0.0003, 128) & 18 & 1194.66 & 350.08 & 47.04 & 9 & 6.94\\
        SC-$m$-16 & 2.23 & 0.03 & (1.0, 0.0003, 128) & 17 & 1634.40 & 399.85 & 54.76 & 17 & 7.12\\
        PSC & 4.93 & 0.01 & (1.0, 0.0003, 128) & 17 & 1779.26 & 306.56 & 36.66 & 24 & 7.18\\
        Envelope & 118.78 & 0.13 & (1.0, 0.0003, 128) & 12 & 1758.88 & 351.92 & 31.63 & 27 & 11.50\\
        Map-construct & 22.77 & 3.73 & (1.0, 0.0003, 128) & 261 & 501.25 & 134.79 & 349.52 & 41 & 8.17\\
        \bottomrule
    \end{tabular}}
    \adjustbox{max width=\textwidth}
    {
    \begin{tabular}{@{}lrrcrrrrrr@{}}
        \toprule
        Name & Seconds & GBytes & $c$ & Clusters & Max Frechet & Avg Frechet & kMeans & Max Size & Avg Size\\
        \midrule
        SC-$\ell$-4 & 1.86 & 0.00 & (1.0, 0.003, 128) & 1 & 23.21 & 18.60 & 127.61 & 14 & 14.00\\
        SC-$\ell$-8 & 6.66 & 0.01 & (1.0, 0.003, 128) & 8 & 99.68 & 52.72 & 124.75 & 11 & 6.75\\
        SC-$\ell$-16 & 5.45 & 0.01 & (1.0, 0.003, 128) & 16 & 200.97 & 103.75 & 106.81 & 11 & 5.12\\
        SC-$\ell$-32 & 3.58 & 0.01 & (1.0, 0.003, 128) & 10 & 334.83 & 136.42 & 104.49 & 7 & 2.80\\
        SC-$\ell$-64 & 0.48 & 0.03 & (1.0, 0.003, 128) & 0 & -inf & nan & 128.00 & -inf & nan\\
        SC-$m$-2 & 6.26 & 0.03 & (1.0, 0.003, 128) & 41 & 305.45 & 153.02 & 84.00 & 3 & 2.05\\
        SC-$m$-4 & 5.40 & 0.03 & (1.0, 0.003, 128) & 32 & 261.43 & 152.28 & 85.83 & 5 & 2.47\\
        SC-$m$-8 & 6.60 & 0.03 & (1.0, 0.003, 128) & 37 & 294.68 & 156.82 & 88.58 & 8 & 1.95\\
        SC-$m$-16 & 6.28 & 0.03 & (1.0, 0.003, 128) & 35 & 305.45 & 182.85 & 87.93 & 8 & 1.91\\
        PSC & 12.73 & 0.01 & (1.0, 0.003, 128) & 28 & 224.47 & 147.84 & 80.07 & 8 & 2.96\\
        Envelope & 106.36 & 0.13 & (1.0, 0.003, 128) & 22 & 815.91 & 182.00 & 82.98 & 12 & 4.55\\
        Map-construct & 22.77 & 3.73 & (1.0, 0.003, 128) & 261 & 501.25 & 134.79 & 1060.63 & 41 & 8.17\\
        \bottomrule
    \end{tabular}}
\end{minipage}
        \caption{ $\blacksquare$ $k$-means clustering $\blacksquare$ Athens-small  $\blacksquare$}
    \end{center} \end{figure}
\newpage 

\begin{figure*}[h] \begin{center}
        \begin{minipage}{\textwidth}
    \adjustbox{max width=\textwidth}
    {
    \begin{tabular}{@{}lrrcrrrrrr@{}}
        \toprule
        Name & Seconds & GBytes & $c$ & Clusters & Max Frechet & Avg Frechet & kCenters & Max Size & Avg Size\\
        \midrule
        SC-$\ell$-4 & 3.54 & 0.01 & (1.0, 0.003, 128) & 1 & 3163.53 & 1403.34 & 10.49 & 2836 & 2836.00\\
        SC-$\ell$-8 & 2.76 & 0.01 & (1.0, 0.003, 128) & 1 & 3163.53 & 1778.20 & 10.49 & 2832 & 2832.00\\
        SC-$\ell$-16 & 1.84 & 0.01 & (1.0, 0.003, 128) & 1 & 3137.06 & 2039.87 & 10.41 & 2824 & 2824.00\\
        SC-$\ell$-32 & 1.08 & 0.02 & (1.0, 0.003, 128 & 1 & 3259.12 & 2120.31 & 10.78 & 2806 & 2806.00\\
        SC-$\ell$-64 & 0.46 & 0.03 & (1.0, 0.003, 128) & 0 & -inf & nan & 128.00 & -inf & nan\\
        SC-$m$-2 & 4.72 & 0.01 & (1.0, 0.003, 128) & 4 & 1637.99 & 955.56 & 8.91 & 2497 & 703.75\\
        SC-$m$-4 & 3.62 & 0.01 & (1.0, 0.003, 128) & 1 & 3163.53 & 1403.34 & 10.49 & 2836 & 2836.00\\
        SC-$m$-8 & 2.65 & 0.01 & (1.0, 0.003, 128) & 1 & 3163.53 & 1778.20 & 10.49 & 2832 & 2832.00\\
        SC-$m$-16 & 1.87 & 0.01 & (1.0, 0.003, 128) & 1 & 3137.06 & 2039.87 & 10.41 & 2824 & 2824.00\\
        PSC & 3.26 & 0.01 & (1.0, 0.003, 128) & 3 & 1638.28 & 1007.44 & 7.91 & 2511 & 944.00\\
        Map-construct & 24.10 & 4.61 & (1.0, 0.03, 128) & 261 & 501.25 & 134.79 & 285.55 & 41 & 8.17\\
        \bottomrule
    \end{tabular}}

    \adjustbox{max width=\textwidth}
    {
    \begin{tabular}{@{}lrrcrrrrrr@{}}
        \toprule
        Name & Seconds & GBytes & $c$ & Clusters & Max Frechet & Avg Frechet & kCenters & Max Size & Avg Size\\
        \midrule
        SC-$\ell$-4 & 3.64 & 0.01 & (1.0, 0.03, 128) & 5 & 818.50 & 440.71 & 33.34 & 1165 & 538.60\\
        SC-$\ell$-8 & 2.76 & 0.01 & (1.0, 0.03, 128) & 5 & 819.20 & 534.19 & 34.98 & 1050 & 486.20\\
        SC-$\ell$-16 & 1.86 & 0.01 & (1.0, 0.03, 128) & 2 & 1637.51 & 1032.72 & 54.96 & 1814 & 1192.50\\
        SC-$\ell$-32 & 1.09 & 0.02 & (1.0, 0.03, 128) & 5 & 409.54 & 377.89 & 85.30 & 38 & 14.40\\
        SC-$\ell$-64 & 0.46 & 0.03 & (1.0, 0.03, 128) & 0 & -inf & nan & 128.00 & -inf & nan\\
        SC-$m$-2 & 4.83 & 0.01 & (1.0, 0.03, 128) & 7 & 818.69 & 509.62 & 32.01 & 1212 & 392.57\\
        SC-$m$-4 & 3.56 & 0.01 & (1.0, 0.03, 128) & 5 & 818.50 & 440.71 & 33.34 & 1165 & 538.60\\
        SC-$m$-8 & 2.78 & 0.01 & (1.0, 0.03, 128) & 5 & 819.20 & 534.19 & 34.98 & 1050 & 486.20\\
        SC-$m$-16 & 1.85 & 0.01 & (1.0, 0.03, 128) & 2 & 1637.51 & 1032.72 & 54.96 & 1814 & 1192.50\\
        PSC & 3.10 & 0.01 & (1.0, 0.03, 128 & 15 & 409.48 & 245.04 & 32.29 & 220 & 78.47\\
        Map-construct & 24.10 & 4.61 & (1.0, 0.3, 128) & 261 & 501.25 & 134.79 & 420.88 & 41 & 8.17\\
        \bottomrule
    \end{tabular}}

    \adjustbox{max width=\textwidth}
    {
    \begin{tabular}{@{}lrrcrrrrrr@{}}
        \toprule
        Name & Seconds & GBytes & $c$ & Clusters & Max Frechet & Avg Frechet & kCenters & Max Size & Avg Size\\
        \midrule
        SC-$\ell$-4 & 3.71 & 0.01 & (1.0, 0.3, 128) & 22 & 204.77 & 142.53 & 114.08 & 270 & 68.82\\
        SC-$\ell$-8 & 2.65 & 0.01 & (1.0, 0.3, 128) & 16 & 204.79 & 166.51 & 119.35 & 201 & 31.00\\
        SC-$\ell$-16 & 1.94 & 0.01 & (1.0, 0.3, 128) & 13 & 204.47 & 163.21 & 122.48 & 33 & 10.15\\
        SC-$\ell$-32 & 1.09 & 0.02 & (1.0, 0.3, 128) & 25 & 0.00 & nan & 115.82 & 1 & 1.00\\
        SC-$\ell$-64 & 0.49 & 0.03 & (1.0, 0.3, 128) & 0 & -inf & nan & 128.00 & -inf & nan\\
        SC-$m$-2 & 4.83 & 0.01 & (1.0, 0.3, 128) & 24 & 204.77 & 123.97 & 107.79 & 316 & 88.17\\
        SC-$m$-4 & 3.57 & 0.01 & (1.0, 0.3, 128) & 22 & 204.77 & 142.53 & 114.08 & 270 & 68.82\\
        SC-$m$-8 & 2.77 & 0.01 & (1.0, 0.3, 128) & 16 & 204.79 & 166.51 & 119.35 & 201 & 31.00\\
        SC-$m$-16 & 1.85 & 0.01 & (1.0, 0.3, 128) & 13 & 204.47 & 163.21 & 122.48 & 33 & 10.15\\
        PSC & 3.06 & 0.01 & (1.0, 0.3, 128) & 35 & 102.36 & 62.52 & 95.32 & 78 & 28.74\\
        Map-construct & 24.10 & 4.61 & (1.0, 3.0, 128) & 261 & 501.25 & 134.79 & 1774.26 & 41 & 8.17\\
        \bottomrule
    \end{tabular}}
\end{minipage}
        \caption{$\blacksquare$ $k$-centre clustering $\blacksquare$ Athens-small  $\blacksquare$}
    \end{center} \end{figure*}
\newpage 

\begin{figure} \begin{center}
    \begin{minipage}{\textwidth}
    \adjustbox{max width=\textwidth}
    {
    \begin{tabular}{@{}lrrcrrrrrr@{}}
        \toprule
        Name & Seconds & GBytes & $c$ & Clusters & Max Frechet & Avg Frechet & kMeans & Max Size & Avg Size\\
        \midrule
        SC-$\ell$-4 & --- & --- & --- & --- & --- & --- & --- & --- & ---\\
        SC-$\ell$-8 & --- & --- & --- & --- & --- & --- & --- & --- & ---\\
        SC-$\ell$-16 & 84686.70 & 0.20 & (1.0, 3e-05, 888) & 71 & 339.70 & 89.15 & 206.50 & 970 & 117.28\\
        SC-$\ell$-32 & --- & --- & --- & --- & --- & --- & --- & --- & ---\\
        SC-$\ell$-64 & --- & --- & --- & --- & --- & --- & --- & --- & ---\\
        SC-$m$-2 & --- & --- & --- & --- & --- & --- & --- & --- & ---\\
        SC-$m$-4 & 80331.20 & 5.29 & (1.0, 3e-05, 888) & 198 & 1343.00 & 228.94 & 227.66 & 5 & 4.59\\
        SC-$m$-8 & 49483.80 & 5.29 & (1.0, 3e-05, 888) & 124 & 1356.23 & 182.11 & 172.73 & 9 & 8.48\\
        SC-$m$-16 & 30217.30 & 5.29 & (1.0, 3e-05, 888) & 90 & 1356.41 & 184.20 & 138.25 & 17 & 15.59\\
        PSC & 3553.34 & 0.12 & (1.0, 3e-05, 888) & 23 & 1355.38 & 294.04 & 44.47 & 191 & 91.70\\
        Envelope & --- & --- & --- & --- & --- & --- & --- & --- & ---\\
        Map-construct & --- & --- & --- & --- & --- & --- & --- & --- & ---\\
        \bottomrule
    \end{tabular}
    }
    \adjustbox{max width=\textwidth}
    {
    \begin{tabular}{@{}lrrcrrrrrr@{}}
        \toprule
        Name & Seconds & GBytes & $c$ & Clusters & Max Frechet & Avg Frechet & kMeans & Max Size & Avg Size\\
        \midrule
        SC-$\ell$-4 & --- & --- & --- & --- & --- & --- & --- & --- & ---\\
        SC-$\ell$-8 & --- & --- & --- & --- & --- & --- & --- & --- & ---\\
        SC-$\ell$-16 & --- & --- & --- & --- & --- & --- & --- & --- & ---\\
        SC-$\ell$-32 & --- & --- & --- & --- & --- & --- & --- & --- & ---\\
        SC-$\ell$-64 & --- & --- & --- & --- & --- & --- & --- & --- & ---\\
        SC-$m$-2 & --- & --- & --- & --- & --- & --- & --- & --- & ---\\
        SC-$m$-4 & --- & --- & --- & --- & --- & --- & --- & --- & ---\\
        SC-$m$-8 & 51656.80 & 5.29 & (1.0, 0.0003, 888) & 139 & 1018.31 & 73.99 & 228.73 & 9 & 8.23\\
        SC-$m$-16 & 31688.40 & 5.29 & (1.0, 0.0003, 888) & 108 & 664.44 & 43.02 & 215.16 & 17 & 16.16\\
        PSC & 6381.29 & 0.13 & (1.0, 0.0003, 888) & 86 & 1478.85 & 38.53 & 152.27 & 175 & 36.24\\
        Envelope & --- & --- & --- & --- & --- & --- & --- & --- & ---\\
        Map-construct & --- & --- & --- & --- & --- & --- & --- & --- & ---\\
        \bottomrule
    \end{tabular}
    }
    \adjustbox{max width=\textwidth}
    {
    \begin{tabular}{@{}lrrcrrrrrr@{}}
        \toprule
        Name & Seconds & GBytes & $c$ & Clusters & Max Frechet & Avg Frechet & kMeans & Max Size & Avg Size\\
        \midrule
        SC-$\ell$-4 & --- & --- & --- & --- & --- & --- & --- & --- & ---\\
        SC-$\ell$-8 & --- & --- & --- & --- & --- & --- & --- & --- & ---\\
        SC-$\ell$-16 & --- & --- & --- & --- & --- & --- & --- & --- & ---\\
        SC-$\ell$-32 & --- & --- & --- & --- & --- & --- & --- & --- & ---\\
        SC-$\ell$-64 & --- & --- & --- & --- & --- & --- & --- & --- & ---\\
        SC-$m$-2 & --- & --- & --- & --- & --- & --- & --- & --- & ---\\
        SC-$m$-4 & --- & --- & --- & --- & --- & --- & --- & --- & ---\\
        SC-$m$-8 & 59279.00 & 5.29 & (1.0, 0.003, 888) & 138 & 84.69 & 22.71 & 338.81 & 9 & 8.18\\
        SC-$m$-16 & 38831.10 & 5.29 & (1.0, 0.003, 888) & 98 & 84.23 & 20.57 & 342.55 & 17 & 15.66\\
        PSC & 10375.70 & 0.13 & (1.0, 0.003, 888) & 96 & 84.90 & 20.40 & 308.53 & 66 & 17.05\\
        Envelope & --- & --- & --- & --- & --- & --- & --- & --- & ---\\
        Map-construct & --- & --- & --- & --- & --- & --- & --- & --- & ---\\
        \bottomrule
    \end{tabular}
    }
    \end{minipage}

        \caption{$\blacksquare$ $k$-means clustering $\blacksquare$ Chicago   $\blacksquare$}

    \end{center} \end{figure}
\newpage 

\begin{figure} \begin{center}
    \begin{minipage}{\textwidth}
    \adjustbox{max width=\textwidth}
    {
    \begin{tabular}{@{}lrrcrrrrrr@{}}
        \toprule
        Name & Seconds & GBytes & $c$ & Clusters & Max Frechet & Avg Frechet & kCenters & Max Size & Avg Size\\
        \midrule
        SC-$\ell$-4 & 22670.20 & 0.21 & (1.0, 10.0, 888) & 147 & 21.24 & 14.78 & 529.41 & 359 & 122.15\\
        SC-$\ell$-8 & 25517.70 & 0.28 & (1.0, 10.0, 888) & 96 & 21.24 & 15.36 & 504.91 & 327 & 105.07\\
        SC-$\ell$-16 & 26519.10 & 0.30 & (1.0, 10.0, 888) & 73 & 21.24 & 16.20 & 526.67 & 260 & 70.49\\
        SC-$\ell$-32 & 47077.50 & 0.48 & (1.0, 10.0, 888) & 59 & 21.24 & 16.68 & 556.16 & 218 & 42.02\\
        SC-$\ell$-64 & 75286.00 & 0.81 & (1.0, 10.0, 888) & 42 & 21.22 & 17.27 & 675.72 & 151 & 23.43\\
        SC-$m$-2 & 20086.90 & 0.17 & (1.0, 10.0, 888) & 218 & 21.24 & 13.90 & 586.90 & 373 & 140.00\\
        SC-$m$-4 & 23542.20 & 0.21 & (1.0, 10.0, 888) & 147 & 21.24 & 14.78 & 529.41 & 359 & 122.15\\
        SC-$m$-8 & 26235.70 & 0.28 & (1.0, 10.0, 888) & 96 & 21.24 & 15.36 & 504.91 & 327 & 105.07\\
        SC-$m$-16 & 24774.00 & 0.30 & (1.0, 10.0, 888) & 73 & 21.24 & 16.20 & 526.67 & 260 & 70.49\\
        PSC & 2987.91 & 0.10 & (1.0, 10.0, 888) & 95 & 21.23 & 15.90 & 397.87 & 218 & 50.75 \\
        \bottomrule
    \end{tabular}}
    \adjustbox{max width=\textwidth}
    {
    \begin{tabular}{@{}lrrcrrrrrr@{}}
        \toprule
        Name & Seconds & GBytes & $c$ & Clusters & Max Frechet & Avg Frechet & kCenters & Max Size & Avg Size\\
        \midrule
        SC-$\ell$-4 & 12192.90 & 0.21 & (1.0, 1.0, 888) & 95 & 84.97 & 63.44 & 231.57 & 3018 & 758.14\\
        SC-$\ell$-8 & 12092.50 & 0.28 & (1.0, 1.0, 888) & 86 & 84.97 & 71.45 & 243.16 & 1731 & 219.01\\
        SC-$\ell$-16 & 12774.90 & 0.30 & (1.0, 1.0, 888) & 59 & 84.97 & 72.65 & 244.42 & 367 & 99.17\\
        SC-$\ell$-32 & 24889.30 & 0.48 & (1.0, 1.0, 888) & 43 & 84.97 & 72.62 & 274.79 & 236 & 68.33\\
        SC-$\ell$-64 & 46779.30 & 0.81 & (1.0, 1.0, 888) & 36 & 42.47 & 34.08 & 444.17 & 156 & 30.00\\
        SC-$m$-2 & 12277.50 & 0.17 & (1.0, 1.0, 888) & 97 & 84.97 & 53.62 & 219.15 & 3725 & 981.74\\
        SC-$m$-4 & 13744.90 & 0.21 & (1.0, 1.0, 888) & 95 & 84.97 & 63.44 & 231.57 & 3018 & 758.14\\
        SC-$m$-8 & 13305.30 & 0.28 & (1.0, 1.0, 888) & 86 & 84.97 & 71.45 & 243.16 & 1731 & 219.01\\
        SC-$m$-16 & 12344.80 & 0.30 & (1.0, 1.0, 888) & 59 & 84.97 & 72.65 & 244.42 & 367 & 99.17\\
        PSC & 2767.08 & 0.10 & (1.0, 1.0, 888) & 75 & 42.49 & 31.83 & 167.75 & 191 & 47.93 \\
        \bottomrule
    \end{tabular}}
    \adjustbox{max width=\textwidth}
    {
    \begin{tabular}{@{}lrrcrrrrrr@{}}
        \toprule
        Name & Seconds & GBytes & $c$ & Clusters & Max Frechet & Avg Frechet & kCenters & Max Size & Avg Size\\
        \midrule
        SC-$\ell$-4 & 9884.09 & 0.21 & (1.0, 0.1, 888) & 24 & 339.89 & 215.40 & 60.17 & 21441 & 4553.62\\
        SC-$\ell$-8 & 10658.20 & 0.28 & (1.0, 0.1, 888) & 22 & 339.90 & 224.77 & 60.56 & 19986 & 4653.41\\
        SC-$\ell$-16 & 9925.72 & 0.30 & (1.0, 0.1, 888) & 23 & 339.89 & 259.59 & 77.27 & 17183 & 3878.83\\
        SC-$\ell$-32 & 21052.60 & 0.48 & (1.0, 0.1, 888) & 25 & 339.89 & 295.97 & 109.17 & 12010 & 1625.88\\
        SC-$\ell$-64 & 46322.50 & 0.81 & (1.0, 0.1, 888) & 3 & 1359.58 & 911.08 & 140.40 & 75493 & 33602.33\\
        SC-$m$-2 & 10749.60 & 0.17 & (1.0, 0.1, 888) & 21 & 339.90 & 207.19 & 57.98 & 22026 & 5411.86\\
        SC-$m$-4 & 11364.30 & 0.21 & (1.0, 0.1, 888) & 24 & 339.89 & 215.40 & 60.17 & 21441 & 4553.62\\
        SC-$m$-8 & 10211.10 & 0.28 & (1.0, 0.1, 888) & 22 & 339.90 & 224.77 & 60.56 & 19986 & 4653.41\\
        SC-$m$-16 & 10222.00 & 0.30 & (1.0, 0.1, 888) & 23 & 339.89 & 259.59 & 77.27 & 17183 & 3878.83\\
        PSC & 3035.34 & 0.10 & (1.0, 0.1, 888) & 25 & 339.89 & 242.66 & 61.35 & 6255 & 1404.08 \\
        \bottomrule
    \end{tabular}}
\end{minipage}
\end{center} 
        \caption{$\blacksquare$ $k$-centre clustering $\blacksquare$ Chicago   $\blacksquare$}
    \end{figure}
\newpage 

\begin{figure} \begin{center}
    \begin{minipage}{\textwidth}
    \adjustbox{max width=\textwidth}
    {
    \begin{tabular}{@{}lrrcrrrrrr@{}}
        \toprule
        Name & Seconds & GBytes & $c$ & Clusters & Max Frechet & Avg Frechet & kMeans & Max Size & Avg Size\\
        \midrule
        SC-$\ell$-4 & 2510.48 & 0.03 & (1.0, 3e-05, 222.0) & 57 & 339.82 & 83.16 & 127.77 & 752 & 261.81\\
        SC-$\ell$-8 & 3011.34 & 0.05 & (1.0, 3e-05, 222.0) & 63 & 1359.43 & 80.09 & 103.54 & 409 & 67.78\\
        SC-$\ell$-16 & 1767.85 & 0.05 & (1.0, 3e-05, 222.0) & 39 & 339.72 & 66.33 & 82.10 & 135 & 36.54\\
        SC-$\ell$-32 & 3226.07 & 0.10 & (1.0, 3e-05, 222.0) & 24 & 1358.95 & 202.53 & 64.15 & 95 & 34.92\\
        SC-$\ell$-64 & 6082.50 & 0.19 & (1.0, 3e-05, 222.0) & 17 & 1358.64 & 329.34 & 66.33 & 75 & 29.47\\
        SC-$m$-2 & 1790.72 & 0.69 & (1.0, 3e-05, 222.0) & 79 & 1334.43 & 354.28 & 92.39 & 3 & 2.73\\
        SC-$m$-4 & 1199.72 & 0.69 & (1.0, 3e-05, 222.0) & 51 & 1355.97 & 297.86 & 73.05 & 5 & 4.61\\
        SC-$m$-8 & 745.31 & 0.69 & (1.0, 3e-05, 222.0) & 35 & 1353.05 & 403.78 & 51.14 & 9 & 7.71\\
        SC-$m$-16 & 534.49 & 0.69 & (1.0, 3e-05, 222.0) & 26 & 1359.13 & 503.27 & 45.68 & 17 & 11.96\\
        PSC & 88.13 & 0.03 & (1.0, 3e-05, 222.0) & 7 & 2641.63 & 673.79 & 15.03 & 67 & 43.14\\
        Envelope & 11871.47 & 4.03 & (1.0, 3e-05, 222.0) & 6 & 2552.42 & 777.98 & 14.17 & 105 & 58.50\\
        Map-construct & 2972.20 & 9.21 & (1.0, 3e-05, 222.0) & 957 & 156.53 & 27.25 & 971.96 & 74 & 15.82\\
        \bottomrule
    \end{tabular}
    }
    \adjustbox{max width=\textwidth}
    {
    \begin{tabular}{@{}lrrcrrrrrr@{}}
        \toprule
        Name & Seconds & GBytes & $c$ & Clusters & Max Frechet & Avg Frechet & kMeans & Max Size & Avg Size\\
        \midrule
        SC-$\ell$-4 & 4487.10 & 0.03 & (1.0, 0.0003, 222) & 68 & 42.25 & 15.20 & 173.01 & 90 & 46.96\\
        SC-$\ell$-8 & 3339.71 & 0.05 & (1.0, 0.0003, 222) & 63 & 42.46 & 18.82 & 130.53 & 83 & 36.73\\
        SC-$\ell$-16 & 2209.07 & 0.05 & (1.0, 0.0003, 222) & 43 & 84.73 & 25.63 & 103.96 & 79 & 30.28\\
        SC-$\ell$-32 & 3632.63 & 0.10 & (1.0, 0.0003, 222) & 32 & 338.50 & 36.39 & 94.84 & 49 & 21.06\\
        SC-$\ell$-64 & 7949.50 & 0.19 & (1.0, 0.0003, 222) & 28 & 339.86 & 118.87 & 102.13 & 38 & 13.43\\
        SC-$m$-2 & 2075.07 & 0.69 & (1.0, 0.0003, 222) & 91 & 679.48 & 213.15 & 112.83 & 3 & 2.40\\
        SC-$m$-4 & 1312.14 & 0.69 & (1.0, 0.0003, 222) & 54 & 1359.07 & 131.40 & 87.47 & 5 & 4.39\\
        SC-$m$-8 & 948.43 & 0.69 & (1.0, 0.0003, 222) & 44 & 2421.63 & 108.58 & 82.98 & 9 & 7.86\\
        SC-$m$-16 & 640.77 & 0.69 & (1.0, 0.0003, 222) & 25 & 679.65 & 94.89 & 80.06 & 17 & 15.52\\
        PSC & 271.15 & 0.04 & (1.0, 0.0003, 222) & 38 & 1863.08 & 78.54 & 67.76 & 44 & 16.13\\
        Envelope & 12282.06 & 3.98 & (1.0, 0.0003, 222) & 19 & 2533.25 & 232.98 & 55.39 & 38 & 24.53\\
        Map-construct & 2972.20 & 9.21 & (1.0, 0.0003, 222) & 957 & 156.53 & 27.25 & 1080.57 & 74 & 15.82\\
        \bottomrule
    \end{tabular}
    }
    \adjustbox{max width=\textwidth}
    {
    \begin{tabular}{@{}lrrcrrrrrr@{}}
        \toprule
        Name & Seconds & GBytes & $c$ & Clusters & Max Frechet & Avg Frechet & kMeans & Max Size & Avg Size\\
        \midrule
        SC-$\ell$-4 & 173.23 & 0.03 & (1.0, 0.003, 222) & 0 & -inf & nan & 222.00 & -inf & nan\\
        SC-$\ell$-8 & 1664.45 & 0.05 & (1.0, 0.003, 222) & 7 & 21.23 & 14.60 & 217.77 & 74 & 54.29\\
        SC-$\ell$-16 & 2107.72 & 0.05 & (1.0, 0.003, 222) & 36 & 21.24 & 16.16 & 165.13 & 63 & 31.50\\
        SC-$\ell$-32 & 4187.72 & 0.10 & (1.0, 0.003, 222) & 28 & 42.32 & 17.69 & 133.72 & 49 & 20.96\\
        SC-$\ell$-64 & 6232.40 & 0.19 & (1.0, 0.003, 222) & 22 & 42.40 & 20.59 & 135.58 & 36 & 11.36\\
        SC-$m$-2 & 2396.72 & 0.69 & (1.0, 0.003, 222) & 96 & 108.68 & 30.94 & 137.59 & 3 & 2.12\\
        SC-$m$-4 & 1569.89 & 0.69 & (1.0, 0.003, 222) & 57 & 84.82 & 24.67 & 116.17 & 5 & 3.96\\
        SC-$m$-8 & 1062.67 & 0.69 & (1.0, 0.003, 222) & 38 & 42.45 & 22.12 & 117.21 & 9 & 7.45\\
        SC-$m$-16 & 790.13 & 0.69 & (1.0, 0.003, 222) & 26 & 41.52 & 19.39 & 121.90 & 17 & 12.73\\
        PSC & 405.62 & 0.04 & (1.0, 0.003, 222) & 38 & 83.71 & 21.15 & 108.81 & 38 & 9.05\\
        Envelope & 13172.97 & 3.97 & (1.0, 0.003, 222) & 34 & 316.20 & 31.26 & 99.36 & 24 & 10.68\\
        Map-construct & 2972.20 & 9.21 & (1.0, 0.003, 222) & 957 & 156.53 & 27.25 & 2166.58 & 74 & 15.82\\
        \bottomrule
    \end{tabular}
    }
\end{minipage}
        \caption{$\blacksquare$ $k$-means clustering $\blacksquare$ Chicago-4   $\blacksquare$}

    \end{center} \end{figure}
\newpage 

\begin{figure} \begin{center}
    \begin{minipage}{\textwidth}
    \adjustbox{max width=\textwidth}
    {
    \begin{tabular}{@{}lrrcrrrrrr@{}}
        \toprule
        Name & Seconds & GBytes & $c2$ & Clusters & Max Frechet & Avg Frechet & kCenters & Max Size & Avg Size\\
        \midrule
        SC-$\ell$-4 & 431.41 & 0.05 & (1.0, 0.1, 222) & 17 & 339.87 & 215.52 & 53.01 & 5453 & 1577.12\\
        SC-$\ell$-8 & 385.76 & 0.07 & (1.0, 0.1, 222) & 18 & 339.88 & 221.10 & 55.57 & 5086 & 1404.22\\
        SC-$\ell$-16 & 309.62 & 0.08 & (1.0, 0.1, 222) & 19 & 339.90 & 259.34 & 60.40 & 4352 & 1146.68\\
        SC-$\ell$-32 & 502.92 & 0.12 & (1.0, 0.1, 222) & 20 & 339.89 & 300.27 & 66.19 & 2006 & 396.90\\
        SC-$\ell$-64 & 1382.97 & 0.21 & (1.0, 0.1, 222) & 5 & 679.79 & 515.84 & 90.98 & 7280 & 2985.60\\
        SC-$m$-2 & 528.91 & 0.05 & (1.0, 0.1, 222) & 16 & 339.89 & 208.02 & 52.46 & 5602 & 1745.44\\
        SC-$m$-4 & 426.38 & 0.05 & (1.0, 0.1, 222) & 17 & 339.87 & 215.52 & 53.01 & 5453 & 1577.12\\
        SC-$m$-8 & 383.91 & 0.07 & (1.0, 0.1, 222) & 18 & 339.88 & 221.10 & 55.57 & 5086 & 1404.22\\
        SC-$m$-16 & 327.37 & 0.08 & (1.0, 0.1, 222) & 19 & 339.90 & 259.34 & 60.40 & 4352 & 1146.68\\
        PSC & 141.25 & 0.03 & (1.0, 0.1, 222) & 19 & 339.90 & 238.42 & 56.04 & 1542 & 496.32\\
        \bottomrule
    \end{tabular}}
    \adjustbox{max width=\textwidth}
    {
    \begin{tabular}{@{}lrrcrrrrrr@{}}
        \toprule
        Name & Seconds & GBytes & $c2$ & Clusters & Max Frechet & Avg Frechet & kCenters & Max Size & Avg Size\\
        \midrule
        SC-$\ell$-4 & 428.39 & None & (1.0, 1.0, 222) & 67 & 84.97 & 64.04 & 177.58 & 736 & 236.34\\
        SC-$\ell$-8 & 350.93 & 0.07 & (1.0, 1.0, 222) & 64 & 42.48 & 32.84 & 143.87 & 90 & 36.81\\
        SC-$\ell$-16 & 307.23 & 0.08 & (1.0, 1.0, 222) & 41 & 21.24 & 16.19 & 131.86 & 64 & 28.85\\
        SC-$\ell$-32 & 612.27 & 0.12 & (1.0, 1.0, 222) & 27 & 21.24 & 16.65 & 129.96 & 49 & 20.93\\
        SC-$\ell$-64 & 1791.53 & 0.21 & (1.0, 1.0, 222) & 20 & 21.24 & 17.48 & 152.83 & 36 & 11.30\\
        SC-$m$-2 & 532.90 & 0.04 & (1.0, 1.0, 222) & 67 & 84.97 & 54.20 & 177.56 & 928 & 313.94\\
        SC-$m$-4 & 431.68 & 0.05 & (1.0, 1.0, 222) & 67 & 84.97 & 64.04 & 177.58 & 736 & 236.34\\
        SC-$m$-8 & 361.35 & 0.07 & (1.0, 1.0, 222) & 64 & 42.48 & 32.84 & 143.87 & 90 & 36.81\\
        SC-$m$-16 & 399.70 & 0.08 & (1.0, 1.0, 222) & 41 & 21.24 & 16.19 & 131.86 & 64 & 28.85\\
        PSC & 136.59 & 0.03 & (1.0, 1.0, 222) & 41 & 21.24 & 16.67 & 108.44 & 44 & 19.61\\
        \bottomrule
    \end{tabular}}
    \adjustbox{max width=\textwidth}
    {
    \begin{tabular}{@{}lrrcrrrrrr@{}}
        \toprule
        Name & Seconds & GBytes & $c2$ & Clusters & Max Frechet & Avg Frechet & kCenters & Max Size & Avg Size\\
        \midrule
        SC-$\ell$-4 & 599.88 & 0.05 & (1.0, 10.0, 222) & 0 & -inf & nan & 222.00 & -inf & nan\\
        SC-$\ell$-8 & 726.98 & 0.07 & (1.0, 10.0, 222) & 0 & -inf & nan & 222.00 & -inf & nan\\
        SC-$\ell$-16 & 448.36 & 0.08 & (1.0, 10.0, 222) & 0 & -inf & nan & 222.00 & -inf & nan\\
        SC-$\ell$-32 & 1269.39 & 0.12 & (1.0, 10.0, 222) & 0 & -inf & nan & 222.00 & -inf & nan\\
        SC-$\ell$-64 & 2397.99 & 0.21 & (1.0, 10.0, 222) & 0 & -inf & nan & 222.00 & -inf & nan\\
        SC-$m$-2 & 560.50 & 0.04 & (1.0, 10.0, 222) & 0 & -inf & nan & 222.00 & -inf & nan\\
        SC-$m$-4 & 489.27 & 0.05 & (1.0, 10.0, 222) & 0 & -inf & nan & 222.00 & -inf & nan\\
        SC-$m$-8 & 590.26 & 0.07 & (1.0, 10.0, 222) & 0 & -inf & nan & 222.00 & -inf & nan\\
        SC-$m$-16 & 562.17 & 0.08 & (1.0, 10.0, 222) & 0 & -inf & nan & 222.00 & -inf & nan\\
        PSC & 146.09 & 0.03 & (1.0, 10.0, 222.0) & 76 & 0.00 & nan & 199.97 & 2 & 1.01\\
        \bottomrule
    \end{tabular}}
\end{minipage}
        \caption{$\blacksquare$ $k$-centre clustering $\blacksquare$ Chicago-4   $\blacksquare$}

    \end{center} \end{figure}
\newpage

    \begin{figure} \begin{center}
        \begin{minipage}{\textwidth}
    \adjustbox{max width=\textwidth}
    {
    \begin{tabular}{@{}lrrcrrrrrr@{}}
        \toprule
        Name & Seconds & GBytes & $c$ & Clusters & Max Frechet & Avg Frechet & kMeans & Max Size & Avg Size\\
        \midrule
                SC-$\ell$-4 & --- & --- & --- & --- & --- & --- & --- & --- & ---\\
SC-$\ell$-8 & --- & --- & --- & --- & --- & --- & --- & --- & ---\\
SC-$\ell$-16 & --- & --- & --- & --- & --- & --- & --- & --- & ---\\
SC-$\ell$-32& --- & --- & --- & --- & --- & --- & --- & --- & ---\\
SC-$\ell$-64 & --- & --- & --- & --- & --- & --- & --- & --- & ---\\
SC-$m$-2 & --- & --- & --- & --- & --- & --- & --- & --- & ---\\
SC-$m$-4 & --- & --- & --- & --- & --- & --- & --- & --- & ---\\
SC-$m$-8 & --- & --- & --- & --- & --- & --- & --- & --- & ---\\
SC-$m$-16 & --- & --- & --- & --- & --- & --- & --- & --- & ---\\
        PSC & 17949.40 & 2.16 & (1.0, 0.01, 2011) & 165 & 57.15 & 9.10 & 432.75 & 36 & 14.02\\
        Envelope & --- & --- & --- & --- & --- & --- & --- & --- & ---\\
Map-construct & --- & --- & --- & --- & --- & --- & --- & --- & ---\\
        \bottomrule
    \end{tabular}}
\end{minipage}
        \caption{$\blacksquare$ $k$-means clustering $\blacksquare$ Drifter $\blacksquare$}

    \end{center} \end{figure}

\newpage

    \begin{figure} \begin{center}
        \begin{minipage}{\textwidth}
    \adjustbox{max width=\textwidth}
    {
    \begin{tabular}{@{}lrrcrrrrrr@{}}
        \toprule
        Name & Seconds & GBytes & $c$ & Clusters & Max Frechet & Avg Frechet & kMeans & Max Size & Avg Size\\
        \midrule
                        SC-$\ell$-4 & --- & --- & --- & --- & --- & --- & --- & --- & ---\\
SC-$\ell$-8 & --- & --- & --- & --- & --- & --- & --- & --- & ---\\
SC-$\ell$-16 & --- & --- & --- & --- & --- & --- & --- & --- & ---\\
SC-$\ell$-32& --- & --- & --- & --- & --- & --- & --- & --- & ---\\
SC-$\ell$-64 & --- & --- & --- & --- & --- & --- & --- & --- & ---\\
SC-$m$-2 & --- & --- & --- & --- & --- & --- & --- & --- & ---\\
SC-$m$-4 & --- & --- & --- & --- & --- & --- & --- & --- & ---\\
SC-$m$-8 & --- & --- & --- & --- & --- & --- & --- & --- & ---\\
SC-$m$-16 & --- & --- & --- & --- & --- & --- & --- & --- & ---\\
        PSC & 78998.10 & 1.62 & (1.0, 10.0, 2011) & 73 & 16.38 & 10.77 & 166090432.35 & 157991 & 21132.48\\
PSC & 78657.30 & 1.62 & (1.0, 30.0, 2011) & 174 & 8.19 & 5.66 & 197774272.29 & 74926 & 6698.74\\
PSC & 85820.10 & 1.49 & (1.0, 100.0, 2011) & 398 & 4.10 & 2.91 & 241716795.73 & 28596 & 2089.96\\
PSC & 81131.30 & 1.62 & (1.0, 300.0, 2011) & 1073 & 0.00 & nan & 1460.38 & 1 & 1.00\\
        Envelope & --- & --- & --- & --- & --- & --- & --- & --- & ---\\
Map-construct & --- & --- & --- & --- & --- & --- & --- & --- & ---\\
        \bottomrule
    \end{tabular}
    }
\end{minipage}
\end{center}
        \caption{$\blacksquare$ $k$-centre clustering $\blacksquare$ Drifter $\blacksquare$}
 \end{figure}

\newpage

    \begin{figure} \begin{center}
        \begin{minipage}{\textwidth}
    \adjustbox{max width=\textwidth}
    {
    \begin{tabular}{@{}lrrcrrrrrr@{}}
        \toprule
        Name & Seconds & GBytes & $c2$ & Clusters & Max Frechet & Avg Frechet & kMeans & Max Size & Avg Size\\
        \midrule
        SC-$\ell$-4 & --- & --- & --- & --- & --- & --- & --- & --- & ---\\
SC-$\ell$-8 & --- & --- & --- & --- & --- & --- & --- & --- & ---\\
SC-$\ell$-16 & --- & --- & --- & --- & --- & --- & --- & --- & ---\\
SC-$\ell$-32& --- & --- & --- & --- & --- & --- & --- & --- & ---\\
SC-$\ell$-64 & --- & --- & --- & --- & --- & --- & --- & --- & ---\\
SC-$m$-2 & --- & --- & --- & --- & --- & --- & --- & --- & ---\\
SC-$m$-4 & --- & --- & --- & --- & --- & --- & --- & --- & ---\\
SC-$m$-8 & --- & --- & --- & --- & --- & --- & --- & --- & ---\\
SC-$m$-16 & --- & --- & --- & --- & --- & --- & --- & --- & ---\\
PSC & 287.47 & 0.30 & (1.0, 0.001, 362) & 10 & 104.73 & 19.12 & 19.15 & 79 & 46.40\\
Envelope & --- & --- & --- & --- & --- & --- & --- & --- & ---\\
Map-construct & --- & --- & --- & --- & --- & --- & --- & --- & ---\\
        \bottomrule
    \end{tabular}
    }
    \adjustbox{max width=\textwidth}
    {
    \begin{tabular}{@{}lrrcrrrrrr@{}}
        \toprule
        Name & Seconds & GBytes & $c2$ & Clusters & Max Frechet & Avg Frechet & kMeans & Max Size & Avg Size\\
        \midrule
        SC-$\ell$-4 & --- & --- & --- & --- & --- & --- & --- & --- & ---\\
SC-$\ell$-8 & --- & --- & --- & --- & --- & --- & --- & --- & ---\\
SC-$\ell$-16 & --- & --- & --- & --- & --- & --- & --- & --- & ---\\
SC-$\ell$-32 & --- & --- & --- & --- & --- & --- & --- & --- & ---\\
SC-$\ell$-64 & --- & --- & --- & --- & --- & --- & --- & --- & ---\\
SC-$m$-2 & --- & --- & --- & --- & --- & --- & --- & --- & ---\\
SC-$m$-4 & --- & --- & --- & --- & --- & --- & --- & --- & ---\\
SC-$m$-8 & --- & --- & --- & --- & --- & --- & --- & --- & ---\\
SC-$m$-16 & --- & --- & --- & --- & --- & --- & --- & --- & ---\\
PSC & 764.06 & 0.30 & (1.0, 0.01, 362) & 34 & 34.95 & 6.86 & 73.01 & 28 & 13.26\\
Envelope & --- & --- & --- & --- & --- & --- & --- & --- & ---\\
Map Construct & --- & --- & --- & --- & --- & --- & --- & --- & ---\\
        \bottomrule
    \end{tabular}
    }
    \adjustbox{max width=\textwidth}
    {
    \begin{tabular}{@{}lrrcrrrrrr@{}}
        \toprule
        Name & Seconds & GBytes & $c2$ & Clusters & Max Frechet & Avg Frechet & kMeans & Max Size & Avg Size\\
        \midrule
        SC-$\ell$-4 & --- & --- & --- & --- & --- & --- & --- & --- & ---\\
SC-$\ell$-8 & --- & --- & --- & --- & --- & --- & --- & --- & ---\\
SC-$\ell$-16 & --- & --- & --- & --- & --- & --- & --- & --- & ---\\
SC-$\ell$-32 & --- & --- & --- & --- & --- & --- & --- & --- & ---\\
SC-$\ell$-64 & --- & --- & --- & --- & --- & --- & --- & --- & ---\\
SC-$m$-2 & --- & --- & --- & --- & --- & --- & --- & --- & ---\\
SC-$m$-4 & --- & --- & --- & --- & --- & --- & --- & --- & ---\\
SC-$m$-8 & --- & --- & --- & --- & --- & --- & --- & --- & ---\\
SC-$m$-16 & --- & --- & --- & --- & --- & --- & --- & --- & ---\\
PSC & 2183.53 & 0.30 & (1.0, 0.1, 362) & 80 & 6.55 & 1.67 & 174.96 & 20 & 3.70\\
Envelope & --- & --- & --- & --- & --- & --- & --- & --- & ---\\
Map-construct & --- & --- & --- & --- & --- & --- & --- & --- & ---\\
        \bottomrule
    \end{tabular}
    }
\end{minipage}
\end{center}
        \caption{$\blacksquare$ $k$-means clustering $\blacksquare$ UniD $\blacksquare$}
 \end{figure}

\newpage 

\begin{figure} \begin{center}
    \begin{minipage}{\textwidth}
\adjustbox{max width=\textwidth}
{
\begin{tabular}{@{}lrrcrrrrrr@{}}
    \toprule
    Name & Seconds & GBytes & $c2$ & Clusters & Max Frechet & Avg Frechet & kCenters & Max Size & Avg Size\\
    \midrule
    SC-$\ell$-4 & 55043.80 & 0.25 & (1.0, 3.0, 362) & 8 & 13.11 & 8.10 & 49.04 & 52548 & 26513.50\\
    SC-$\ell$-8 & 63151.50 & 0.24 & (1.0, 3.0, 362) & 8 & 13.11 & 8.15 & 48.95 & 52240 & 26399.12\\
    SC-$\ell$-16 & 68155.60 & 0.42 & (1.0, 3.0, 362) & 9 & 13.11 & 8.21 & 49.69 & 51592 & 23293.78\\
    SC-$\ell$-32 & --- & --- & --- & --- & --- & --- & --- & --- & ---\\
    SC-$\ell$-64 & --- & --- & --- & --- & --- & --- & --- & --- & ---\\
    SC-$m$-4 & 67209.00 & 0.26 & (1.0, 3.0, 362) & 27 & 6.55 & 4.10 & 49.06 & 21096 & 7827.89\\
    SC-$m$-8 & 55189.10 & 0.24 & (1.0, 3.0, 362) & 8 & 13.11 & 8.10 & 49.04 & 52548 & 26513.50\\
    SC-$m$-16 & 59843.10 & 0.24 & (1.0, 3.0, 362) & 8 & 13.11 & 8.15 & 48.95 & 52240 & 26399.12\\
    SC-$m$-32 & 72703.90 & 0.42 & (1.0, 3.0, 362) & 9 & 13.11 & 8.21 & 49.69 & 51592 & 23293.78\\
    PSC & 1833.13 & 0.21 & (1.0, 3.0, 362) & 10 & 13.11 & 8.43 & 50.69 & 47466 & 18912.10\\
    Map-construct & --- & --- & --- & --- & --- & --- & --- & --- & ---\\
    \bottomrule
\end{tabular}
}
\adjustbox{max width=\textwidth}
{
\begin{tabular}{@{}lrrcrrrrrr@{}}
    \toprule
    Name & Seconds & GBytes & $c2$ & Clusters & Max Frechet & Avg Frechet & kCenters & Max Size & Avg Size\\
    \midrule
    SC-$\ell$-4 & 44324.10 & 0.47 & (1.0, 30.0, 362) & 68 & 3.28 & 2.11 & 185.63 & 9322 & 2911.51\\
    SC-$\ell$-8 & 37831.30 & 0.49 & (1.0, 30.0, 362) & 71 & 3.28 & 2.13 & 186.04 & 9261 & 2758.72\\
    SC-$\ell$-16 & 45932.90 & 0.66 & (1.0, 30.0, 362) & 70 & 3.28 & 2.14 & 187.94 & 9129 & 2704.01\\
    SC-$\ell$-32 & --- & --- & --- & --- & --- & --- & --- & --- & ---\\
    SC-$\ell$-64 & --- & --- & --- & --- & --- & --- & --- & --- & ---\\
    SC-$m$-2 & 46785.30 & 0.40 & (1.0, 30.0, 362) & 72 & 3.28 & 2.11 & 186.06 & 9353 & 2807.14\\
    SC-$m$-4 & 44663.20 & 0.47 & (1.0, 30.0, 362) & 68 & 3.28 & 2.11 & 185.63 & 9322 & 2911.51\\
    SC-$m$-8 & 39783.20 & 0.49 & (1.0, 30.0, 362) & 71 & 3.28 & 2.13 & 186.04 & 9261 & 2758.72\\
    SC-$m$-16 & 46280.00 & 0.66 & (1.0, 30.0, 362) & 70 & 3.28 & 2.14 & 187.94 & 9129 & 2704.01\\
    PSC & 1811.57 & 0.21 & (1.0, 30.0, 362) & 61 & 3.28 & 2.22 & 186.02 & 8092 & 794.84\\
    Map-construct & --- & --- & --- & --- & --- & --- & --- & --- & ---\\
    \bottomrule
\end{tabular}
}
\adjustbox{max width=\textwidth}
{
\begin{tabular}{@{}lrrcrrrrrr@{}}
    \toprule
    Name & Seconds & GBytes & $c2$ & Clusters & Max Frechet & Avg Frechet & kCenters & Max Size & Avg Size\\
    \midrule
    SC-$\ell$-4 & 53806.00 & 0.31 & (1.0, 300.0, 362) & 10 & 0.00 & 0.00 & 353.41 & 2071 & 1148.10\\
    SC-$\ell$-8 & 59203.00 & 0.25 & (1.0, 300.0, 362) & 10 & 0.00 & 0.00 & 353.41 & 2032 & 1134.10\\
    SC-$\ell$-16 & 65860.00 & 0.41 & (1.0, 300.0, 362) & 8 & 0.00 & 0.00 & 353.55 & 1996 & 1207.50\\
    SC-$\ell$-32 & --- & --- & --- & --- & --- & --- & --- & --- & ---\\
    SC-$\ell$-64 & --- & --- & --- & --- & --- & --- & --- & --- & ---\\
    SC-$m$-2 & 56941.30 & 0.28 & (1.0, 300.0, 362) & 10 & 0.00 & 0.00 & 353.41 & 2082 & 1153.40\\
    SC-$m$-4 & 53550.40 & 0.30 & (1.0, 300.0, 362) & 10 & 0.00 & 0.00 & 353.41 & 2071 & 1148.10\\
    SC-$m$-8 & 52683.50 & 0.24 & (1.0, 300.0, 362) & 10 & 0.00 & 0.00 & 353.41 & 2032 & 1134.10\\
    SC-$m$-16 & 67746.20 & 0.43 & (1.0, 300.0, 362) & 8 & 0.00 & 0.00 & 353.55 & 1996 & 1207.50\\
    PSC & 1686.52 & 0.21 & (1.0, 300.0, 362) & 135 & 0.00 & nan & 273.17 & 1 & 1.00\\
    Map-construct & --- & --- & --- & --- & --- & --- & --- & --- & ---\\
    \bottomrule
\end{tabular}
}
\end{minipage}
\end{center}
    \caption{$\blacksquare$ $k$-centre clustering $\blacksquare$ UniD $\blacksquare$}
\end{figure}
\newpage 

\begin{figure} \begin{center}
    \begin{minipage}{\textwidth}
    \adjustbox{max width=\textwidth}
    {
    \begin{tabular}{@{}lrrcrrrrrr@{}}
        \toprule
        Name & Seconds & GBytes & $c$ & Clusters & Max Frechet & Avg Frechet & kMeans & Max Size & Avg Size\\
        \midrule
        SC-$\ell$-4 & 3.08 & 0.01 & (1.0, 0.001, 50) & 0 & -inf & nan & 50.00 & -inf & nan\\
        SC-$\ell$-8 & 3.06 & 0.01 & (1.0, 0.001, 50) & 0 & -inf & nan & 50.00 & -inf & nan\\
        SC-$\ell$-16 & 3.22 & 0.02 & (1.0, 0.001, 50) & 0 & -inf & nan & 50.00 & -inf & nan\\
        SC-$\ell$-32 & 4.14 & 0.03 & (1.0, 0.001, 50) & 0 & -inf & nan & 50.00 & -inf & nan\\
        SC-$\ell$-64 & 9.75 & 0.05 & (1.0, 0.001, 50) & 1 & 518.10 & 425.55 & 49.80 & 4 & 4.00\\
        SC-$m$-2 & 18.13 & 0.05 & (1.0, 0.001, 50) & 35 & 951.62 & 720.54 & 47.58 & 2 & 1.43\\
        SC-$m$-4 & 16.68 & 0.05 & (1.0, 0.001, 50) & 34 & 936.01 & 673.83 & 48.27 & 4 & 1.47\\
        SC-$m$-8 & 17.53 & 0.05 & (1.0, 0.001, 50) & 36 & 892.24 & 716.90 & 47.60 & 4 & 1.39\\
        SC-$m$-16 & 29.69 & 0.05 & (1.0, 0.001, 50) & 47 & 882.42 & 851.20 & 49.79 & 2 & 1.04\\
        PSC & 52.54 & 0.02 & (1.0, 0.001, 50) & 14 & 913.07 & 721.77 & 46.79 & 4 & 2.29\\
        Envelope & 6640.05 & 0.23 & (1.0, 0.001, 50) & 6 & 1778.38 & 1074.17 & 56.45 & 16 & 8.33\\
        Map-construct & 55.43 & 1.89 & (1.0, 0.001, 50) & 1431 & 817.76 & 36.36 & 1654.28 & 18 & 5.13\\
        \bottomrule
    \end{tabular}
    }
    \adjustbox{max width=\textwidth}
    {
    \begin{tabular}{@{}lrrcrrrrrr@{}}
        \toprule
        Name & Seconds & GBytes & $c$ & Clusters & Max Frechet & Avg Frechet & kMeans & Max Size & Avg Size\\
        \midrule
        SC-$\ell$-4 & 3.18 & 0.01 & (1.0, 0.0003, 50) & 0 & -inf & nan & 50.00 & -inf & nan\\
        SC-$\ell$-8 & 3.17 & 0.01 & (1.0, 0.0003, 50) & 0 & -inf & nan & 50.00 & -inf & nan\\
        SC-$\ell$-16 & 3.09 & 0.02 & (1.0, 0.0003, 50) & 0 & -inf & nan & 50.00 & -inf & nan\\
        SC-$\ell$-32 & 14.24 & 0.03 & (1.0, 0.0003, 50) & 5 & 535.35 & 494.06 & 47.23 & 17 & 13.60\\
        SC-$\ell$-64 & 21.25 & 0.05 & (1.0, 0.0003, 50) & 12 & 1070.33 & 836.97 & 42.11 & 9 & 5.33\\
        SC-$m$-2 & 12.21 & 0.05 & (1.0, 0.0003, 50) & 21 & 1063.61 & 932.75 & 33.11 & 3 & 2.38\\
        SC-$m$-4 & 8.56 & 0.05 & (1.0, 0.0003, 50) & 15 & 1068.51 & 937.33 & 30.75 & 5 & 3.33\\
        SC-$m$-8 & 6.75 & 0.05 & (1.0, 0.0003, 50) & 10 & 1063.83 & 875.74 & 29.73 & 8 & 5.00\\
        SC-$m$-16 & 6.81 & 0.05 & (1.0, 0.0003, 50) & 12 & 1071.65 & 998.84 & 36.50 & 10 & 4.25\\
        PSC & 14.48 & 0.02 & (1.0, 0.0003, 50) & 7 & 2143.57 & 985.30 & 22.11 & 14 & 7.29\\
        Envelope & 6486.82 & 0.23 & (1.0, 0.0003, 50) & 3 & 1924.87 & 1297.20 & 25.08 & 40 & 17.67\\
        Map-construct & 55.43 & 1.89 & (1.0, 0.0003, 50) & 1431 & 817.76 & 36.36 & 1501.10 & 18 & 5.13\\
        \bottomrule
    \end{tabular}
    }
    \adjustbox{max width=\textwidth}
    {
    \begin{tabular}{@{}lrrcrrrrrr@{}}
        \toprule
        Name & Seconds & GBytes & $c$ & Clusters & Max Frechet & Avg Frechet & kMeans & Max Size & Avg Size\\
        \midrule
        SC-$\ell$-4 & 3.24 & 0.01 & (1.0, 0.0001, 50) & 0 & -inf & nan & 50.00 & -inf & nan\\
        SC-$\ell$-8 & 3.29 & 0.01 & (1.0, 0.0001, 50) & 0 & -inf & nan & 50.00 & -inf & nan\\
        SC-$\ell$-16 & 13.92 & 0.02 & (1.0, 0.0001, 50) & 5 & 536.11 & 404.39 & 47.66 & 36 & 23.20\\
        SC-$\ell$-32 & 13.89 & 0.03 & (1.0, 0.0001, 50) & 8 & 1061.05 & 576.50 & 38.16 & 23 & 16.00\\
        SC-$\ell$-64 & 16.02 & 0.05 & (1.0, 0.0001, 50) & 8 & 1069.96 & 958.22 & 30.63 & 19 & 10.25\\
        SC-$m$-2 & 11.24 & 0.05 & (1.0, 0.0001, 50) & 20 & 2077.13 & 1009.06 & 26.63 & 3 & 2.65\\
        SC-$m$-4 & 7.85 & 0.05 & (1.0, 0.0001, 50) & 13 & 1816.18 & 990.90 & 22.03 & 5 & 4.31\\
        SC-$m$-8 & 5.86 & 0.05 & (1.0, 0.0001, 50) & 8 & 1065.53 & 922.80 & 21.57 & 9 & 7.88\\
        SC-$m$-16 & 4.23 & 0.05 & (1.0, 0.0001, 50) & 5 & 1862.33 & 1001.98 & 23.97 & 16 & 11.20\\
        PSC & 12.85 & 0.02 & (1.0, 0.0001, 50) & 4 & 1809.92 & 1149.63 & 10.90 & 17 & 12.75\\
        Envelope & 6689.78 & 0.23 & (1.0, 0.0001, 50) & 3 & 1949.27 & 1359.89 & 11.66 & 36 & 20.67\\
        Map-construct & 55.43 & 1.89 & (1.0, 0.0001, 50) & 1431 & 817.76 & 36.36 & 1457.34 & 18 & 5.13\\
        \bottomrule
    \end{tabular}
    }
\end{minipage}
\end{center}
        \caption{$\blacksquare$ $k$-means clustering $\blacksquare$ Berlin-10 $\blacksquare$}
 \end{figure}
\newpage

\begin{figure} \begin{center}
    \begin{minipage}{\textwidth}
    \adjustbox{max width=\textwidth}
    {
    \begin{tabular}{@{}lrrcrrrrrr@{}}
        \toprule
        Name & Seconds & GBytes & $c$ & Clusters & Max Frechet & Avg Frechet & kMeans & Max Size & Avg Size\\
        \midrule
        SC-$\ell$-4 & --- & --- & --- & --- & --- & --- & --- & --- & ---\\
        SC-$\ell$-8 & --- & --- & --- & --- & --- & --- & --- & --- & ---\\
        SC-$\ell$-16 & --- & --- & --- & --- & --- & --- & --- & --- & ---\\
        SC-$\ell$-32 & --- & --- & --- & --- & --- & --- & --- & --- & ---\\
        SC-$\ell$-64 & --- & --- & --- & --- & --- & --- & --- & --- & ---\\
        SC-$m$-2 & --- & --- & --- & --- & --- & --- & --- & --- & ---\\
        SC-$m$-4 & --- & --- & --- & --- & --- & --- & --- & --- & ---\\
        SC-$m$-8 & --- & --- & --- & --- & --- & --- & --- & --- & ---\\
        SC-$m$-16 & --- & --- & --- & --- & --- & --- & --- & --- & ---\\
        PSC & --- & --- & --- & --- & --- & --- & --- & --- & ---\\
        Envelope & --- & --- & --- & --- & --- & --- & --- & --- & ---\\
        Map-construct & --- & --- & --- & --- & --- & --- & --- & --- & ---\\
        \bottomrule
    \end{tabular}
    }
    \adjustbox{max width=\textwidth}
    {
    \begin{tabular}{@{}lrrcrrrrrr@{}}
        \toprule
        Name & Seconds & GBytes & $c$ & Clusters & Max Frechet & Avg Frechet & kMeans & Max Size & Avg Size\\
        \midrule
        SC-$\ell$-4 & --- & --- & --- & --- & --- & --- & --- & --- & ---\\
        SC-$\ell$-8 & --- & --- & --- & --- & --- & --- & --- & --- & ---\\
        SC-$\ell$-16 & --- & --- & --- & --- & --- & --- & --- & --- & ---\\
        SC-$\ell$-32 & --- & --- & --- & --- & --- & --- & --- & --- & ---\\
        SC-$\ell$-64 & --- & --- & --- & --- & --- & --- & --- & --- & ---\\
        SC-$m$-2 & --- & --- & --- & --- & --- & --- & --- & --- & ---\\
        SC-$m$-4 & --- & --- & --- & --- & --- & --- & --- & --- & ---\\
        SC-$m$-8 & --- & --- & --- & --- & --- & --- & --- & --- & ---\\
        SC-$m$-16 & --- & --- & --- & --- & --- & --- & --- & --- & ---\\
        PSC & --- & --- & --- & --- & --- & --- & --- & --- & ---\\
        Envelope & --- & --- & --- & --- & --- & --- & --- & --- & ---\\
        Map-construct & --- & --- & --- & --- & --- & --- & --- & --- & ---\\
        \bottomrule
    \end{tabular}
    }
    \adjustbox{max width=\textwidth}
    {
    \begin{tabular}{@{}lrrcrrrrrr@{}}
        \toprule
        Name & Seconds & GBytes & $c$ & Clusters & Max Frechet & Avg Frechet & kMeans & Max Size & Avg Size\\
        \midrule
        SC-$\ell$-4 & --- & --- & --- & --- & --- & --- & --- & --- & ---\\
        SC-$\ell$-8 & --- & --- & --- & --- & --- & --- & --- & --- & ---\\
        SC-$\ell$-16 & --- & --- & --- & --- & --- & --- & --- & --- & ---\\
        SC-$\ell$-32 & --- & --- & --- & --- & --- & --- & --- & --- & ---\\
        SC-$\ell$-64 & --- & --- & --- & --- & --- & --- & --- & --- & ---\\
        SC-$m$-2 & --- & --- & --- & --- & --- & --- & --- & --- & ---\\
        SC-$m$-4 & --- & --- & --- & --- & --- & --- & --- & --- & ---\\
        SC-$m$-8 & --- & --- & --- & --- & --- & --- & --- & --- & ---\\
        SC-$m$-16 &--- & --- & --- & --- & --- & --- & --- & --- & ---\\
        PSC & --- & --- & --- & --- & --- & --- & --- & --- & ---\\
        Envelope & --- & --- & --- & --- & --- & --- & --- & --- & ---\\
        Map-construct & --- & --- & --- & --- & --- & --- & --- & --- & ---\\
        \bottomrule
    \end{tabular}
    }
    \end{minipage}

        \caption{$\blacksquare$ $k$-means clustering $\blacksquare$ Berlin   $\blacksquare$}
    \end{center} \end{figure}
\newpage

\begin{figure} \begin{center}
    \begin{minipage}{\textwidth}
    \adjustbox{max width=\textwidth}
    {
    \begin{tabular}{@{}lrrcrrrrrr@{}}
        \toprule
        Name & Seconds & GBytes & $c$ & Clusters & Max Frechet & Avg Frechet & kMeans & Max Size & Avg Size\\
        \midrule
        SC-$\ell$-4 & --- & --- & --- & --- & --- & --- & --- & --- & ---\\
        SC-$\ell$-8 & --- & --- & --- & --- & --- & --- & --- & --- & ---\\
        SC-$\ell$-16 & --- & --- & --- & --- & --- & --- & --- & --- & ---\\
        SC-$\ell$-32 & --- & --- & --- & --- & --- & --- & --- & --- & ---\\
        SC-$\ell$-64 & --- & --- & --- & --- & --- & --- & --- & --- & ---\\
        SC-$m$-2 & --- & --- & --- & --- & --- & --- & --- & --- & ---\\
        SC-$m$-4 & --- & --- & --- & --- & --- & --- & --- & --- & ---\\
        SC-$m$-8 & --- & --- & --- & --- & --- & --- & --- & --- & ---\\
        SC-$m$-16 & --- & --- & --- & --- & --- & --- & --- & --- & ---\\
        PSC & --- & --- & --- & --- & --- & --- & --- & --- & ---\\
        Envelope & --- & --- & --- & --- & --- & --- & --- & --- & ---\\
        Map-construct & --- & --- & --- & --- & --- & --- & --- & --- & ---\\
        \bottomrule
    \end{tabular}
    }
    \adjustbox{max width=\textwidth}
    {
    \begin{tabular}{@{}lrrcrrrrrr@{}}
        \toprule
        Name & Seconds & GBytes & $c$ & Clusters & Max Frechet & Avg Frechet & kMeans & Max Size & Avg Size\\
        \midrule
        SC-$\ell$-4 & --- & --- & --- & --- & --- & --- & --- & --- & ---\\
        SC-$\ell$-8 & --- & --- & --- & --- & --- & --- & --- & --- & ---\\
        SC-$\ell$-16 & --- & --- & --- & --- & --- & --- & --- & --- & ---\\
        SC-$\ell$-32 & --- & --- & --- & --- & --- & --- & --- & --- & ---\\
        SC-$\ell$-64 & --- & --- & --- & --- & --- & --- & --- & --- & ---\\
        SC-$m$-2 & --- & --- & --- & --- & --- & --- & --- & --- & ---\\
        SC-$m$-4 & --- & --- & --- & --- & --- & --- & --- & --- & ---\\
        SC-$m$-8 & --- & --- & --- & --- & --- & --- & --- & --- & ---\\
        SC-$m$-16 & --- & --- & --- & --- & --- & --- & --- & --- & ---\\
        PSC & --- & --- & --- & --- & --- & --- & --- & --- & ---\\
        Envelope & --- & --- & --- & --- & --- & --- & --- & --- & ---\\
        Map-construct & --- & --- & --- & --- & --- & --- & --- & --- & ---\\
        \bottomrule
    \end{tabular}
    }
    \adjustbox{max width=\textwidth}
    {
    \begin{tabular}{@{}lrrcrrrrrr@{}}
        \toprule
        Name & Seconds & GBytes & $c$ & Clusters & Max Frechet & Avg Frechet & kMeans & Max Size & Avg Size\\
        \midrule
        SC-$\ell$-4 & --- & --- & --- & --- & --- & --- & --- & --- & ---\\
        SC-$\ell$-8 & --- & --- & --- & --- & --- & --- & --- & --- & ---\\
        SC-$\ell$-16 & --- & --- & --- & --- & --- & --- & --- & --- & ---\\
        SC-$\ell$-32 & --- & --- & --- & --- & --- & --- & --- & --- & ---\\
        SC-$\ell$-64 & --- & --- & --- & --- & --- & --- & --- & --- & ---\\
        SC-$m$-2 & --- & --- & --- & --- & --- & --- & --- & --- & ---\\
        SC-$m$-4 & --- & --- & --- & --- & --- & --- & --- & --- & ---\\
        SC-$m$-8 & --- & --- & --- & --- & --- & --- & --- & --- & ---\\
        SC-$m$-16 &--- & --- & --- & --- & --- & --- & --- & --- & ---\\
        PSC & --- & --- & --- & --- & --- & --- & --- & --- & ---\\
        Envelope & --- & --- & --- & --- & --- & --- & --- & --- & ---\\
        Map-construct & --- & --- & --- & --- & --- & --- & --- & --- & ---\\
        \bottomrule
    \end{tabular}
    }
    \end{minipage}

        \caption{$\blacksquare$ $k$-centre clustering $\blacksquare$ Berlin   $\blacksquare$}
    \end{center} \end{figure}
\newpage 

\begin{figure} \begin{center}
    \begin{minipage}{\textwidth}
    \adjustbox{max width=\textwidth}
    {
    \begin{tabular}{@{}lrrcrrrrrr@{}}
        \toprule
        Name & Seconds & GBytes & $c$ & Clusters & Max Frechet & Avg Frechet & kMeans & Max Size & Avg Size\\
        \midrule
        SC-$\ell$-4 & 3.08 & 0.01 & (1.0, 0.001, 50) & 0 & -inf & nan & 50.00 & -inf & nan\\
        SC-$\ell$-8 & 3.06 & 0.01 & (1.0, 0.001, 50) & 0 & -inf & nan & 50.00 & -inf & nan\\
        SC-$\ell$-16 & 3.22 & 0.02 & (1.0, 0.001, 50) & 0 & -inf & nan & 50.00 & -inf & nan\\
        SC-$\ell$-32 & 4.14 & 0.03 & (1.0, 0.001, 50) & 0 & -inf & nan & 50.00 & -inf & nan\\
        SC-$\ell$-64 & 9.75 & 0.05 & (1.0, 0.001, 50) & 1 & 518.10 & 425.55 & 49.80 & 4 & 4.00\\
        SC-$m$-2 & 18.13 & 0.05 & (1.0, 0.001, 50) & 35 & 951.62 & 720.54 & 47.58 & 2 & 1.43\\
        SC-$m$-4 & 16.68 & 0.05 & (1.0, 0.001, 50) & 34 & 936.01 & 673.83 & 48.27 & 4 & 1.47\\
        SC-$m$-8 & 17.53 & 0.05 & (1.0, 0.001, 50) & 36 & 892.24 & 716.90 & 47.60 & 4 & 1.39\\
        SC-$m$-16 & 29.69 & 0.05 & (1.0, 0.001, 50) & 47 & 882.42 & 851.20 & 49.79 & 2 & 1.04\\
        PSC & 52.54 & 0.02 & (1.0, 0.001, 50) & 14 & 913.07 & 721.77 & 46.79 & 4 & 2.29\\
        Envelope & 6640.05 & 0.23 & (1.0, 0.001, 50) & 6 & 1778.38 & 1074.17 & 56.45 & 16 & 8.33\\
        Map-construct & 55.43 & 1.89 & (1.0, 0.001, 50) & 1431 & 817.76 & 36.36 & 1651.56 & 18 & 5.13\\
        \bottomrule
    \end{tabular}
    }
    \adjustbox{max width=\textwidth}
    {
    \begin{tabular}{@{}lrrcrrrrrr@{}}
        \toprule
        Name & Seconds & GBytes & $c$ & Clusters & Max Frechet & Avg Frechet & kMeans & Max Size & Avg Size\\
        \midrule
        SC-$\ell$-4 & 3.18 & 0.01 & (1.0, 0.0003, 50) & 0 & -inf & nan & 50.00 & -inf & nan\\
        SC-$\ell$-8 & 3.17 & 0.01 & (1.0, 0.0003, 50) & 0 & -inf & nan & 50.00 & -inf & nan\\
        SC-$\ell$-16 & 3.09 & 0.02 & (1.0, 0.0003, 50) & 0 & -inf & nan & 50.00 & -inf & nan\\
        SC-$\ell$-32 & 14.24 & 0.03 & (1.0, 0.0003, 50) & 5 & 535.35 & 494.06 & 47.23 & 17 & 13.60\\
        SC-$\ell$-64 & 21.25 & 0.05 & (1.0, 0.0003, 50) & 12 & 1070.33 & 836.97 & 42.11 & 9 & 5.33\\
        SC-$m$-2 & 12.21 & 0.05 & (1.0, 0.0003, 50) & 21 & 1063.61 & 932.75 & 33.11 & 3 & 2.38\\
        SC-$m$-4 & 8.56 & 0.05 & (1.0, 0.0003, 50) & 15 & 1068.51 & 937.33 & 30.75 & 5 & 3.33\\
        SC-$m$-8 & 6.75 & 0.05 & (1.0, 0.0003, 50) & 10 & 1063.83 & 875.74 & 29.73 & 8 & 5.00\\
        SC-$m$-16 & 6.81 & 0.05 & (1.0, 0.0003, 50) & 12 & 1071.65 & 998.84 & 36.50 & 10 & 4.25\\
        PSC & 14.48 & 0.02 & (1.0, 0.0003, 50) & 7 & 2143.57 & 985.30 & 22.11 & 14 & 7.29\\
        Envelope & 6486.82 & 0.23 & (1.0, 0.0003, 50) & 3 & 1924.87 & 1297.20 & 25.08 & 40 & 17.67\\
        Map-construct & 55.43 & 1.89 & (1.0, 0.0003, 50) & 1431 & 817.76 & 36.36 & 1498.39 & 18 & 5.13\\
        \bottomrule
    \end{tabular}
    }
    \adjustbox{max width=\textwidth}
    {
    \begin{tabular}{@{}lrrcrrrrrr@{}}
        \toprule
        Name & Seconds & GBytes & $c$ & Clusters & Max Frechet & Avg Frechet & kMeans & Max Size & Avg Size\\
        \midrule
        SC-$\ell$-4 & 3.24 & 0.01 & (1.0, 0.0001, 50) & 0 & -inf & nan & 50.00 & -inf & nan\\
        SC-$\ell$-8 & 3.29 & 0.01 & (1.0, 0.0001, 50) & 0 & -inf & nan & 50.00 & -inf & nan\\
        SC-$\ell$-16 & 13.92 & 0.02 & (1.0, 0.0001, 50) & 5 & 536.11 & 404.39 & 47.66 & 36 & 23.20\\
        SC-$\ell$-32 & 13.89 & 0.03 & (1.0, 0.0001, 50) & 8 & 1061.05 & 576.50 & 38.16 & 23 & 16.00\\
        SC-$\ell$-64 & 16.02 & 0.05 & (1.0, 0.0001, 50) & 8 & 1069.96 & 958.22 & 30.63 & 19 & 10.25\\
        SC-$m$-2 & 11.24 & 0.05 & (1.0, 0.0001, 50) & 20 & 2077.13 & 1009.06 & 26.63 & 3 & 2.65\\
        SC-$m$-4 & 7.85 & 0.05 & (1.0, 0.0001, 50) & 13 & 1816.18 & 990.90 & 22.03 & 5 & 4.31\\
        SC-$m$-8 & 5.86 & 0.05 & (1.0, 0.0001, 50) & 8 & 1065.53 & 922.80 & 21.57 & 9 & 7.88\\
        SC-$m$-16 & 4.23 & 0.05 & (1.0, 0.0001, 50) & 5 & 1862.33 & 1001.98 & 23.97 & 16 & 11.20\\
        PSC & 12.85 & 0.02 & (1.0, 0.0001, 50) & 4 & 1809.92 & 1149.63 & 10.90 & 17 & 12.75\\
        Envelope & 6689.78 & 0.23 & (1.0, 0.0001, 50) & 3 & 1949.27 & 1359.89 & 11.66 & 36 & 20.67\\
        Map-construct & 55.43 & 1.89 & (1.0, 0.0001, 50) & 1431 & 817.76 & 36.36 & 1454.62 & 18 & 5.13\\
        \bottomrule
    \end{tabular}
    }
  \end{minipage}
  \end{center}
        \caption{
        $\blacksquare$ $k$-means clustering $\blacksquare$ Synthetic $\blacksquare$ 200 initial, 50 trajectories, 50\% sampling $\blacksquare$}
 \end{figure}
\newpage 

\begin{figure} \begin{center}
    \begin{minipage}{\textwidth}
    \adjustbox{max width=\textwidth}
    {
    %
    }
    \end{minipage}

        \caption{
        $\blacksquare$ $k$-means clustering $\blacksquare$ Synthetic $\blacksquare$ 200 initial, 50 trajectories, 90\% sampling $\blacksquare$}

    \end{center} \end{figure}
\newpage 

\begin{figure} \begin{center}
    \begin{minipage}{\textwidth}
    \adjustbox{max width=\textwidth}
    {
    %
    }
    \end{minipage}

        \caption{$\blacksquare$ $k$-means clustering $\blacksquare$ Synthetic $\blacksquare$ 200 initial, 50 trajectories, 95\% sampling $\blacksquare$}

    \end{center} \end{figure}
\newpage

    \begin{figure} \begin{center}
    \begin{minipage}{\textwidth}
    \adjustbox{max width=\textwidth}
    {
    %
    }
  \end{minipage}
  \end{center}
        \caption{
        $\blacksquare$ $k$-means clustering $\blacksquare$ Synthetic $\blacksquare$ 100 initial, 100 trajectories, 50\% sampling $\blacksquare$}

    \end{figure}

\newpage

    \begin{figure} \begin{center}
    \begin{minipage}{\textwidth}
    \adjustbox{max width=\textwidth}
    {
    %
    }
  \end{minipage}
  \end{center}
        \caption{
        $\blacksquare$ $k$-means clustering $\blacksquare$ Synthetic $\blacksquare$ 100 initial, 100 trajectories, 90\% sampling $\blacksquare$}

    \end{figure}

\newpage

    \begin{figure} \begin{center}
    \begin{minipage}{\textwidth}
    \adjustbox{max width=\textwidth}
    {
    %
    }
  \end{minipage}
  \end{center}
        \caption{
        $\blacksquare$ $k$-means clustering $\blacksquare$ Synthetic $\blacksquare$ 100 initial, 100 trajectories, 95\% sampling $\blacksquare$}

    \end{figure}

\newpage

\end{document}